\documentclass[12pt,authoryear,review]{elsarticle}

\usepackage{geometry}
\usepackage{verbatim}
\usepackage{amsthm}
\usepackage{amsmath}
\usepackage{amssymb}
\usepackage{graphicx}
\usepackage[position=top]{subfig}
 \usepackage{booktabs}
\usepackage{epstopdf}

\usepackage{appendix}

  \theoremstyle{definition}
  \newtheorem{defn}{\protect\definitionname}[section]
  \theoremstyle{plain}
  \newtheorem{prop}{\protect\propositionname}[section]
  \theoremstyle{plain}
  \newtheorem{thm}{\protect\theoremname}[section]
  \theoremstyle{plain}
  
  \theoremstyle{plain}
  \newtheorem{lem}{\protect\lemmaname}[section]

\usepackage{scalefnt}
\usepackage{amsfonts}
\usepackage{setspace}
\usepackage{indentfirst}
\usepackage{endnotes}
\usepackage{rotating}
\usepackage{verbatim}

\def\@biblabel#1{\hspace*{-\labelsep}}

\providecommand{\definitionname}{Definition}
\providecommand{\lemmaname}{Lemma}
\providecommand{\propositionname}{Proposition}
\providecommand{\corollaryname}{Corollary}
\providecommand{\theoremname}{Theorem}
\newcommand{\com}[1]{}
\def\@biblabel#1{\hspace*{-\labelsep}}

\newcommand{\nocontentsline}[3]{}
\newcommand{\tocless}[2]{\bgroup\let\addcontentsline=\nocontentsline#1{#2}\egroup}

\usepackage[unicode=true, bookmarks=false,  breaklinks=false,pdfborder={0 0 1}, colorlinks=false]{hyperref}
\hypersetup{ colorlinks=true, linkcolor = blue, citecolor=blue,urlcolor=blue}
\usepackage{bookmark}

\begin{document}
\begin{frontmatter}{}

\title{Positive feedback in coordination games: Stochastic evolutionary dynamics and the logit choice rule\tnoteref{t1}}

\tnotetext[t1]{\today.
\setstretch{0.8}
The research of S.-H. H. was supported by
the Ministry of Education of the Republic of Korea and the National Research Foundation of Korea (NRF-2019S1A5A8035341). The research of L. R.-B. was supported
by the US National Science Foundation (DMS-1515712). We greatly appreciate comments by the advisory editor and two anonymous referees. Especially, we would like to express special thanks to the late Bill Sandholm who carefully read the early version
(\url{http://arxiv.org/abs/OOO}) of this paper and generously offered us many helpful suggestions.}

\author[A1]{Sung-Ha Hwang\corref{cor1}}

\ead{sungha@kaist.ac.kr}

\author[A2]{Luc Rey-Bellet}

\ead{luc@math.umass.edu }

\cortext[cor1]{Corresponding author.}

\address[A1]{ Korea Advanced Institute of Science and Technology (KAIST), Seoul, Korea}

\address[A2]{Department of Mathematics and Statistics, University of Massachusetts
Amherst, MA, U.S.A.}

\singlespace
\begin{abstract}
We study the problem of stochastic stability for evolutionary dynamics under the logit choice rule.
We consider general classes of coordination games, symmetric or asymmetric, with an arbitrary number of strategies, which satisfies the marginal bandwagon property (i.e., there is positive feedback to coordinate). Our main result is that the most likely evolutionary escape paths from a status quo convention consist of a series of identical mistakes.
As an application of our result, we show that the Nash bargaining solution arises as the long run convention for the evolutionary Nash demand game under the usual logit choice rule. We also obtain a new bargaining solution if the logit choice rule is combined with intentional idiosyncratic plays. The new bargaining solution is more egalitarian than the Nash bargaining solution, demonstrating that intentionality implies equality under the logit choice model.

\end{abstract}
\end{frontmatter}{}

\thispagestyle{empty}

\newpage{}
\setstretch{1.1}

\tocless \section{Introduction\setcounter{page}{1}}\label{sec:intro}

Conventions and customs are sometimes determining factors for formal contracts. For example, \citet{Young01} show that local custom is a driving force in setting up the crop sharing contract terms in the state of Illinois. Customary patterns of behaviors such as asymmetric norms between racial groups and genders can also produce a mechanism by which inequalities persist for a long period of time \citep{Naidu17}. Changes in informal convention sometimes induce formal institutional changes which may contribute to long-run economic growth \citep{HNB2016, Robinson2005}. Thus, understanding both the disruption and emergence of conventions can shed light on problems such as economic incentives, inequality, and long-run growth. Conventions or social norms that typically form over time and last for a long period of time are frequently modeled as long run equilibria in stochastic evolutionary dynamics, in which agents play myopic best responses subject to mistakes, errors, or idiosyncratic plays \citep{Young93, Kandori93, Bowles04}. Therefore identifying the most likely evolutionary paths  escaping from an existing convention or  transitioning between conventions is a key step in studying the disruption and emergence of conventions.

Recently, one of the behavioral rules of myopic agents, the logit choice model, has become popular among researchers because of its analytic convenience (see Section \ref{sec:lit} for existing studies using the logit choice rule). Under the logit choice rule, the probability of an agent's mistake decreases log-linearly in the payoff losses incurred by such a mistake. Recent experimental literature also supports the hypothesis that mistake probabilities decrease in payoff losses (see  Section \ref{sec:lit}). In the widely used uniform mistake model, in which all possible mistakes are equally likely, the more likely path can be easily determined  by comparing the number of mistakes involved and, hence, the lengths of paths (e.g., \citet{Kandori93}). However, under the logit choice rule, determining the most likely path is far from obvious, because the probability of a path depends on the kinds of mistakes involved as well as on the length of the path. For example, \citet{HandN2016} provide an example in which two different kinds of mistake plays are involved in the most likely escape path from a convention under a finite population logit model (see Example 1 in \citet{HandN2016}).


Because of the complexity of the logit choice rule, there has been, so far, no general way to analyze the optimal evolutionary paths from one convention to another. As a concrete example, it is unknown which kind of contract conventions will emerge and persist when agents play the evolutionary version of the familiar Nash bargaining game \citep{Nash53} under the logit rules(see Section \ref{sec:application}). The goal of this paper is to fill in this gap in the literature for general  classes of games. Two recent studies address similar questions. \citet{HandN2016} study two-population coordination evolutionary models with \emph{zero off-diagonal payoffs} and an arbitrary number of strategies, for  both finite and infinite populations. \citet{SS2014} study a one-population coordination evolutionary model with \emph{three strategies} in the infinite population limit. These two  papers are discussed in more detail in  Section \ref{sec:lit} but we note here that their results are limited to specific classes of games (either games with zero off-diagonal payoffs or games with three strategies).


One of the main novelties of this paper is a new method---we call it  ``comparison principles''---which can be applied to  one- or  two-population games with an arbitrary number of strategies and which allows greatly reducing the complexity of finding the most likely evolutionary paths under the logit choice rule.
Specifically, we find that under the logit choice rule, the positive feedback of agents (to coordinate) plays a key role. \citet{Kandori98} introduce the ``marginal bandwagon property'' to capture the positive feedback aspect of network externality, requiring that the advantage of strategy $i$ over $j$ is greater when the other player is playing strategy $i$. In stochastic evolutionary game theory, the (un-)likeliness of a path is measured by a quantity called the ``cost'':  the less costly a path is, the more likely is the transition induced by that path. We show that for \emph{finite} population models with the logit choice rule, \textbf{(i)} positive feedback (defined by the marginal bandwagon property) implies that along the minimum cost escape paths from a status quo convention, agents always deviate first from the status quo convention strategy before deviating from other strategies (Lemma \ref{lem:Pos-feedbacks} (i)), and \textbf{(ii)} the relative strength of the positive feedback effects implies that the transitions from the status quo convention to another convention must  occur consecutively in the cost optimal escape paths (Lemma \ref{lem:Pos-feedbacks} (ii)).

We then apply these results to the exit problem---the problem of finding a cost minimum path escaping from a convention---under finite population models and characterize the candidates for the cost minimizing paths, as follows. The candidates consist of (possibly) different kinds of repeated identical mistakes deviating from the status quo convention strategy (Proposition \ref{prop:red}). Finally, to pin down the exact minimum cost escaping path, we consider the infinite population limit as in \citet{SS2014} (Proposition \ref{prop:main-approx}) and show that the most likely escape paths from the status quo convention involve only \emph{one kind} of repeated identical mistakes of agents (Proposition \ref{prop:main-f-red}, Proposition \ref{prop:main-binding-con}, and Theorem \ref{thm:escape}). This result holds for any coordination games satisfying the marginal bandwagon properties and some regularity conditions with {\em an arbitrary number of strategies}, regardless of symmetric or asymmetric games (hence, one- or two-population models; Theorems \ref{thm:escape} and \ref{thm:main-2p}). To the best of our knowledge, this is a novel result.

As an application of our main results, we study the evolutionary bargaining convention for the Nash demand games under the logit choice rule and show that the Nash bargaining solution arises as the stochastically stable convention under the usual logit choice rule (called the unintentional logit dynamic). We also obtain a new bargaining convention when the logit choice rule is combined with intentional idiosyncratic (non-best response) plays (called the intentional logit dynamic). By \emph{intentional idiosyncratic plays}, we mean that agents always experiment with  strategies under which they would do better, should that strategy induce a convention \citep{Naidu10, HLNN2016}. We show that the new solution under intentional logit dynamics is more egalitarian than the Nash bargaining solution, hence intentionality implies equality (Proposition \ref{prop:intent}). The reason for equality is as follows: under the unintentional logit rule, some transitions from the egalitarian convention (the equal division convention) to the Nash bargaining convention (the unequal division convention) are driven by a population who stands to lose by such transitions. Under the intentional logit rule, every transition is driven by the population who stands to benefit. Thus, some unfavorable transitions leading to the Nash bargaining convention are replaced by favorable transitions to the deviant population, leading to a more equal convention than the Nash bargaining convention.
It can be easily seen that our comparison principle as well as our remaining arguments for the logit choice rule hold for the uniform mistake models under the assumption of the marginal bandwagon property. Thus, our comparison principles provide a unified framework for analyzing evolutionary dynamics, including the uniform mistake and logit models.

%

This paper is organized as follows. Section \ref{sec:lit} discusses the related literature. Section \ref{sec:setup} introduces the basic setup and discusses, in some detail, an example illustrating our methods. We present our main results for the exit problem for one population models in Section \ref{sec:escape}.   In Section \ref{sec:application}, we present results for two population models and analyze the Nash demand game. In the appendix, we show that our result for the exit problem can be used to study the stochastic stability problem. The appendix also provides the technical details and proofs of the paper's results.

\com{and that our comparison method can also be directly applied to the stochastic stability problem for three strategy games (Appendix \ref{sec:direct-sse}). }

\tocless \section{Related Literature \label{sec:lit}}
There are many recent contributions to the analysis of stochastic evolutionary dynamics\footnote{Among them, \citet{Sawa18} study stochastic evolutionary dynamics of loss-averse agents who compares each strategy to a reference point (symmetric 2x2 coordination games) and show that a loss-dominant convention emerges in the long-run (see also \citet{Nax19}). \citet{Bilancini2020} study the stag-hunt game (the symmetric two strategy game) to study the emergence of a convention and transitions between conventions. They introduce condition-dependent mistakes in which errors converge to zero at a rate that is positively related to the payoff earned in the past and show that the payoff-dominant convention emerges when interactions are sufficiently persistent, while the maximin convention can emerge when interactions are volatile. See also \citet{Maruta02}, \citet{Peski10}, \citet{Sandholm10}. See \citet{Newton18} for an extensive survey.
}. Here, we will focus on the following topics which relate most directly to our study: payoff-dependent mistake models, logit choice rules and related experimental evidence, and evolutionary bargaining.

\tocless \subsection{Stochastic stability of payoff-dependent mistake models}
As mentioned earlier, when mistake probabilities depend on payoffs, determining the most likely path seems \emph{a priori} a daunting task, because the probability of a path depends on the kinds of mistakes involved as well as on the length of the path. Indeed, theoretical results for the exit and stochastic stability problems in the literature are limited to symmetric three strategy games in the one population setting or asymmetric two strategy games in the two population setting, except for only a few works.

There are two distinctive approaches in studying these questions. The first one is the so-called small noise double limit approach \citep{Sandholm08, Staudigl2012, SS2014, Srinivas20}.\footnote{
While \citet{Staudigl2012} and \citet{SS2014} study the logit choice rule,  \citet{Srinivas20}  studies the exit problem of the symmetric three strategy coordination game under the probit choice rule and identifies conditions under which the solution to the exit problem under probit choice is qualitatively similar to the logit choice. For the probit models, see \citet{Myatt03}, \citet{Dokumaci11}.}
 This approach takes a zero error rate limit first, as in the standard literature, and then taking an infinite population limit. Via these double limits, they obtain an optimal control problem over all escaping paths.  Then, to find solutions to the optimal control problem, they solve the Hamilton-Jacobi equation associated with the obtained optimal control problem. This method, while providing a systematic approach, requires solving a (possibly challenging)  Hamilton-Jacobi equation associated with the optimal control problem; hence, the results are limited to symmetric games with three strategies (see the discussion section in \citet{SS2014}; \citet{Srinivas20}) or asymmetric games with two strategies \citep{Staudigl2012}.

The second approach is to analyze minimum cost paths for the finite population model obtained by a zero-error limit and reduce the complex finite population problem into a lower dimension problem as in \citet{HandN2016}. Our first step (Lemma \ref{lem:Pos-feedbacks} \textbf{(i)}, Proposition \ref{prop:red} \textbf{(i)}) of comparing paths to show that agents switch first from the status quo convention strategy generalize the approach in  \citet{HandN2016}. Similarly to the current study, \citet{HandN2016} exploit, though implicitly, cost comparison arguments by constructing a lower bound function (see Section 5 in the cited paper). However, the arguments presented by \citet{HandN2016} differ from the current ones, as follows. First, in \citet{HandN2016}, cost estimations of the constructed lower bound functions are possible only because all off-diagonal payoffs are zeros and because the cost functions are (multi-) linear with respect to the populations' states. Second, the lower bound function constructed in \citet{HandN2016} does not correspond to an \emph{actual} path and, thus, the method does not provide guidance for how to construct a similar lower bound function for games other than those with zero off-diagonal payoffs. Their arguments thus cannot be applied to games with nonzero off-diagonal payoffs (e.g., the Nash demand game) or single population models in which the cost function of a path is quadratic with respect to population states. In contrast, under the condition of the \emph{marginal bandwagon property}, our comparison methods are used to construct a lower cost path for a given arbitrary path in the state space (i.e., the simplex) and our arguments can be applied to any coordination games with an arbitrary number of strategies satisfying the marginal bandwagon property.

%
%
%
%
%
%
%

\tocless \subsection{The logit choice rule and experimental evidence for payoff dependent mistake models}
The logit choice rule, introduced by \citet{Blume93} to evolutionary game theory, has been widely used in stochastic evolutionary dynamics. Among them, \citet{Young01}, mentioned in the introduction, adopted the logit choice model to study the contractual custom of cropsharing. \citet{KY1} show that fast convergence can occur when the error rate is small but non-vanishing even in a large population under logit dynamics. \citet{BandB13} also use the logit choice rule to study the effect of endogenous preferences and institutions on trade liberalization\footnote{Also, \citet{AlosNetzer10}, and \citet{Okada12} studied problems related to various revision rules and local potentials, respectively, under the logit choice rule.}.

The logic choice rule also receives special attention from the literature on random utility models and stochastic choices \citep{McKelvey95}. \citet{Hofbauer02} derive the logit choice rule in two different ways: one from the random utility model and another from the optimization of perturbed expected utility. Recently, \citet{Fudenberg15} provide two easily understood axioms  under which stochastic choice corresponds to the maximization of the perturbed expected utility, and when the perturbation cost function is the commonly used one (namely, an entropy function), the optimal stochastic choice rule becomes the logit choice rule. Relatedly, in the literature on rational inattention or information acquisition, \citet{Matejka15} show that the decision maker’s optimal information-processing strategy results in probabilistic choices that follow a logit model, where the parameter of perturbation is interpreted as the cost of information.

Experimental evidence for the class of payoff dependent mistake models to which the logit choice rule belongs  is as follows. \citet{MasNax2016} find in their experiment that a payoff decrease in the previous period would induce higher deviation rates from myopic best response behavior, which indicates that subjects’ choices are sensitive to past payoff losses. \citet{LimNeary2016} find that individual mistakes depend on the payoff of the myopic best-response payoff. \citet{HLNN2016} also provide experimental evidence suggesting non-best response play depends on payoffs: higher rates of non-best response play from subjects for whom the expected payoff from the best response strategy is lower (in Fig. 5 of the cited paper). Moreover, they conducted a logistic regression on payoff differences, providing further support for payoff dependence mistake models.


%
%
%
%
%

\tocless \subsection{Evolutionary bargaining}
This paper also adds to the literature on evolutionary bargaining \citep{Young93JET, Young98Res, BSY1, Naidu10, HLNN2016}. In particular, \citet{Young93} shows that the Nash bargaining solution emerges under the unintentional uniform deviation model of the Nash demand game and similarly, \citet{Naidu10} find that the Nash bargaining solution emerges as well under the intentional uniform deviation model. In this paper, we find  that  the Nash bargaining solution is again stochastically stable under the unintentional logit choice rule, while a new bargaining convention which is more egalitarian than the Nash bargaining solution arises under the intentional logit choice rule, as explained. This shows quite well how evolutionary bargaining approaches can complement and extend the existing axiomatic bargaining approaches, hence contribute to the understanding of a  certain bargaining convention in a society.

\tocless \section{Stochastic Evolutionary Dynamics: Setup and Example \label{sec:setup}}

\tocless \subsection{Basic setup: one population model}

Consider a population of $n$ agents who play a symmetric coordination
game with strategy set $S=\{1,2,\cdots,$ $|S|\}$ and payoff matrix $A$. The population state is described as a vector of fractions of agents using each strategy; that is, the state of the population is $x\in\Delta^{(n)}$, where $\Delta^{(n)}$ is the simplex
\[
    \Delta^{(n)}:=\{(x_{1},\cdots,x_{|S|}) \in \frac{1}{n} \mathbb{Z}^{|S|}:\,\sum_{i}x(i)=1,\, x(i) \ge0\,\,\textrm{for all }\}.
\]
The expected payoff  to an agent who chooses  strategy $i$ at population
state $x$ is given by $\pi(i,x):=\sum_{j \in S}A_{ij} x(j)$.

We consider a discrete time strategy updating process, defined as follows. At each period, a randomly chosen agent selects a new strategy. The new population state induced by the agent's switching from strategy $i$ to $j$ is denoted by $x^{i,j}$ and the state induced by two agents' transitions, first from $i$ to $j$ and then from $k$ to $l$, is denoted by $x^{(i,j)(k,l)}$. More precisely,
\begin{equation}\label{eq:x-ij}
  x^{i,j}:=x +\frac{1}{n}(e_{j}-e_{i}),\,\,\text{and } x^{(i,j)(k,l)}:=x +\frac{1}{n}(e_{j}-e_{i}) + \frac{1}{e}(e_l - e_k),
\end{equation}
where $e_{h}$ is the $h$-th element of the standard basis for $\mathbb{R}^{|S|}$. The conditional probability that an agent with strategy $i$ chooses new strategy $j$ given population state $x$ is specified by the logit choice rule \citep{Blume93},
\begin{equation}
\text{Logit choice rule: }p^{(n)}_{\eta}(j|i,x)=\frac{\exp(\eta^{-1}\pi(j,x))}{\sum_{l}\exp(\eta^{-1}\pi(l,x))},
\label{eq:logit}
\end{equation}
where $\eta > 0$ is a positive parameter interpreted as the degree of (ir)rationality (or noise level). That is, as $\eta$ decreases to $0$, equation \eqref{eq:logit} converges to the so-called best-response rule, whereas as $\eta$ increases to $\infty$, equation \eqref{eq:logit} converges to a choice rule that assigns equal probabilities to each strategy---namely, a pure randomization rule. \com{For finite $\eta$, the probability of choosing strategy $j$ increases because the agent expects a higher payoff from strategy $j$ than she does from other  strategies.}

The transition probabilities for the updating dynamics are
\[
P^{(n)}_\eta(x,x^{i,j})=x(i) p^{(n)}_\eta(j|i,x)
\]
for $i\neq j$, where  factor $x{(i)}$ accounts for the fact that at each period, one agent is randomly chosen
to revise her strategy. The unlikeliness of a transition in stochastic evolutionary game theory is measured by the cost, $c^{(n)}(x,y)$, between two states, $x, y$:
\[
c^{(n)}(x,y):=\begin{cases}
-\lim_{\eta \rightarrow 0} \eta \ln P^{(n)}_{\eta}(x,y) & \text{ if } y=x^{i,j}\,\text{for some}\, i,j,\, i\neq j\\
\text{0}  & \text{ if } y=x \\
\infty & \,\,\text{otherwise}
\end{cases}
\]
which becomes
\begin{align}
c^{(n)}(x,x^{i,j}) & =\max\{\pi(l,x):l\in S\}-\pi(j,x) \label{eq: l-cost}
\end{align}
under the logit choice rule \eqref{eq:logit}. When $\eta$ is sufficiently small, we have $P^{(n)}_\eta(x,y) \asymp e^{-\eta^{-1} c^{(n)}(x,y)}$ and thus the cost between $x$ and $y$, $c^{(n)}(x,y)$,  is  the exponential rate of decay of the probability of transition from $x$ to $y$. In equation (\ref{eq: l-cost}), the first term $\max\{\pi(l,x):l\in S\}$ is equal to the payoff  $\pi({\bar m},x)$  for agents playing a best response $\bar{m}$ in the population state $x$, and the second term is the payoff to agents playing the new strategy, $j$. Thus, when a strategy-revising agent adopts the best response $\bar{m}$, the cost of such an action is zero. However, when she adopts a sub-optimal strategy $j \neq \bar m$, the cost is the payoff loss due to choosing this strategy instead of the best response. Under the uniform mistake model, $c^{(n)}(x,x^{i,j}) = 1$ if $j \neq \bar m$ and  $c^{(n)}(x,x^{i,j}) = 0$ if $j =\bar m$;  that is the cost in the uniform mistake model is state independent.

We consider a symmetric coordination game $A$ in which every symmetric strategy profile (i.e. the strategy profile in which row and column players choose the same strategy)  is a strict Nash equilibrium (see Condition A).
A convention is defined as a state in which every agent plays the same strategy which is a strict Nash equilibrium of $A$. Thus, $x$ is a convention if $x= e_{i}$ for some strategy $i$, which is a strict Nash  equilibrium strategy (recall that $e_i$ is the $i$-th element of the standard basis of $\mathbb{R}^{|S|}$). Focusing our attention on one such convention, we refer to convention $\bar m$ as a status quo convention.
A path $\gamma$ is a sequence of states, $\gamma:=(x_{1},x_{2},\cdots,x_{T}),$ such
that $x_{t+1}=(x_{t})^{i,j}$ for some $i,j$ and for all $t$, and we define the cost of a path as the sum of the costs of the  transitions between the states in equation (\ref{eq: l-cost}), i.e.
\begin{equation}
I^{(n)}(\gamma):=\sum_{t=1}^{T-1}c^{(n)}(x_{t},x_{t+1}).
\label{eq: cost-path}
\end{equation}

Next, recall {\em the marginal bandwagon property} (\textbf{MBP}) introduced by  \citet{Kandori98}. A symmetric game with payoff matrix $A$ satisfies the \textbf{MBP} if
\begin{equation}
A_{ii}-A_{ji}>A_{ik}-A_{jk}\textrm{\,\,\ for\,\ all\,\ distinct\,}\, i,j,k.\label{eq:MBP}
\end{equation}
\noindent The condition in \eqref{eq:MBP} says that the advantage of playing strategy $i$ over strategy $j$ is greater when the other player plays strategy $i$ rather than another strategy $k$. In our population dynamic model, this implies a positive feedback effect in which the marginal advantage of switching into strategy $i$ increases in the number of agents adopting strategy $i$. We also consider coordination games in which $A_{ii} > A_{ji}$ for all $i, j$, and  assume the existence of  mixed-strategy Nash equilibria supported on any arbitrary subset of the strategy set $S$.

\bigskip
\noindent
\textbf{Condition A}:
(i) A game with payoff matrix $A$ is a coordination game (i.e., $A_{ii} > A_{ji}$ for all $i, j$),  satisfying  the \textbf{MBP} (see equation \eqref{eq:MBP}), \\
(ii) Suppose that
for any $T = \{i_1, \cdots, i_K \} \subset S$, there exists a unique $q \in \Delta$ with support $T$ such that
\[
    \pi(i_1, q) = \cdots= \pi(i_K, q)
\]
and $q$ is a mixed strategy Nash equilibrium. $\square$
\bigskip

We define the basin of attraction of $e_{\bar{m}}$, $D^{(n)}(e_{\bar{m}})$, and its boundary, $\partial D_{\bar m, j}^{(n)}$, as follows:
\begin{align*}
D^{(n)}(e_{\bar{m}}): & =\{x\in\Delta^{(n)}:\pi(\bar{m},x)\geq\pi(k,x)\text{\,\ for all }\, k\,\}\\
\partial D_{\bar m, j}^{(n)} : & = \{ x \in D^{(n)}(e_{\bar{m}}): x^{\bar m, j} \in D^{(n)}(e_j) \}
\end{align*}
In fact, if $x$ belongs to $D^{(n)}(e_{\bar{m}})$, the cost of a transition from $i$ to $j$ is
\begin{equation}
    c^{(n)}(x,x^{i,j})=\pi(\bar{m},x)-\pi(j,x)\label{eq:cost-boa}
\end{equation}
for $i \neq j$. If $j=\bar m$, the cost in equation \eqref{eq:cost-boa} is zero, and the \textbf{MBP} implies that convention $\bar m$ can be reached from any $x \in D^{(n)}(e_{\bar m})$ at no cost; thus, $D^{(n)}(e_{\bar m})$ is indeed the basin of attraction  of convention $\bar m$. Observe that the requirement for the existence of mixed-strategy Nash equilibria with arbitrary support implies the existence of the distinctive basins of attraction for all pure strategies. Using this setup, we next present a simple example of a three-strategy game to illustrate the main ideas and results of the paper.

\tocless \subsection{Illustration of the main results\label{subsec:example}}

%
%
%

\begin{figure}
\centering
\includegraphics[scale=0.65]{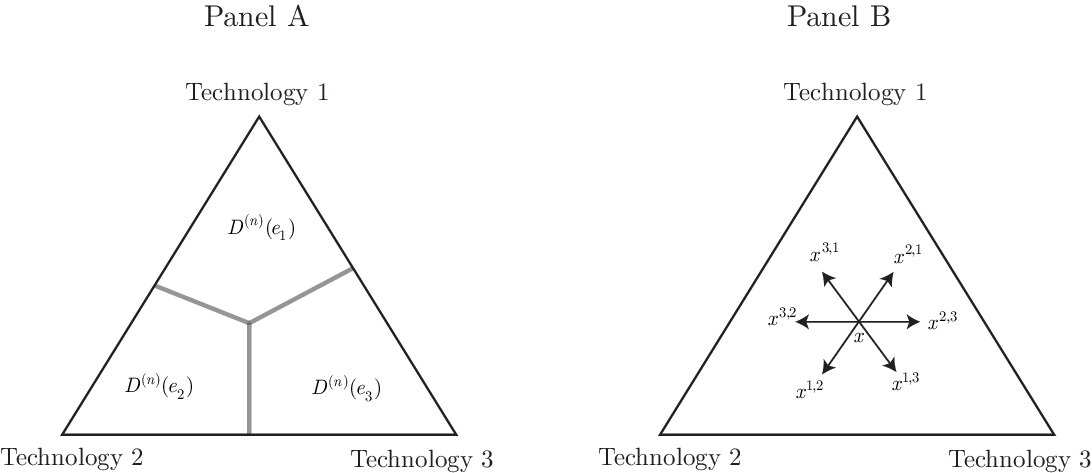}

\caption{\textbf{ Basins of attraction and paths involving a single agent's transitions}. Panel A shows the basins of attractions of conventions. Panel B illustrates new states induced by a single agent's switching from state $x$. For example, $x^{1,2}$ is the new state induced by a single agent's switching from strategy 1 to strategy 2. }

\textbf{\label{fig:bofa}}
\end{figure}

Consider a technology choice game consisting of three technologies indexed by 1, 2, and 3, respectively. For example, for PC operating systems, consider Windows, OSX, and Linux. Let $b_{i}$ be the benefit of technology $i$ that the user obtains when interacting with another user of the same technology; thus, $b_{i}$ is related to the inherent quality of technology $i$. Suppose that the user of technology $i$ experiences some utility or disutility when interacting with users of a different technology. For the sake of simplicity, the users of technologies $2,3$, and $1$ derive utility $d$ when interacting with the users of technologies $1,2$, and $3$, respectively. By the same token, the users of technologies $1,2$, and $3$ experience disutility $d$ when interacting with the users of technologies $2,3$, and $1$, respectively. In summary, the payoff matrix is given by
\begin{equation}
A=\begin{pmatrix}b_{1} & -d & d\\
d & b_{2} & -d\\
-d & d & b_{3}
\end{pmatrix}.\label{eq:ex-game}
\end{equation}

\noindent In the context of the technology choice game, we are interested in the positive feedback effects, where the advantages of a technology increase as the number of users of that technology increases. If
\begin{equation}
3 d<\min_{i} b_{i} \label{eq:con1}
\end{equation}
holds, all pure strategies 1, 2, and 3 are strict Nash equilibria, and the \textbf{MBP} condition in (\ref{eq:MBP}) is satisfied.

Suppose that the agents' strategy revision rule is the logit choice
rule. Our example is the exit problem from a convention (see \citet{FW98}). Specifically, given the status quo convention of technology $1$, what is the most likely way to upset this convention? Panel A of Figure \ref{fig:bofa} shows the basins of attraction of conventions. Thus, our problem can be stated succinctly as
\begin{equation}
\min\{I^{(n)}(\gamma):\gamma\textrm{\,\ escapes }D^{(n)}(e_{1})\}. \label{eq:min-1st}
\end{equation}
We now explain how the \textbf{MBP} in \eqref{eq:MBP} significantly reduces the complexity of solving the minimization problem in  \eqref{eq:min-1st}.


\subsection*{Comparison principle 1: Lemma \ref{lem:Pos-feedbacks} (i), Proposition \ref{prop:red} (i)}
Our idea is to develop systematic ways of comparing the costs of various paths and to reduce the number of candidate solutions for the minimization problem in \eqref{eq:min-1st} under the logit choice rule. First, consider the two paths in Panel A in Figure \ref{fig:comparison-1}:
\begin{align}
    & x\rightarrow x^{2,3}\rightarrow x^{(2,3)(1,3)} \label{eq:comp1-1st}\\
    & x\rightarrow x^{1,3}\rightarrow x^{(1,3)(2,3)} \label{eq:comp1-2nd},
\end{align}
where $x^{(2,3)(1,3)} = x^{(1,3)(2,3)}$(see the definitions of $x^{i,j}$ and $x^{(i,j)(k,l)}$ in equation \eqref{eq:x-ij}). The cost difference between these two paths is
\begin{equation}
\Delta_1 c:=c^{(n)}(x,x^{2,3})+c^{(n)}(x^{2,3},x^{(2,3)(1,3)})-[c^{(n)}(x,x^{1,3})+c^{(n)}(x^{1,3},x^{(1,3)(2,3)})].
\label{eq:comp1}
\end{equation}

\begin{figure}
\centering
\includegraphics[scale=0.7]{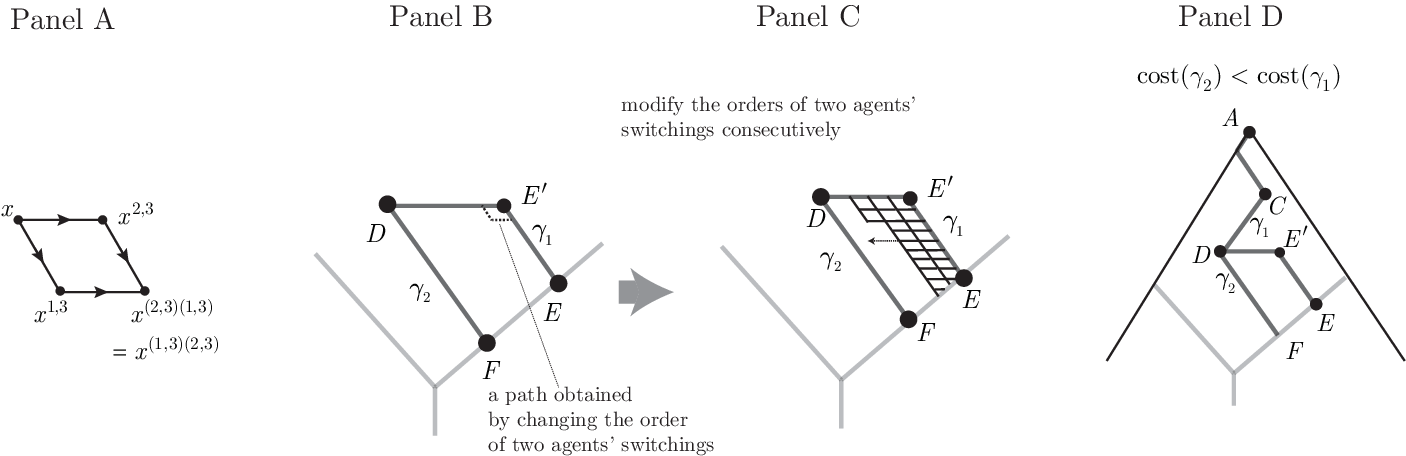}

\caption{\textbf{ Comparison Principle 1}. Panel A shows two paths in equations \eqref{eq:comp1-1st} and \eqref{eq:comp1-2nd}. Panels B and C show how we can obtain a lower cost path by modifying the orders of two agents' switching from one strategy to another. In Panel D, we show that the cost of $\gamma_2$ in \eqref{eq:gamma_2} is lower than that of $\gamma_1$ in \eqref{eq:gamma_1} by applying this procedure repeatedly.}

\textbf{\label{fig:comparison-1}}
\end{figure}

Note that under the logit choice rule, an agent switching from
strategy $2$ to strategy $3$ compares the expected payoff of strategy $3$ with the best response (strategy $1$) rather than with strategy 2, implying that the costs of transitions from strategy $2$ to strategy $3$ and from strategy $1$ to strategy $3$ are the same. Indeed using equation (\ref{eq: l-cost}) we verify that
\begin{equation}
c^{(n)}(x,x^{2,3}) = c^{(n)}(x,x^{1,3}), \, \, \,  c^{(n)}(x^{1,3}, x^{(1,3),(2,3)})= c^{(n)}(x^{1,3}, x^{(1,3),(1,3)})
\label{eq:cost-1}
\end{equation}
Thus, using equation \eqref{eq:cost-1}, we can simplify \eqref{eq:comp1} as follows:
\[
	\Delta_1 c = c^{(n)}(x^{2,3},x^{(2,3)(1,3)}) -c^{(n)}(x^{1,3}, x^{(1,3),(1,3)})
\]
Note that $x^{2,3}$ is the state where exactly one more agent than in $x^{1,3}$ plays strategy 1, because $x^{2,3} - x^{1,3} = \frac{1}{n}(e_1 - e_2)$. The positive feedback effect of strategy 1 over strategy 3, induced by the \textbf{MBP}, means that the payoff advantage of strategy 1 over strategy 3 is greater when the other player uses strategy 1. Thus, the marginal advantage of switching to strategy 1 is greater at state ($x^{2,3}$) where one more agent than in the other state ($x^{1,3}$) plays strategy 1. This means that the payoff loss due to the mistake of not playing strategy 1 is greater at $x^{2,3}$ than it is at $x^{1,3}$. Thus, we expect that under the logit dynamic, the transition from strategy 1 to strategy 3 will be more costly at $x^{2,3}$ than it will be at $x^{1,3}$. Indeed, we find that
\begin{align}
	\Delta_1 c = c^{(n)}(x^{2,3},x^{(2,3)(1,3)})-c^{(n)}(x^{1,3},x^{(1,3)(1,3)}) & =\pi(1,x^{2,3})-\pi(3,x^{2,3})-[\pi(1,x^{1,3})-\pi(3,x^{1,3})]\nonumber \\
 & =\frac{1}{n}(A_{11}-A_{31}-(A_{12}-A_{32})) >0 \label{eq:cost-2}
\end{align}
which is positive from the marginal bandwagon property (condition \eqref{eq:MBP}). Equation \eqref{eq:cost-2} also shows that when agents make the same mistake (switching from 1 to 3), the cost becomes cheaper ($c^{(n)}(x^{1,3},x^{(1,3)(1,3)})< c^{(n)}(x^{2,3},x^{(2,3)(1,3)})$). In sum, under the assumption of the marginal bandwagon property, the cost of path
$x\rightarrow x^{1,3}\rightarrow x^{(1,3)(2,3)}$ is cheaper than
that of $x\rightarrow x^{2,3}\rightarrow x^{(2,3)(1,3)}$ (see Lemma \ref{lem:Pos-feedbacks} (i)).

Now, consider the new path (shown as a dotted line) obtained by altering
a single agent's switching in Panel B of Figure \ref{fig:comparison-1}.
The cost difference between the original and new paths in
Panel B of Figure \ref{fig:comparison-1} is precisely the cost difference
between the two paths in Panel A. Thus, if equation (\ref{eq:cost-2}) is positive, the cost of
the new path in Panel B is strictly lower than that of the original path. Then, by successively altering a single agent's switching, we can apply
the same arguments repeatedly as in Panel C of Figure
\ref{fig:comparison-1} (Proposition \ref{prop:red}, (i)). In this way, we find that the cost of path $D\rightarrow F$ is cheaper than that of path $D\rightarrow E'\rightarrow E$ and finally the cost of path $\gamma_2$ is cheaper than that of path $\gamma_1$ (Panel D in Figure \ref{fig:comparison-1}), where
\begin{align}
  \gamma_1 : \,\,\, & A \rightarrow C \rightarrow D \rightarrow E' \rightarrow E  \label{eq:gamma_1} \\
  \gamma_2 : \,\,\, & A \rightarrow C \rightarrow D \rightarrow F \label{eq:gamma_2}.
\end{align}

\subsection*{Comparison principle 2: Lemma \ref{lem:Pos-feedbacks} (ii), Proposition \ref{prop:red} (ii) }


\begin{figure}
\centering
\includegraphics[scale=0.7]{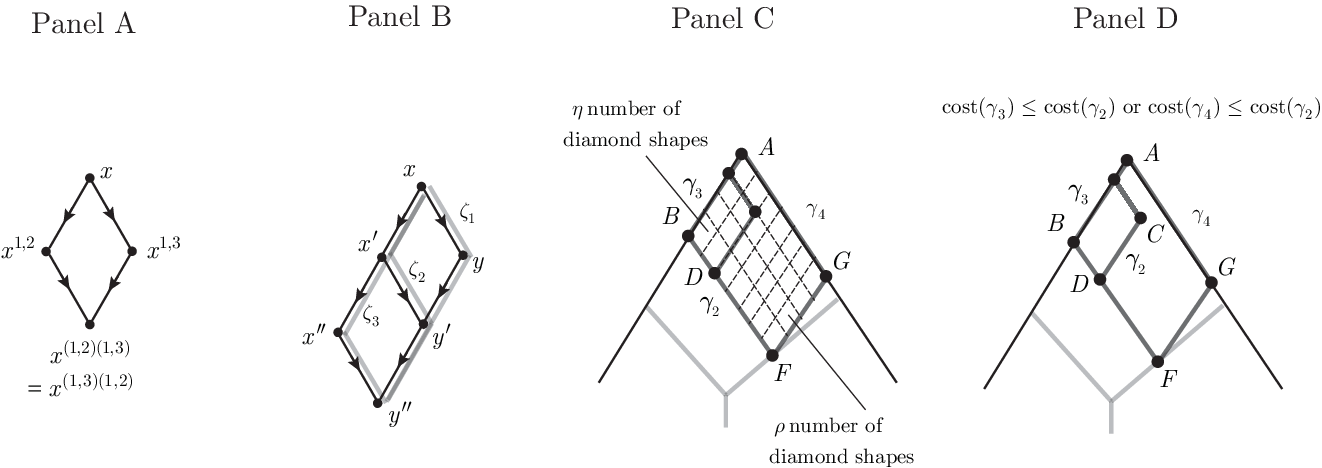}

\caption{\textbf{ Comparison Principle 2}. Panel A shows two paths in \eqref{eq:comp2-1st} and \eqref{eq:comp2-2nd}. Panel B illustrates equations in \eqref{eq:comp-3-diff}, \eqref{eq:sum}, and \eqref{eq:alt}, showing that the cost of $\zeta_2$ is no less than that of $\zeta_1$ or $\zeta_3$. Panel C illustrates the argument in equations \eqref{eq:rho-comp}, \eqref{eq:rho-eta-comp}, and \eqref{eq:alt2}. Thus, in Panel D, we obtain that the cost of $\gamma_2$ is again no less than that of $\gamma_3$ or $\gamma_4$. }

\textbf{\label{fig:comparison-2}}
\end{figure}

Next, we explain the second comparison principle. Similarly, we first consider the two paths in Panel A of Figure \ref{fig:comparison-2}:
\begin{align}
   & x\rightarrow x^{1,3}\rightarrow x^{(1,3)(1,2)} \label{eq:comp2-1st} \\
   & x\rightarrow x^{1,2}\rightarrow x^{(1,2)(1,3)} \label{eq:comp2-2nd}
\end{align}
where $x^{(1,3)(1,2)}=x^{(1,2)(1,3)}$. We find
\begin{align}
 	\Delta_2 c := & c^{(n)}(x,x^{1,3})+c^{(n)}(x^{1,3},x^{(1,3)(1,2)})-[c^{(n)}(x,x^{1,2})+c^{(n)}(x^{1,2},x^{(1,2)(1,3)})] \nonumber \\
 =	&\underbrace{[ c^{(n)}(x, x^{1,3}) - c^{(n)}(x^{1,2}, x^{(1,2)(1,3)}) ]}_{\textrm{(i) positive feedback of 1 over 3}} - \underbrace{[ c^{(n)}(x, x^{1,2}) - c^{(n)}(x^{1,3}, x^{(1,3)(1,2)}) ]}_{\textrm{(ii) positive feedback of 1 over 2}}.  \label{eq:comp-2}
\end{align}
If we let $x=y^{2,3}$, then $x^{1,2}=y^{1,3}$ and, as in equation \eqref{eq:cost-2}, (i) in equation \eqref{eq:comp-2} becomes
\begin{align*}
	c^{(n)}(x, x^{1,3}) - c^{(n)}(x^{1,2}, x^{(1,2)(1,3)}) & = c^{(n)}(y^{2,3}, y^{(2,3)(1,3)})-c^{(n)}(y^{1,3}, y^{(1,3)(1,3)}) \\
	&=\frac{1}{n}( A_{11} - A_{31} - (A_{12} - A_{32})).
\end{align*}
Furthermore, if we let $x = z^{3, 2}$, then $x^{1,3} = z^{1,2}$ and (ii) in equation \eqref{eq:comp-2} becomes
\begin{align*}
	c^{(n)}(x, x^{1,2}) - c^{(n)}(x^{1,3}, x^{(1,3)(1,2)}) & = c^{(n)}(z^{3,2}, z^{(3,2)(1,2)})-c^{(n)}(z^{1,2}, z^{(1,2)(1,2)}) \\
	&= \frac{1}{n}( A_{11} - A_{21} - (A_{13} - A_{23}))\,,
\end{align*}
which is the positive feedback effect of strategy 1 over 2. Thus,
\begin{equation} \label{eq:MS-con}
 	\Delta_2 c=\frac{1}{n}( -A_{12}+A_{13}+A_{21}-A_{23}-A_{31}+A_{32})=\frac{6d}{n}.
\end{equation}
Here, $6d$ in equation \eqref{eq:MS-con} can be positive, negative, or zero. When $6d =0$, the game is a potential game---this is the well-known test for potential games by \citet{Hofbauer85}. \citet{SS2014} also define \eqref{eq:MS-con} by ``skew'' and use it to compare the costs of paths in the infinite population model.

Next, using this result, we compare the three paths in Panel B of Figure \ref{fig:comparison-2}, defined as follows:
\begin{align*}\label{eq:comp-3}
  \zeta_1:  x \rightarrow y \rightarrow y' \rightarrow y'', \,\,
  \zeta_2:  x \rightarrow x' \rightarrow y' \rightarrow y'',\,\,
  \zeta_3:  x \rightarrow x' \rightarrow x'' \rightarrow y''
\end{align*}
\noindent Then, applying equation \eqref{eq:MS-con}, we find that
\begin{equation}\label{eq:comp-3-diff}
  I^{(n)}(\zeta_1) - I^{(n)}(\zeta_2) = \frac{6d}{n} \text{ and } I^{(n)}(\zeta_2) - I^{(n)}(\zeta_3) =  \frac{6d}{n}
\end{equation}
which shows that
\begin{equation}\label{eq:sum}
   [I^{(n)}(\zeta_2) - I^{(n)}(\zeta_1)] + [I^{(n)}(\zeta_2)-I^{(n)}(\zeta_3)]=0.
\end{equation}
This, in turn, implies that either
\begin{equation}\label{eq:alt}
  I^{(n)}(\zeta_2) \geq  I^{(n)} (\zeta_1) \text{ or } I^{(n)}(\zeta_2) \geq I^{(n)} (\zeta_3)
\end{equation}
holds. Thus, from the inequalities in \eqref{eq:alt}, either $\zeta_1$ or $\zeta_3$ costs less than (or is equal to) $\zeta_2$ and we can remove $\zeta_2$ from the candidate paths minimizing the problem in equation \eqref{eq:min-1st}.

Next, we compare the costs of the three paths, $\gamma_2: A \rightarrow C \rightarrow D  \rightarrow F$ in \eqref{eq:gamma_2}, $\gamma_3$, and $\gamma_4$ (Panels C and D of Figure \ref{fig:comparison-2}), where
\begin{align}
  \gamma_3 : & \,\,\, A \rightarrow B \rightarrow F \label{eq:gamma_3} \\
  \gamma_4 : & \,\,\, A \rightarrow G \rightarrow F \label{eq:gamma_4}.
\end{align}
For the purpose of exposition, assume that there are $\rho$ diamond shapes in the area between $\gamma_4$ and $\gamma_2$, and $\eta$ diamond shapes between $\gamma_2$ and $\gamma_3$ (see Panel C of Figure \ref{fig:comparison-2}). Now, by applying the comparison results in equation \eqref{eq:comp-3-diff} successively, we find that
\begin{equation}\label{eq:rho-comp}
  I^{(n)}(\gamma_4) - I^{(n)}(\gamma_2) = \rho \frac{6d}{n}, \qquad I^{(n)}(\gamma_2) - I^{(n)}(\gamma_3) = \eta  \frac{6d}{n}.
\end{equation}
which yields
\begin{equation}\label{eq:rho-eta-comp}
  \eta [I^{(n)}(\gamma_2) - I^{(n)}(\gamma_4)] + \rho [I^{(n)}(\gamma_2) - I^{(n)}(\gamma_3)] =0.
\end{equation}
This, in turn, implies that either
\begin{equation}\label{eq:alt2}
  I^{(n)}(\gamma_2) \geq  I^{(n)}(\gamma_4) \text{ or  } I^{(n)}(\gamma_2) \geq  I^{(n)} (\gamma_3)
\end{equation}
holds. Using this, we can also remove $\gamma_2$ from the minimum cost candidate paths (see Panel D of Figure \ref{fig:comparison-2}).

\begin{figure}
\centering
\includegraphics[scale=0.7]{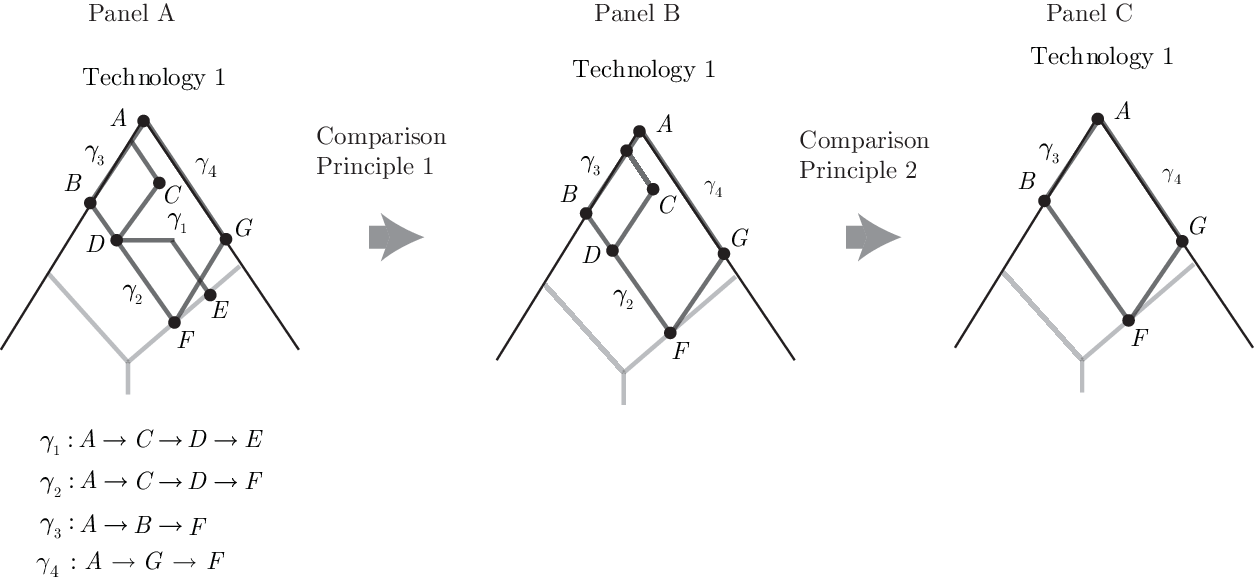}

\caption{\textbf{Implications of Comparison Principles 1 and 2}. Panel A shows four different pathes, $\gamma_1, \gamma_2, \gamma_3$ and $\gamma_4$. By applying comparison principle 1, we eliminate $\gamma_1$ from the set of candidate solution to the exit problem in \eqref{eq:min-1st} as in Panel B. Then by applying comparison principle 2, we eliminate $\gamma_2$ as in Panel C. }

\textbf{\label{fig:comp-combined}}
\end{figure}

Finally, we apply these two comparison principles, to obtain a class of paths comprising the candidate solutions to the cost minimization problem in equation \eqref{eq:min-1st} (see Figure \ref{fig:comp-combined}). One of these paths consists of consecutive transitions first from technology 1 to technology 2 and then from technology 1 to technology 3 (see $\gamma_3$ in Panel C in Figure \ref{fig:comp-combined}). Alternatively, there could be another path consisting of consecutive transitions first from technology 1 to technology 3 and then from technology 1 to technology 2 (see $\gamma_4$ in  Panel C in Figure \ref{fig:comp-combined}). Thus, we can reduce the complicated objective function in equation (\ref{eq:min-1st}) to a function of two variables (i.e., the number of transitions from technology 1 to technology 2, and those from technology 1 to technology 3) and easily study the minimization problem of this simple objective function using the \textbf{MBP} again. In general, we reduce the objective function with an arbitrary number of variables in equation \eqref{eq:min-1st} to an objective function with $|S|-1$ variables, where $|S|$ is the number of strategies of the underlying game. We then prove that the lowest cost transition path to escape convention 1 involves the repetition of the same kinds of mistakes in the infinite population limit. That is, graphically, these paths lie on the edges of the simplex from strategy 1 to strategy 2 and from strategy 1 to strategy 3.

\tocless \section{Exit from the basin of attraction of a convention \label{sec:escape}: one population models}



We now present our comparison principles for games with an arbitrary number of strategies for a finite population. Our first comparison principle shows that given two paths, $x \rightarrow x^{(i,k)} \rightarrow x^{(i,k)(\bar m, l)}$ and $x \rightarrow x^{(\bar m,k)} \rightarrow x^{(i,k)(\bar m, l)}$, it  always costs less (or the same) to first switch away from strategy $\bar{m}$ and then to switch away from the other strategies, as already explained in Section \ref{sec:setup} (Lemma \ref{lem:Pos-feedbacks}). Our second comparison principle is based on the fact that the sizes of the two different positive feedback effects (\textbf{MBP}) can be compared as follows (see equation \eqref{eq:MS-con}):
\begin{align}\label{eq:rel-str}
    [A_{ii} - A_{ji} - A_{ik} + A_{jk}]-[A_{ii} - A_{ki} - A_{ij} + A_{kj}]
  =  [A_{ij}-A_{ji}]+[A_{jk}-A_{kj}]+[A_{ki}-A_{ik}]
\end{align}
\

\begin{lem}\label{lem:Pos-feedbacks} The following statements hold.  \\
(i)\textbf{(Comparison principle 1).} Suppose that the \textbf{MBP} holds. Consider two paths $\gamma_{1}$ and $\gamma_{2}$ (Panel A, Figure \ref{fig:comparison-1}) in $D^{(n)}(e_{\bar{m}})$:
\begin{align*}
\gamma_{1}:  \, x\rightarrow x^{\bar{m},k}\rightarrow x^{(\bar{m},k)(i,l)}, \quad \gamma_{2}:  \, x\rightarrow x^{i,k}\rightarrow x^{(i,k)(\bar{m},l)},
\end{align*}
where $i \neq \bar m, k \not= i, \bar m $, $l \not= i, \bar m$.
Then,
\[
I^{(n)}(\gamma_{2}) \geq I^{(n)}(\gamma_{1}) \,.
\]
(ii) \textbf{(Comparison principle 2).} Consider three paths
$\zeta_{1}$, $\zeta_{2}$ and , $\zeta_{3}$ (Panel B, Figure \ref{fig:comparison-2}) in $D(e_{\bar{m}})$:
\begin{align*}
\zeta_{1}: & \, x \rightarrow x^{\bar{m},j}\rightarrow  x^{(\bar{m},j)(\bar m, i)} \rightarrow x^{(\bar{m},j)(\bar m, i)(\bar m, i)} \\
\zeta_{2}: & \, x\rightarrow x^{\bar{m}, i}\rightarrow x^{(\bar{m},i)(\bar{m},j)} \rightarrow x^{(\bar{m},i)(\bar{m},j)(\bar m, i)} \\
\zeta_{3}: & \, x \rightarrow x^{\bar m, i} \rightarrow x^{(\bar m, i)(\bar m, i)} \rightarrow x^{(\bar m, i)(\bar m, i)  (\bar m, j )}
\end{align*}
where $i \neq j$. Then, we have either
\[
I^{(n)}(\zeta_2) \geq I^{(n)}(\zeta_1) \text{ or } I^{(n)}(\zeta_2) \geq I^{(n)}(\zeta_3)
\]
\end{lem}

\begin{proof}
See Appendix \ref{appen:one-pop}.
\end{proof}

Next, we present our main result on the exit problem from convention $\bar m$ for finite population $n$.  Let $\mathcal{G}_{\bar{m}}^{(n)}$ be the set of all paths escaping the basin of attraction of convention
$\bar{m}$ at some arbitrary time $T$; that is,
\begin{align}\label{eq:path-G}
  \mathcal{G}^{(n)}_{\bar{m}}: =\{ \gamma=(x_{0},\cdots,x_{T})  \,:\, &  x_0= e_{\bar{m}}\,, x_{t}\in D^{(n)}(e_{\bar{m}}) \text{~for~} 0<t <T-1 \notag  \\
  & \text{~and~}  x_{T} \notin   D^{(n)}(e_{m}) \text{ for some } T>0 \}.
\end{align}
Thus, our exit problem can be written formally as
\begin{equation}\label{eq:min-n}
\min \{ I^{(n)}(\gamma):  {\gamma\in\mathcal{G}^{(n)}_{\bar{m}}} \}.
\end{equation}

Using Lemma \ref{lem:Pos-feedbacks}, we significantly reduce the number of candidate solutions to the problem in \eqref{eq:min-n} as explained in Section \ref{subsec:example} in the special case of three strategies. First, the comparison principle 1 in Lemma \ref{lem:Pos-feedbacks} implies that we can focus on the class of paths for which all transitions are from $\bar{m}$ to some other strategy $i$ (i.e., all paths consisting of straight lines, parallel to the edges of the simplex) (see Figure \ref{fig:path})---the set of paths defined as follows:
\begin{align} \label{eq:path-J}
{\mathcal J}^{(n)}_{\bar{m}} : & =\{\gamma =(x_{1},x_{2},\cdots,x_{T}) \in {\mathcal G}^{(n)}_{\bar{m}} \,;\,  x_{t+1} = (x_t)^{\bar{m},i} \text{ for some } i, \text{ for all } t \le T-1\}
\end{align}

Then, using the comparison principle 2 in Lemma \ref{lem:Pos-feedbacks}, we can further reduce the number of candidate solutions to \eqref{eq:min-n} and hence consider a subset of $\mathcal{J}^{(n)}_{\bar m}$ in which identical transitions from $\bar m$ occur consecutively:
\[
\gamma:\,\, x\substack{\textrm{from}\,\bar{m}\,\textrm{to}\, i_{1}\\
\longrightarrow\\
t_{1}\textrm{\,\ times}
}
y\substack{\textrm{from}\,\bar{m}\,\textrm{to}\, i_{2}\\
\longrightarrow\\
t_{2}\textrm{\,\ times}
}
z\,\,\,\cdots\substack{\textrm{from}\,\bar{m}\,\textrm{to}\, i_{K}\\
\longrightarrow\\
t_{K}\textrm{\,\ times}
}
w\,
\]
for some given (distinct) $i_1, i_2, \cdots, i_K$.
That is, $\gamma$ consists of a series of consecutive transitions first from $\bar{m}$ to $i_{1}$, then from $\bar{m}$ to $i_{2}$, ..., and finally from $\bar{m}$ to $i_{K}$. More precisely, we write $\gamma$
as
\begin{equation}
\gamma=(x_{1},x_{2},\cdots,x_{T})=\begin{cases}
x_{t+1}=(x_{t})^{\bar{m,}i_{1}} & \textrm{if}\,1\le t\leq t_{1}\\
x_{t+1}=(x_{t})^{\bar{m,}i_{2}} & \textrm{if}\, t_{1}+1\le t\leq t_{1}+t_{2}\\
\,\,\,\,\,\vdots\\
x_{t+1}=(x_{t})^{\bar{m},i_{K}} & \textrm{if}\,\sum_{l=1}^{K-1}t_{l}+1\leq t\leq\sum_{l=1}^{K}t_{l}=:T-1
\end{cases}\label{eq: ent_path}\,
\end{equation}
and define
\begin{align} \label{eq:path-K}
\mathcal{K}^{(n)}_{\bar{m}}:  = & \{\gamma:\gamma=(x_{1},x_{2},\cdots,x_{T})\in\mathcal{J}^{(n)}_{\bar{m}},\,\, \gamma\,\ \text{ is given by (\ref{eq: ent_path}) }\\
    & \text{ for some } t_{1},t_{2},\cdots,t_{K}, \text{ for some distinct } i_1, i_2, \cdots, i_K \} \nonumber.
\end{align}
For example, consider the following three paths:
\begin{align*}
    \gamma_1:  & \,\, x ^{(1)} \xrightarrow{\textrm{from}\,1 \,\textrm{to}\, 2} \,\,  x ^{(2)}  \,\, \xrightarrow{\textrm{from}\,1 \,\textrm{to}\, 3} \,\, x ^{(3)} \,\,  \xrightarrow{\textrm{from}\,1 \,\textrm{to}\, 2} \,\, x ^{(4)} \\
    \gamma_2:  & \,\, y ^{(1)} \xrightarrow{\textrm{from}\,1 \,\textrm{to}\, 2} \,\,  y ^{(2)}  \,\, \xrightarrow{\textrm{from}\,1 \,\textrm{to}\, 3} \,\, y ^{(3)}  \\
    \gamma_3:  & \,\, z ^{(1)} \xrightarrow{\textrm{from}\,1 \,\textrm{to}\, 3} \,\,  z ^{(2)}  \,\, \xrightarrow{\textrm{from}\,1 \,\textrm{to}\, 2} \,\, z ^{(3)}
\end{align*}
Then $\gamma_1, \gamma_2, \gamma_3 \in \mathcal{J}^{(n)}_1$ and $\gamma_2, \gamma_3 \in \mathcal{K}^{(n)}_1$, but $\gamma_1 \not \in \mathcal{K}^{(n)}_1$ since $\mathcal{K}^{(n)}_1$ contains only paths in which identical transitions occur consecutively (see Panels B and C in Figure \ref{fig:path}).
\noindent The main idea of the following Proposition \ref{prop:red} is illustrated in Section \ref{subsec:example}.
\begin{prop}[\textbf{Finite populations}]\label{prop:red}  Suppose that Condition $\textbf{A}$ holds. Then, we have the following characterizations: \\ \smallskip
(i) {\bf Comparison principle I}
\[
\min \{I^{(n)}(\gamma): {\gamma\in\mathcal{G}^{(n)}_{\bar{m}}} \}=\min \{I^{(n)}(\gamma): {\gamma\in\mathcal{J}^{(n)}_{\bar{m}}} \}.
\] \\ \smallskip
(ii) {\bf Comparison principle II}
\[
\min \{I^{(n)}(\gamma): {\gamma\in\mathcal{G}^{(n)}_{\bar{m}}} \}=\min \{I^{(n)}(\gamma): {\gamma\in\mathcal{K}^{(n)}_{\bar{m}}} \}
\]

\end{prop}
\begin{proof}
  See Appendix \ref{appen:one-pop}.
\end{proof}

\begin{figure}[t]
\centering\includegraphics[scale=0.6]{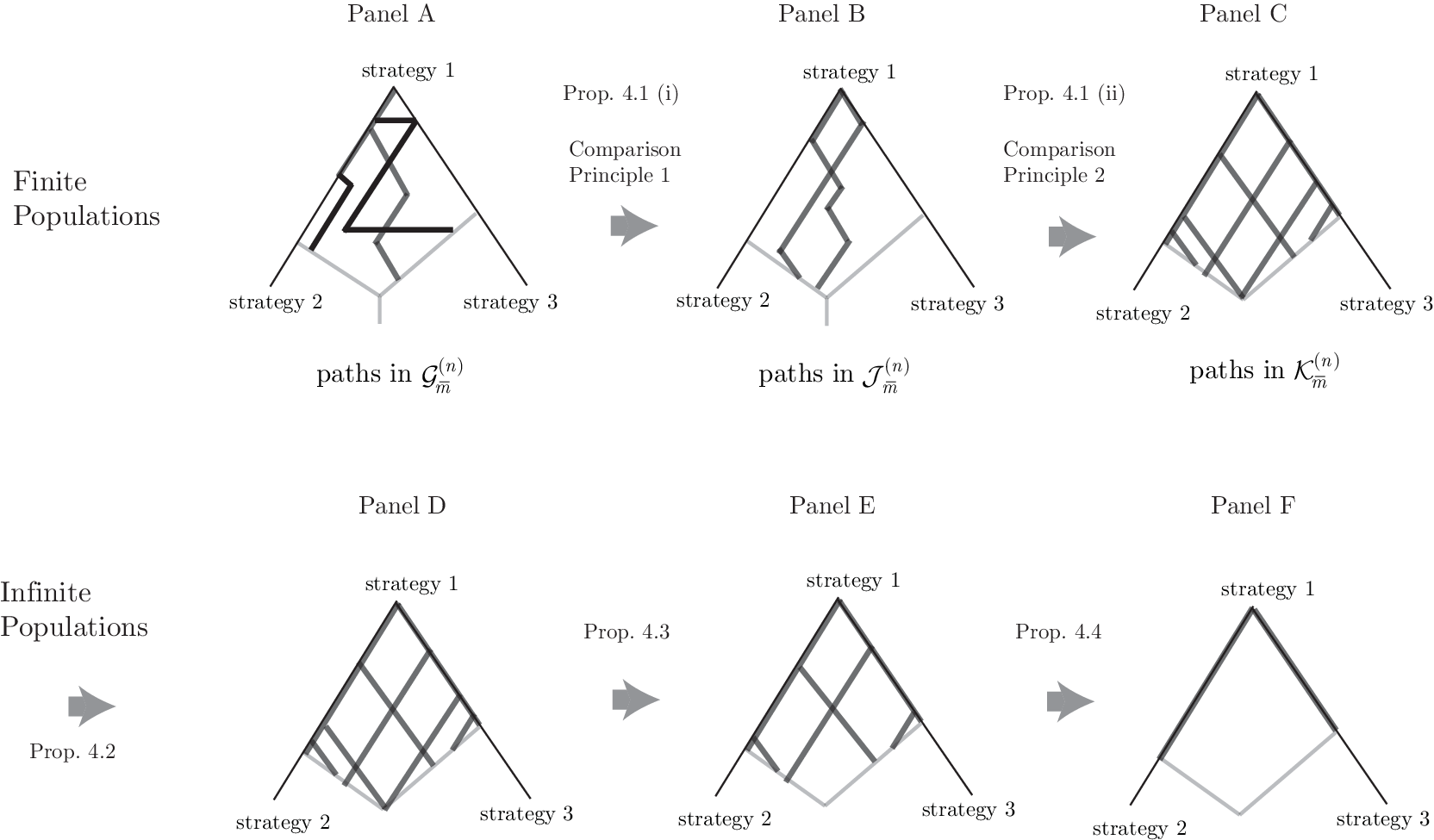}

\caption{\textbf{Outline of the proof of Theorem \ref{thm:escape}}. Each panel shows how each proposition reduces the number of candidate solutions to the exit problem.}
\label{fig:path}
\end{figure}


Even after eliminating a substantial number of irrelevant paths from the set of candidate solutions to the problem in \eqref{eq:min-n} in the setting of the finite population using Proposition \ref{prop:red} (Panels A, B, and C in Figure \ref{fig:path}), function $I^{(n)}$ still remains complicated, with negligible terms (in the order of $n$) when the population is large.
Thus, we consider an infinite population limit as in \citet{SS2014}. However, unlike \citet{SS2014} who solve the Hamilton Jabcobi equation to find the optimal path, we directly show that the minimizing escaping paths lie on the boundaries of the simplex using the properties of \textbf{MBP} (Panels D, E and F in Figure \ref{fig:path}). This shows that the \textbf{MBP} is also the key condition for comparing two paths in the infinite population model. Figure \ref{fig:path} illustrates the outline of the proof of the main theorem in this section, Theorem \ref{thm:escape}.


Specifically, to study the infinite population problem, we need to first find the continuous version of $c^{(n)}(x,x^{i,j})$ in \eqref{eq: l-cost}. For this, we denote by $\bar D(e_{i})$ and $\partial \bar D_{i,j}$ the continuous versions of $D^{(n)}(e_{i})$ and $\partial D^{(n)}_{i,j}$, respectively (in an appropriate convergence sense; see equation \eqref{eq:cont-basin}). Suppose that  $p,q \in \Delta$  with $q=p + \alpha(e_i-e_j)$ for some $\alpha >0$.  For $p,q \in \bar D(e_{\bar{m}})$,
define
\begin{equation}
\bar{c}(p,q):=\frac{1}{2}(p_{j}-q_{j})(\pi (\bar m,p+q)- \pi(i,p+q)).
\label{eq:main-con-cost}
\end{equation}
\noindent We show that the definition of a continuous cost function in \eqref{eq:main-con-cost} is precisely the limit of the discrete cost function at $n = \infty$ (see Lemma \ref{lem:con_cost} in Appendix \ref{appen:one-pop}).
We will denote the cost of a continuum path $\zeta$ as $\bar I(\zeta)$.

Next, we define a continuum analogue of the set of paths  ${\mathcal K}^{(n)}_{\bar{m}}$ (equation \eqref{eq:path-K}) in the limit of $n \to \infty$ and define an associated cost function as follows.  Roughly speaking $\bar{\mathcal{K}}_{\bar{m}}$ is the set of paths consisting  of a collection of piecewise straight lines. More precisely, for $t=((t_1, \cdots, t_K); (i_1, \cdots, i_K))$ where $t_l \in [0,1]$ for all $l$,
\begin{align}
p^{(0)}=e_{\bar{m}}, \,\,  p^{(1)}= e_{\bar{m}}+ t_{1}(e_{i_{1}}-e_{\bar{m}}), \cdots,\,\, p^{(K)}\,=\,  e_{\bar{m}}+\sum_{l=1}^{K}t_{l}(e_{i_{l}} - e_{\bar{m}}),\label{eq:con-transition}
\end{align}
where $p^{(l)} \in \bar D(e_{\bar{m}})$ for all $l$ and $p^{(K)} \in \partial \bar D(e_{\bar{m}}) := \bigcup_{l \neq \bar m} \partial \bar D_{\bar m , j}$. We also set
\begin{equation}\label{eq: endpoint}
    q(t)\,:= p^{(K)}, \zeta(t)\, := (p^{(0)}, p^{(1)}, \cdots, p^{(K)}),\text{ and }\quad \omega(t):= \bar{I} (\zeta(t)) = \sum_{t=0}^{K-1}{\bar c}(p^{(t)},p^{(t+1)})
\end{equation}
\noindent Note that path $\zeta=\zeta(t) \in \bar{\mathcal{K}}_{\bar m}$ is uniquely determined by the vector $t = ((t_1, \cdots, t_K);(i_1, \cdots, i_K))$. We then have the following infinite population limit result.

\begin{prop}[\textbf{Infinite Population Limit}]\label{prop:main-approx}
\[
\lim_{n\rightarrow\infty} \frac{1}{n} \min
\{  I^{(n)}(\gamma) \,:\, \gamma \in\mathcal{K}_{\bar{m}}^{(n)} \}=\min \left\{ \omega(t)\,:\,  \zeta(t) \in \bar{\mathcal{K}}_{\bar{m}}  \right\}.
\]
%
%
\end{prop}
\begin{proof}
    See Appendix \ref{appen:one-pop}.
\end{proof}
\noindent Proposition \ref{prop:main-approx} can be seen as  a special case of Theorem 9 in \citet{SS2014}, which shows the convergence over all possible paths. 
For completeness, we provide our own proof of Proposition \ref{prop:main-approx} in Appendix \ref{appen:one-pop}. Proposition \ref{prop:main-approx} leads us to consider the following minimization problem over the set  of continuous paths in $ \bar{\mathcal{K}}_{\bar{m}}$:

\begin{equation}\label{eq:con-exit-prob}
  \min \left\{ \omega(t)\,:\,  \zeta(t) \in \bar{\mathcal{K}}_{\bar{m}}  \right\}.
\end{equation}

\begin{figure}
\centering\includegraphics[scale=0.6]{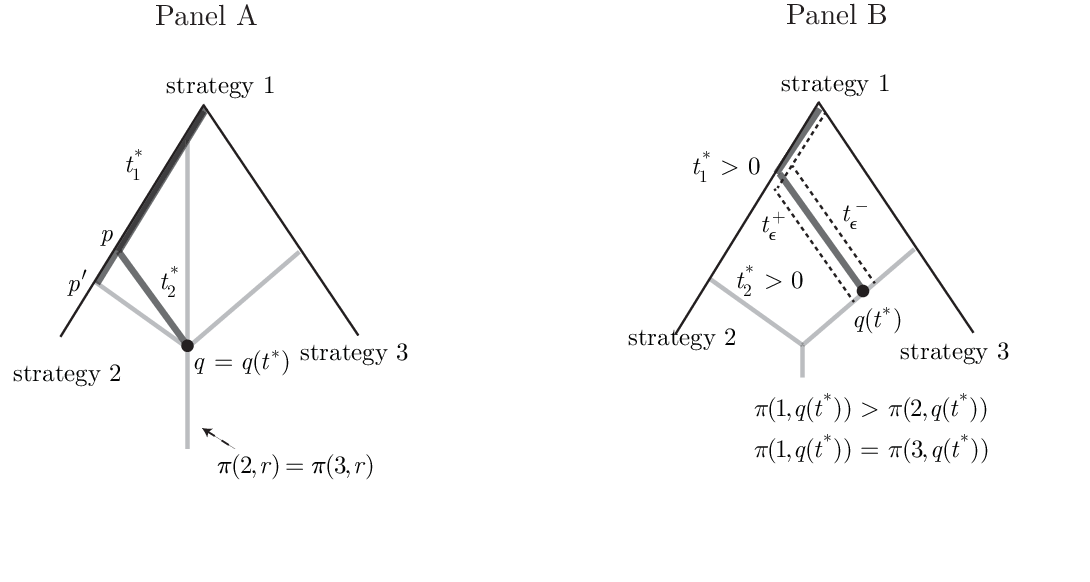}

\caption{\textbf{Proofs of Propositions \ref{prop:main-f-red} and \ref{prop:main-binding-con}.} Panel A illustrates Proposition \ref{prop:main-f-red}; Panel B illustrates Proposition \ref{prop:main-binding-con}. Panel A shows that if $t_1^* > 0$, $t_2^* >0$, $\pi(1,q(t^*))>\pi(2, q(t^*))$, and $\pi(1,q(t^*))=\pi(3, q(t^*))$, then either the path $t_\epsilon^+$ or the path $t_\epsilon^-$ can be lower than or equal to the cost of the path $(t_1^*, t_2^*)$. Panel B illustrates two paths---one involving $p$ and $p'$, and the other involving $p$ and $q$. Proposition \ref{prop:main-binding-con}  shows that the path through $p$ to $p'$ has a lower cost than the path through $p$ to $q$.  }
\label{fig:main-con}
\end{figure}

Next, we further reduce the number of candidate solutions to \eqref{eq:con-exit-prob} by comparing the costs of paths in the set of $\bar{\mathcal{K}}_{\bar m }$ (see Panels D, E, and F in Figure \ref{fig:path}). Note that at the end point $q$ of the escape path,
the following constraints must be satisfied:
\begin{equation}\label{eq:var-constraints}
  \pi (\bar m, q) \geq \pi(l, q) \text{ for all } l  \text{ with at least one constraint binding}
\end{equation}
Proposition \ref{prop:main-f-red} shows  that at the end point of the minimal escape path, \emph{only one} constraint is binding (i.e., $\pi(\bar m ,q) = \pi(k,q)$ for some $k$) and all other constraint are non-binding (i.e., $\pi(\bar m ,q) > \pi(l,q)$ for some $l \neq k$). To explain why, let  us  consider a three-strategy game with the path $t^{*}=((t_{1}^{*},t_{2}^{*});(2,3))$ passing through points $p$ and $q=q(t^*)$ in Panel A of Figure \ref{fig:main-con}, where $\pi(1,q(t^{*}))=\pi(2,q(t^{*}))=\pi(3,q(t^{*}))$. Furthermore, consider an alternative path exiting directly along the boundary between strategies 1 and 2 at $p'$; hence, $\pi(1,p')=\pi(2, p'), \, \pi(1, p')> \pi(3, p')$.  To compare the costs of these two paths, using the definition in \eqref{eq:main-con-cost}, we find that
\begin{align*}
\bar c(p,q) & =\frac{1}{2}t_{2}^{*}(\pi(1,p+q
)-\pi(3,p+q))\,, \quad \bar c(p,p')  =\frac{1}{2}(p_{2}'-p_{2})( \pi(1,p+p')- \pi(2,p+p')).
\end{align*}
Thus, from $t_2^* > p_2' -p_2$ (see Panel A of Figure \ref{fig:main-con}) and the fact that $q$ and $p'$ are mixed-strategy Nash equilibria,
\begin{align*}
\bar c(p,q)- \bar c(p,p')  \geq\frac{1}{2}t_{2}^{*}[ \pi(1,p)- \pi(3,p)-( \pi(1,p)- \pi(2,p))]
 =\frac{1}{2}t_{2}^{*}( \pi(2,p)-\pi(3,p)).
\end{align*}
Clearly, $ \pi(2,p)> \pi(3,p)$ from Panel A of Figure \ref{fig:main-con}, because $p$ is located to the left of the line $\pi(2, r) = \pi(3,r)$. Indeed, we confirm that
\begin{align*}
 \pi(2,p)- \pi(3,p) & = \pi(2,q+t_2^{*}(e_{1}-e_{3}))- \pi(3,q+t_2^{*}(e_{1}-e_{3}))\\
 & =t_2^{*}(A_{21}-A_{23}-A_{31}+A_{33}) > 0
\end{align*}
using the \textbf{MBP}, where we again use the fact that $q$ is the complete mixed strategy Nash equilibrium. Thus, we find that $\bar c(p,q) \geq \bar c(p, p')$. Once again, the underlying principle is that the \textbf{MBP} implies that a path with the same consecutive transitions is cheaper than those with different transitions.

\begin{prop} \label{prop:main-f-red}
  Suppose that Condition \textbf{A} holds. Suppose that $t^*$ such that $\omega(t^*)=\min \{\omega(t): \zeta(t) \in \mathcal{K}_{\bar m} \}$. Then for some $k$,
  \begin{equation}\label{eq:binding-con}
    \pi(\bar m , q(t^*)) = \pi(k, q(t^*)) \text{ and }  \pi(\bar m , q(t^*)) > \pi(l, q(t^*)) \text{ for all } l \neq k.
  \end{equation}
\end{prop}
\begin{proof}
  See Appendix \ref{appen:one-pop}.
\end{proof}

Our final step shows that the optimal path always exit through the edge of the simplex (i.e., the mixed strategy Nash equilibrium involving only two strategies).  To explain the idea behind Proposition \ref{prop:main-binding-con}, consider again the three-strategy game in Panel B  of Figure \ref{fig:main-con}. Suppose
that the optimal solution is $t^{*}=((t_{1}^{*},t_{2}^{*});(2,3))$ and that $t_{1}^{*},t_{2}^{*}>0$, $\pi(1,q(t^{*}))>\pi(2,q(t^{*}))$ and $\pi(1,q(t^{*}))=\pi(3,q(t^{*}))$ (see the dotted line in Panel B of Figure \ref{fig:main-con}).
Then, using the linearity of the
payoff functions, we consider $t^+_\epsilon := (t_{1}^{*}+\epsilon_{1},t_{2}^{*}-\epsilon_{2})$
and $t^-_\epsilon:=(t_{1}^{*}-\epsilon_{1},t_{2}^{*}+\epsilon_{2})$ which \emph{still}
satisfy the constraint of escaping from the basin of attraction. Then, by direct computation, we have
\begin{align*}
      & (\omega(t^+_\epsilon) - \omega(t))+(\omega(t^-_\epsilon)-\omega(t))
     =       - H_{1,2;2} \epsilon_1^2 + 2  H_{1,3; 2}  \epsilon_1 \epsilon_2- H_{1,3;3} \epsilon_2^2 \leq - (\sqrt{H_{1,2; 2}}\epsilon_1 -
     \sqrt{H_{1,3; 3}} \epsilon_2)^2 < 0
  \end{align*}
where $H_{i, j;k} := (A_{ii}- A_{ji}) - (A_{ik} - A_{jk})$
\com{
\begin{align*}
      & (\omega(t^+_\epsilon) - \omega(t))+(\omega(t^-_\epsilon)-\omega(t))
     =       - H^{(1)}_{2 2} \epsilon_1^2 + 2  H^{(1)}_{3 2}  \epsilon_1 \epsilon_2- H^{(1)}_{3 3} \epsilon_2^2 \\
      & \leq - H^{(1)}_{2 2} \epsilon_1^2 + 2 \sqrt{H^{(1)}_{2 2}}\sqrt{H^{(1)}_{3 3}} \epsilon_1 \epsilon_2  - H^{(1)}_{ 3 3} \epsilon_2^2 \leq - (\sqrt{H^{(1)}_{2 2}}\epsilon_1 -
     \sqrt{H^{(1)}_{3 3}} \epsilon_2)^2 < 0
  \end{align*}
}
Thus, either $\omega(t^+_\epsilon) < \omega(t)$ or $\omega(t^-_\epsilon) < \omega(t)$ holds. Again, this shows that the \textbf{MBP} plays an important role in determining the minimum cost escape path under the infinite population model. The proof of Proposition \ref{prop:main-binding-con} generalizes these arguments for an arbitrary $n$ strategy game.

\begin{prop} \label{prop:main-binding-con}
   Suppose that Condition \textbf{A} holds.   Let $t^*$ be the solution to the minimization problem in \eqref{eq:con-exit-prob}.
  Suppose that for some $k$ and $l$,
  \[
    \pi(\bar m , q(t^*)) = \pi (k, q(t^*)) \text{ and } \pi(\bar m, q(t^*)) > \pi(l, q(t^*))
  \]
  Then $t_l^* =0$.
\end{prop}
\begin{proof}
    See Appendix \ref{appen:one-pop}.
\end{proof}


Finally, by applying Propositions \ref{prop:red}, \ref{prop:main-approx}, we obtain 
\begin{align} \label{eq:thm-step1}
\lim_{n\rightarrow\infty} \frac{1}{n} \min
\{  I^{(n)}(\gamma) \,:\, \gamma \in\mathcal{G}_{\bar{m}}^{(n)} \}
     =\min \left\{ \omega(t)\,:\,  \zeta(t) \in \bar{\mathcal{K}}_{\bar{m}}  \right\}.
\end{align}
Now let $t^*$ be the solution to $\min \left\{ \omega(t)\,:\,  \zeta(t) \in \bar{\mathcal{K}}_{\bar{m}}  \right\}$.  Then, from Proposition \ref{prop:main-f-red},  there exists $k$ such that equations in \eqref{eq:binding-con} hold. By applying Proposition \ref{prop:main-binding-con}, we conclude that $t^*$ involves only transitions from $\bar m$ to $k$ and this path exits at the mixed strategy involving only $\bar m$ and $k$. Then, we find the cost of $t^*$ as follows (see  Lemma \ref{lem:con_cost})\
\[
    \omega(t^*) = \frac{1}{2} (1- \frac{A_{kk}-A_{\bar m k}}{(A_{\bar m \bar m}-A_{k \bar m }) + (A_{kk}-A_{\bar m k})})(\pi(\bar m, e_{\bar m})- \pi(k, e_{\bar m})) = \frac{1}{2}\frac{(A_{\bar m \bar m} - A_{k \bar m})^2}{(A_{\bar m \bar m} - A_{k \bar m})+(A_{k k } - A_{\bar m k})}.
\]
Thus, we obtain 
\begin{equation}\label{eq:thm-ineq}
  \min \left\{ \omega(t)\,:\,  \zeta(t) \in \bar{\mathcal{K}}_{\bar{m}}  \right\}=\omega(t^*) \geq \min_{j \neq \bar m} \left \{ \frac{1}{2} \frac{(A_{\bar m \bar m} - A_{j \bar m})^2}{(A_{\bar m \bar m} - A_{j \bar m })+(A_{jj} - A_{\bar m j})} \right \}.
\end{equation}
Since the expression in curly braces of the right hand side of \eqref{eq:thm-ineq} is the cost of the straight line escape path from $\bar m$ to $j$ and these paths belong to $\bar{\mathcal{K}}_{\bar m}$, \eqref{eq:thm-ineq} holds as an equality.  Thus, \eqref{eq:thm-step1} and \eqref{eq:thm-ineq} (as an equality) yield the following theorem.

\begin{thm}[\bf Exit problem: one population model]\label{thm:escape}
Assume that Condition \textbf{A} holds.  Then, we have
\begin{equation}\label{eq:logit-cost}
\lim_{n\rightarrow\infty} \frac{1}{n} \min  \{  I^{(n)}(\gamma) \,:\, \gamma \in\mathcal{G}_{\bar{m}}^{(n)} \} = \min_{j \neq \bar m} \left \{ \frac{1}{2} \frac{(A_{\bar m \bar m} - A_{j \bar m})^2}{(A_{\bar m \bar m} - A_{j \bar m })+(A_{jj} - A_{\bar m j})} \right \}
\end{equation}
\end{thm}
\begin{proof}
  See Appendix \ref{appen:one-pop}.
\end{proof}

%

In the much-studied uniform mistake model, the probabilities of mistakes are identical for all states and, hence, are independent of the states (see, e.g, \citet{BSY1}). Therefore, the threshold number of deviant agents who induce other agents to change their best responses is the only determinant of the expected escape time and stochastic stability. The number $n \frac{(A_{\bar{m}\bar{m}}-A_{i\bar{m}})}{(A_{\bar{m}\bar{m}}-A_{i\bar{m}})+(A_{ii}-A_{\bar{m}i})}$ is the threshold number of agents deviating from strategy $\bar m$ to strategy $i$ and inducing others to best respond with strategy $i$. Obviously, our arguments for Theorem \ref{thm:escape} can be applied to the uniform mistake model, as explained in the introduction. In particular, comparison principe 1 holds as an equality, which can still be used to find a minimum cost path and comparison principle 2 holds without any modification. Thus, when the $\textbf{MBP}$ holds,
\begin{equation} \label{eq:uniform-cost}
\lim_{n \rightarrow \infty}   \frac{1}{n}  \min_{\gamma \in\mathcal{G}_{\bar{m}}^{(n)}}I_{uniform}^{(n)}(\gamma) = \min_{j \neq \bar m} \left \{  \frac{(A_{\bar{m}\bar{m}}-A_{j \bar{m}})}{(A_{\bar{m}\bar{m}}-A_{j \bar{m}})+(A_{jj}-A_{\bar{m}j})} \right \}
\end{equation}
where $I_{uniform}$ is the cost function under the uniform mistake model (see also \citet{BSY1, Kandori98}).

Comparing \eqref{eq:logit-cost} and \eqref{eq:uniform-cost} reveals that the logit rule cost also accounts for the opportunity cost of mistakes ($A_{\bar m \bar m} - A_{j \bar m}$) adopting  sub-optimal strategy $j$ instead of the best response strategy $\bar m$ as well as the threshold fraction of the population inducing a new best response. Specifically, under the logit model, the threshold fraction of agents inducing a new best response is weighted by individuals' average cost of mistakes, $\frac{1}{2}(A_{\bar m \bar m} - A_{j \bar m})$ (equation \eqref{eq:logit-cost}). This (plausible) modification of costs creates a different prediction for the exit path and time, as well as a stochastically stable state from the uniform model (see Section \ref{sec:ndg}).


\tocless \section{Two population models and the application: the logit solutions for the Nash demand game \label{sec:application}}

\tocless \subsection{Two population models}

For the application of our results to the Nash demand game, we briefly introduce the two population setup and state the result for two population models (Theorem \ref{thm:main-2p}), which is proved in a similar way to the one population model result(see Appendix \ref{appen:2p-exit}). Consider two populations denoted by $\alpha$ and $\beta$, consisting of the same number of agents $n$ and a bimatrix game $(A^{\alpha},A^{\beta})$, where $A^{\kappa}$ is an $|S|\times|S|$ matrix for $\kappa =\alpha,\beta$. An $\alpha$-agent playing $i$ against $j$ obtains a payoff $A_{ij}^{\alpha}$, while a $\beta$-agent playing $j$ against $i$ obtains $A_{ij}^{\beta}.$
We introduce the following definition.
\begin{defn}
\label{def:2p-MBP}We say that $(A^{\alpha},A^{\beta})$ is a coordination
game if $A_{ii}^{\alpha}>A_{ji}^{\alpha}$ and $A_{ii}^{\beta}>A_{ij}^{\beta}$
for all $i,j$. We also say that $(A^{\alpha},A^{\beta})$ satisfies
the \emph{weak} marginal bandwagon property (\textbf{WBP}) if
\begin{equation}\label{eq:wmbp}
  A_{\bar{m}\bar{m}}^{\alpha}-A_{i\bar{m}}^{\alpha} \geq A_{\bar{m}j}^{\alpha}-A_{ij}^{\alpha}\mbox{ \text{and }}A_{\bar{m}\bar{m}}^{\beta}-A_{\bar{m}i}^{\beta} \geq A_{j\bar{m}}^{\beta}-A_{ji}^{\beta}
\end{equation}
for all distinct $i,j,\bar{m}$.
\end{defn}
Note that condition \eqref{eq:wmbp} is weaker than the \textbf{MBP} which requires strict inequalities. We relax the \textbf{MBP}, because we would like to study a broader class of games including Nash demand games (\citet{Nash53})\footnote{It is straightforward to check that a discrete Nash demand game satisfies \eqref{eq:wmbp}. See Appendix \ref{appen:sss-NDG}.}. Let $x=(x_{\alpha},x_{\beta})\in \Delta_\alpha^{(n)} \times \Delta_\beta^{(n)}$, where $\Delta_\kappa^{(n)}:=\{ x_\kappa: \sum x_\kappa(i) = 1, \text{ and } x_\kappa(i) \geq \text{ for all } i \}$ for $\kappa = \alpha, \beta$. Then, the expected payoffs are similarly given by
\[
\pi_{\alpha}(i,x)=\pi_{\alpha}(i,x_{\beta})=\sum_{j=1}^{n} x_\beta (j) A_{ij}^{\alpha},\,\,\pi_{\beta}(j,x)=\pi_{\beta}(i,x_{\alpha})=\sum_{i=1}^{n} x_\alpha (i) A_{ij}^{\beta}
\]

Similarly to Condition \textbf{A}, we make the following assumptions:\medskip

\noindent \textbf{Condition B}  \\
(i) A game $(A^\alpha, A^\beta)$ is a coordination game, satisfying the \textbf{WBP}. \\
(ii) Suppose that \\
for any $T=\{i_1, \cdots, i_K\} \subset S$, there exists  a unique $q \in \Delta_\beta$ with support $T$ such that
\[
    \pi_\alpha(i_1, q) = \cdots =\pi_\alpha(i_K, q),
\]
for any $T=\{i_1, \cdots, i_K\} \subset S$, there exists a unique $p \in \Delta_\alpha$ with support $T$  such that
\[
   \pi_\beta(i_1, p) =\cdots =\pi_\beta(i_K, p).
\]
$\square$


We denote by  $x^{\alpha,i,j}$ (or $x^{\beta,i,j}$)
the state induced by an $\alpha$-agent's (or $\beta$-agent's) switching from $i$
to $j$ from $x$. We also denote by $x^{\alpha,i,j,\rho}$ (or $x^{\beta,i,j,\rho}$)
the state induced from $x$ by $\alpha$-agents' $\rho$-times consecutive
switchings (or $\beta$-agents $\rho$-times consecutive transitions)
from $i$ to $j$. We write $e_{i}^{\alpha}$ and $e_{j}^{\beta}$ as the $i$-th and $j$-th standard basis for $\mathbb{R}^{|S|}$ and
thus $(e_{i}^{\alpha},e_{i}^{\beta})\in \Delta_\alpha^{(n)} \times \Delta_\beta^{(n)}$ is a convention.
We similarly define a basin of attraction of convention $i$ as follows:
\[
D^{(n)}(e_{\bar{m}}):=\{x=(x_{\alpha},x_{\beta})\in \Delta_\alpha^{(n)} \times \Delta_\beta^{(n)}:\,\pi_{\alpha}(\bar{m},x_{\beta})\geq\pi_{\alpha}(k,x_{\beta}), \,\pi_{\beta}(\bar{m},x_{\alpha})\geq\pi_{\beta}(k,x_{\alpha})\,\mbox{for all \ensuremath{k}}\}
\]
and compute the cost functions between states:
\begin{equation}\label{eq:2p-cost}
  c^{(n)}(x,x^{\alpha,i,j}):=\pi_{\alpha}\mbox{(\ensuremath{\bar{m}},}x)-\pi_{\alpha}(j,x),\,\,c^{(n)}(x,x^{\beta,i,j}):=\pi_{\beta}\mbox{(\ensuremath{\bar{m}},}x)-\pi_{\beta}(j,x)
\end{equation}
for $x\in D^{(n)}(e_{\bar{m}})$. We similarly define the cost function of a path,
$I^{(n)}$, as in equation \eqref{eq: cost-path} and the set of all paths escaping the basin
of attraction of convention $\bar{m}$ as $\mathcal{G}_{\bar{m}}^{(n)}$, as in equation \eqref{eq:path-G}.

For two population models, two kinds of deviant behaviors naturally arise depending on the intentionality of deviant plays \citep*{Naidu10, HLNN2016}, namely an unintentional or intentional logit choice rule. The unintentional logit choice rule refers to the standard logit rule in the context of the two population model, defined as follows:
\begin{equation} \label{eq:appen-2p-logit}
\textbf{Unintentional logit choice rule: } \,\, p^{U, (\kappa)}_\eta (l|x):= \frac{\exp(\eta^{-1} \pi_\kappa (l, x))}{\sum_{l'} \exp(\eta^{-1} \pi_ \kappa(l', x)} \text{ for } \kappa = \alpha, \beta.
\end{equation}

By the intentional choice rule, we mean that agents choose a non-best-response strategy  from the set of strategies that would give a higher payoff than the payoff at the status quo convention when adopted as a convention. Thus, the intentional logit choice rule means that agents play a non-best-response strategy among the set restricted by ``intentional'' behaviors, with probabilities specified by the logit choice rule. Experimental evidence for intentional as well as payoff-dependent behaviors captured by the logit rule is provided in \citet{MasNax2016}, \citet{LimNeary2016}, and \citet*{HLNN2016}(see Section \ref{sec:lit}).\footnote{\citet{LimNeary2016} find that individual mistakes are directed in the sense that they are group-dependent. The directed mistakes in their paper are intentional behaviors of deviant agents; for example, they find that 2.25\% of subjects play mistakes when the best response is the preferred strategy, whereas 20.85\% of subjects play mistakes when the best response is the less preferred strategy (Figure 5 (a) on p .19). }

More precisely, to define the intentional logit choice rule, we introduce the sets of the permissible suboptimal strategies under the intentional dynamic:
\begin{equation}\label{eq:s-int-dyn}
  \tilde S_m(\kappa) := \{ l : A_{ll}^\kappa \geq A^\kappa_{mm} \}, \quad
  \hat S_x (\kappa):= \{ l : A_{ll}^\kappa \geq A_{m m}^{\kappa} \text{ for } m \in \arg \max_i \{ \pi_\kappa (i, x) \} \}
\end{equation}
That is,  $\tilde S_m(\kappa)$ is the set of all strategies which yield payoffs higher than (or equal to) strategy $m$ if  adopted as a convention, while $\hat S_x(\kappa)$ is the set of all strategies which yield payoffs higher than (or equal to) the current convention at state $x$ if adopted as a convention.  We make the following assumption for conflict of interests between the two populations, $\alpha$ and $\beta$:
\begin{equation}\label{eq:conflict}
          \tilde S_m(\alpha) \cup \tilde S_m(\beta) = S, \quad \tilde S_m(\alpha) \cap \tilde S_m(\beta) = m
\end{equation}
Thus, except for the current convention strategy $m$, the set of strategies that $\alpha$ population prefers is different from the set of strategies that $\beta$ population prefers. Then under the intentional logit choice rule, the conditional probability that a  $\kappa = \alpha, \beta$ agent chooses strategy $l$ at a given state $x$ is given by
\begin{equation} \label{eq:appen-2p-int-logit}
  \textbf{Intentional logit choice rule: }   p^{I, (\kappa)}_\eta (l|x):= \begin{cases}
                                 \frac{\exp(\eta^{-1} \pi_\kappa (l, x))}{\sum_{l' \in \hat S_\kappa(x)} \exp(\eta^{-1} \pi_ \kappa(l', x)} & \mbox{if } l \in \hat S_\kappa (x) \\
                                 0  & \mbox{if } l \not \in \hat S_\kappa (x).
                               \end{cases}
\end{equation}
\com{Then the costs for the unintentional and intentional logit choice rules are given as follows:
\begin{align*}
 \text{Unintentional: }\,\, c^U(x, x^{\kappa, i, j}) & = \max \{ \pi_\kappa (l, x) : l \in S \} - \pi_\kappa (j, x) \\
 \text{Intentional:  }\,\,  c^I(x, x^{\kappa, i, j}) & = \begin{cases}
                                 \max \{ \pi_\kappa (l, x) :  l \in S \} - \pi_\kappa(j, x) & \mbox{if } j \in \hat S_\kappa(x) \\
                                 \infty & \mbox{if } j \not \in \hat S_\kappa(x).
                               \end{cases}
\end{align*}
}
To state our result for the two population model, we let $\zeta_\alpha$ (or $\zeta_\beta$)  the threshold fractions of $\beta$-population (or $\alpha$-population) inducing a new best response in $\alpha$-population (or $\beta$-population) be
\[
    \zeta^\alpha_{m, j} := \frac{(A_{m m}^{\alpha}
-A_{j m }^{\alpha})}{(A_{m m}^{\alpha}
-A_{j m }^{\alpha})+(A_{jj}^{\alpha}-A_{mj}^{\alpha})},
\qquad
    \zeta^\beta_{m, j} := \frac{(A_{mm}^{\beta}-A_{mj}^{\beta})}
{(A_{mm}^{\beta}-A_{mj}^{\beta})
+(A_{jj}^{\beta}-A_{jm}^{\beta})}.
\]

In Theorem \ref{thm:main-2p} below, we show that the minimum cost escaping path from convention $m$ is similarly given by the threshold fraction weighted by the individual's opportunity cost ($A^\beta_{mm}-A^\beta_{mj}$ for $\beta$-population agents and $A^\alpha_{mm}-A^\alpha_{jm}$ for $\alpha$-population agents). To state this, we let
\begin{equation}\label{eq:cost-min}
  R^U_{mj} := (A_{mm}^\beta - A_{mj}^\beta) \zeta^\alpha_{mj} \wedge (A_{mm}^\alpha - A_{jm}^\alpha) \zeta^\beta_{mj}, \, \qquad \,
    R^I_{mj} := \begin{cases}
                  (A^\alpha_{m m} - A^\alpha_{jm}) \zeta_{mj}^\beta, & \mbox{if } j \in \tilde S_m(\alpha) \\
                  (A^\beta_{m m} - A^\beta_{mj}) \zeta_{mj}^\alpha, & \mbox{if } j \in \tilde S_m(\beta).
                \end{cases}
\end{equation}
\noindent Note that under the intentional logit choice rule ($R^I_{mj}$ in \eqref{eq:cost-min}), the sets of strategies for which minimum cost transitions occur are precisely the sets of strategies that the transition driving population prefers. This is because under the intentional dynamic, the deviant plays always involve strategies that the deviant population prefers. For the two population models, the minimum cost escape path from convention $m$ again occurs at the boundary of the simplex (the state space), as is the case for the one population model:
\begin{thm} \label{thm:main-2p}
Suppose that Condition \textbf{B} and equations \eqref{eq:conflict} hold. Then
 \begin{align*}
 \textbf{Unintentional:} &  \lim_{n\rightarrow\infty}\frac{1}{n}\min_{\gamma \in\mathcal{G}_{\bar{m}}^{(n)}}I^{(n)}(\gamma) = \min_{j \neq m} R^U_{mj}   \\
  \textbf{Intentional:} & \lim_{n\rightarrow\infty}\frac{1}{n}\min_{\gamma \in\mathcal{G}_{\bar{m}}^{(n)}}I^{(n)}(\gamma) = \min_{j \neq m} R^I_{mj}
\end{align*}

\end{thm}
\begin{proof}
See Appendix \ref{appen:2p-exit}.
\end{proof}

\tocless \subsection{The logit solutions for the Nash demand game \label{sec:ndg}}

In this section, we present the application of our results to the Nash demand games \citep{Nash53}, using Theorem \ref{thm:main-2p}. Consider a bargaining set, $S \subset \mathbb{R}^2$, consisting of payoffs to two populations when they agree to split. We normalize the ``disagree'' point to $(0,0)$.  We describe the bargaining frontier by the function $f$ that is decreasing, differentiable and strictly concave. That is,
\[
    S = \{ (x, y) : y \leq f(x) \}
\]
where $f(x) \geq 0$, $f'(x)<0$, and $f''(x)<0$ for all $x$. We let $\bar s$ be the maximum payoff to populations $\alpha$: i.e., $\bar s := \max \{x: (x,y) \in S\}$.
 Bargaining solutions dictate how to divide the surpluses (defined by the bargaining set) between two populations\footnote{Three axiomatic bargaining solutions are most commonly used: the Nash bargaining solution \citep{Nash1}, the Kalai-Smorondinsky bargaining solution \citep{KS1}, and the egalitarian bargaining solution \citep{Kalai1}.}.

Which is the bargaining norm arising through decentralized evolutionary bargaining processes under the logit choice rule? To answer this question, following the standard literature on evolutionary bargaining \citep{Young93JET, Young98Res, BSY1, HLNN2016}, we discretize the bargaining set as follows and consider a Nash demand game. More specifically, let $L$ be a positive integer and $\delta := \frac{\bar s}{L}$. Define a Nash demand game:
\begin{equation}\label{eq:main-ndg}
  (A^\alpha_{ij}, A^\beta_{ij}) :=
  \begin{cases}
    (\delta i, f(\delta j)), & \mbox{if } i \leq j \\
    (0, 0) , & \mbox{if } i > j,
  \end{cases}
\end{equation}
where $i \in \{ 1, 2, \cdots, L-1 \}$ and we will let $\delta \rightarrow 0$ (or $L \rightarrow \infty$)  eventually.  Then, we can easily check that the game in \eqref{eq:main-ndg} satisfies Condition \textbf{B} (See Appendix \ref{appen:sss-NDG}).

\begin{table}\label{tab:solutions}
\setstretch{1.5}
\centering
	\begin{tabular}{c|c|c}
    \toprule

          &  Unintentional & Intentional \\
	\hline
	Uniform & $\frac{f(s^{NB})}{s^{NB}} = -f'(s^{NB})$ & $\frac{f(s^{NB})}{s^{NB}} = -f'(s^{NB})$ \\
	\hline
    Logit & $\frac{f(s^{NB})}{s^{NB}} = -f'(s^{NB})$ &  $\left(\frac{ f(s^I)}{s^I}\right)^2=- f'(s^I)$\\
    \bottomrule

	\end{tabular}
		\caption{\textbf{Bargaining solutions for the Nash demand game.} The bargaining solutions for the Nash demand game under the logit choice rule are new.} \label{tab:solutions}
\end{table}

The assumption for conflict in interests in \eqref{eq:conflict} also holds for the game defined in \eqref{eq:main-ndg}. Thus, while the $\alpha$-population prefers strategy $m$ to strategy $m'$ (for $m>m'$), the $\beta$-population prefers strategy $m'$ to strategy $m$ (for $m>m'$). This is because at a higher (or lower) index convention, an $\alpha$-population agent (or a $\beta$-population agent) claims a large share. Thus, under the intentional logic dynamic, an $\alpha$ population agent idiosyncratically plays a suboptimal strategy from the set of strategies with higher indices than the current strategy, whereas a $\beta$ population agent does the opposite.

Next, we find a stochastic stable state using Theorem \ref{thm:main-2p} (see Appendix \ref{appen:suff} and Theorem \ref{thm:SSE}). To explain this, first consider the intentional logit choice rule for its simplicity. We would like to find the minimum cost of transitions from convention $m$ ($R^I_{mj}$ in \eqref{eq:cost-min}); the costs of transition driven by each population are given as follows:
\begin{align}
\alpha \textbf{-pop.:  } & \min_{j \in \tilde S_m(\alpha)} \{ (A_{m m}^{\alpha}-A_{j m}^{\alpha})
\frac{(A_{m m}^{\beta}-A_{m j}^{\beta})}
{(A_{m m}^{\beta}-A_{m j}^{\beta})
+(A_{jj}^{\beta}-A_{j m}^{\beta})}\}
= \min_{m < j} m \delta \frac{f(\delta m) - f(\delta j)}{f(\delta m)}
\label{eq:r-bargain2}
\\
\beta \textbf{-pop.:  } & \min_{j \in \tilde S_m(\beta)} \{(A_{m m}^{\beta}
-A_{m j}^{\beta})\frac{(A_{m m}^{\alpha}
-A_{j m}^{\alpha})}{(A_{m m}^{\alpha}
-A_{j m}^{\alpha})+(A_{jj}^{\alpha}-A_{m j}^{\alpha})}\}
=\min_{m > j} f(m \delta) \frac{(m-j) \delta}{m \delta}
 \label{eq:r-bargain1}
\end{align}
for a given $m$.
In the second equation in \eqref{eq:r-bargain2}, the minimum is taken over $\{j:  m<j \}$ since the $\alpha$-population prefers the strategy with a higher index, while the opposite holds true for the $\beta$-population in \eqref{eq:r-bargain1}. We then find the minimum cost of transitions, also called the radius for convention $m$, as follows:
\begin{equation} \label{eq:bargain-res}
	\min_j R^I_{m j} = \min \{ f(m \delta) \frac{\delta }{m \delta} , m \delta \frac{ f(\delta m ) - f(\delta (m+1))}{f(\delta m)} \}
\end{equation}
It can be shown that a convention $m$ which maximizes $\min_j R^I_{m j}$ is a stochastically stable state for the game in \eqref{eq:main-ndg} under the unintentional or intentional logit choice rule (see Appendix \ref{appen:sss-NDG}).
Note that the first term in the minimum of equation \eqref{eq:bargain-res} is decreasing in $m$ and the second term in the minimum of equation \eqref{eq:bargain-res} is increasing in $m$. Thus, the maximum of \eqref{eq:bargain-res} is achieved where the gap between the two terms in the minimum of equation \eqref{eq:bargain-res} is smallest. Heuristically, we find that
\[
 f( \delta m) \frac{\delta }{ \delta m} = \delta m \frac{ f(\delta m ) - f(\delta (m+1))}{f(\delta m)} \iff
 \left (\frac{f( \delta m)}{ \delta m} \right)^2  =\frac{ f(\delta m ) - f(\delta (m+1))}{\delta} \rightarrow \left  (\frac{f(x)}{x} \right)^2 = - f'(x)
\]
if $ \delta m \rightarrow x$ and $\delta \rightarrow 0$

More precisely, we define
\begin{equation} \label{eq:l-solution}
    \left (\frac{f(s^I)}{s^I} \right )^2 = -f'(s^I) \text{ and }  \frac{f(s^{NB})}{s^{NB}} = -f'(s^{NB}).
\end{equation}
and $s^{NB}$ is the familiar Nash bargaining solution and $s^I$ is a new solution under the intentional logit choice rule. We also let $s^E$ be the egalitarian solution:
\[
    f(s^E) = s^E.
\]
We then find
\begin{equation}\label{eq:min}
	(\min_j R^I_{m j},  \arg \min_j R^I_{m j})=\begin{cases}
		(m \delta \frac{f(\delta m) - f(\delta (m+1))}{f(\delta m)}, \,\, m +1)
		&\text{ if } \delta m  < s^I \\
		(\frac{f(m \delta)}{m \delta} \delta, \,\, m -1)
		& \text{ if } \delta m  > s^I
	\end{cases}
\end{equation}

Let $m^I \in \arg \max_m \min_j R^I_{mj}$. Then it is straightforward to show that $m^I$ is stochastically stable and
\[
	\delta m^I \rightarrow s^I \text{ as } \delta \rightarrow 0
\]
(see Appendix \ref{appen:sss-NDG}); thus, the solution $s^I$ defined in \eqref{eq:l-solution} is the stochastically stable bargaining convention under the intentional logit dynamic. In addition, the following theorem shows that the Nash bargaining solution defined in \eqref{eq:l-solution} is the stochastically stable bargaining norm under the unintentional logit rule (See Table \ref{tab:solutions}).

\begin{thm} \label{thm:NB}
    Suppose that $f(x) \geq 0$, $f'(x)<0$, and $f''(x)<0$ for all $x$. Let $m^*$ and $m^{I}$ be the stochastically stable states under the unintentional  and intentional logit models, respectively. Then we have
    \[
        \delta m^* \rightarrow s^{NB} \text{ and }  \delta m^I \rightarrow  s^{I}
    \]
    as $\delta \rightarrow 0$.
\end{thm}
\begin{proof}
  See Appendix \ref{appen:sss-NDG}.
\end{proof}

Figure \ref{fig:tree} compares cost minimizing paths under the various mistake models for the Nash demand game (see also Table \ref{tab:solutions}). Under the unintentional logit choice rule, the Nash bargaining solution arises as a bargaining norm, as in the uniform model (Theorem \ref{thm:NB} and see \citet{BSY1}). However, the underlying mechanism is dramatically different from the uniform model as Figure \ref{fig:tree} shows. Under the unintentional logit choice rule, unlike the uniform mistake model, when $\alpha$ population's share at the Nash bargaining solution ($s^{NB}$) is greater than the equal division ($s^E$), transitions from the egalitarian solution to the Nash bargaining solution are induced by $\beta$-population agents even if these transitions are $\alpha$ -population's favorite transitions (transitions to a higher index convention). This is because under the logit dynamic, this opportunity cost of the deviant play (as well as the threshold fraction)  matters and when the $\alpha$ -population claims a larger share than the $\beta$-population, the opportunity cost can make transitions induced by the $\alpha$-population more costly than those induced by the $\beta$-population. Thus, under the unintentional logit choice model, when the Nash bargaining solution favors the $\alpha$-population, transitions from the egalitarian solution are driven by the $\beta$-population (see $x_1, x_2, x_3, x_4$ in Figure \ref{fig:tree}), leading to the Nash bargaining solution in which the $\alpha$-population claims a larger share than the $\beta$-population.

\begin{figure}
  \centering
  \includegraphics[scale=0.65]{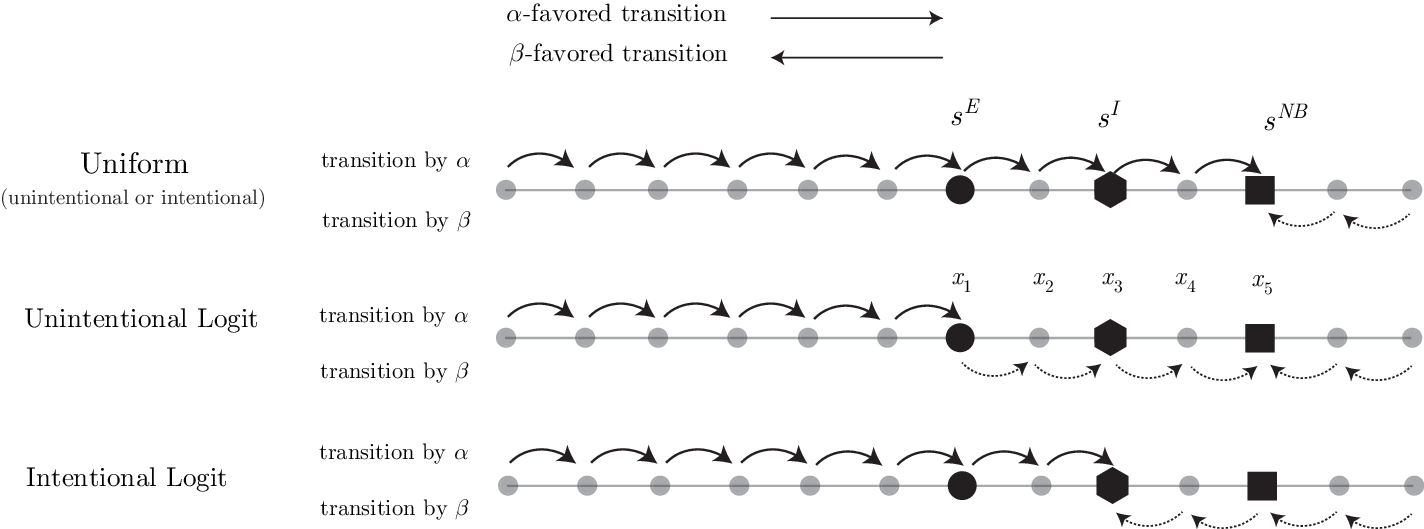}
  \caption{\textbf{Minimal trees}. The closed circle, hexagon, and square correspond to the egalitarian solution ($s^E$), the intentional logit solution ($s^I$), and the Nash bargaining solution ($s^{NB}$), respectively.  Each small dot represents a convention. The regular arrows show transitions induced by the $\alpha$-population, while the dotted arrows show transitions induced by the $\beta$-population. Under the unintentional logit choice rule, when $\alpha$'s share of Nash bargaining solution $s^{NB}$ is greater than that of the egalitarian solution $s^E$, the transitions from the egalitarian solution to the Nash bargaining solution are driven by $\beta$-population.}\label{fig:tree}
\end{figure}

Under the intentional logit choice rule, transitions always occur by those who stand to benefit from the transitions. That is, every transition to a higher index convention is driven by the $\alpha$-population, while every transition to a lower index convention is driven by the $\beta$-population. Thus, some transitions from the egalitarian solution to a higher index convention driven by the $\beta$-population under the unintentional model \com{ (transition OOO in Figure OOO)} are now replaced by  transitions to a lower index convention by the $\beta$-population and this pushes the stochastically stable state toward a more equal convention (see $x_4, x_5$ in Figure \ref{fig:tree}). In this way,  intentionality equalizes the division of the surplus.

   \begin{prop} \label{prop:intent}
    The logit intentional bargaining solution is more equal than the logit unintentional bargaining solution: i.e., either
    \[
        s^{NB} > s^I > s^{E} \text{ or }  s^{E} > s^I > s^{NB}
    \]
   holds.
   \end{prop}
   \begin{proof}
     See Lemma \ref{lem:comp}.
   \end{proof}

\citet{HLNN2016} also find that for the contract game (the coordination game in which all off-diagonal payoffs are zeroes), intentionality makes the stochastically stable  division of surpluses between the two populations more equal under the logit choice rule. However, the
mechanism under which equality is achieved is again different. Under the contract game,
intentionality equalizes the threshold fractions of one population inducing a new
best response for the other population (at $1/2$), and equality is achieved. For the
Nash demand game, intentionality replaces unfavorable transitions for the deviant population
by the favorable transitions and a more equal convention arises.

\tocless \section{Summary \label{sec:sum}}

Relying on positive feedback conditions and the relative strengths of these effects we developed methods to identify the most likely paths for evolutionary population dynamics under the logit rule. We identified two main factors determining the minimum cost path escaping from a convention: (1) the existence of  positive feedback effects, and (2) the relative strengths of positive feedback effects. This leads us to simple but powerful comparison principles that drastically reduce the number of candidate paths for minimizing the escaping cost from a convention. To summarize, we showed that the path with minimal cost involves only the repeated identical mistakes of the agents. We also applied our finding to the bargaining problem to find the stochastically stable states and obtain a new bargaining convention.


\tocless
\newpage{}

\bibliographystyle{chicago}
\bibliography{evolutionary_games}

\newpage{}
\appendix


\scalefont{0.8}

\section*{\large{Appendix: Only for online publication}}

\tableofcontents

\newpage

\renewcommand{\thesection}{A}
\section{Exit problem: one population models \label{appen:one-pop}}


\begin{proof}[\textbf{Proof of Lemma \ref{lem:Pos-feedbacks}}]
(i) Since $c^{(n)}(x, x^{i, k}) = \pi(\bar{m},x)-\pi(k,x)$, we obtain
\begin{align*}
I^{(n)}(\gamma_{2})-I^{(n)}(\gamma_{1})
 & =[\pi(\bar{m},x)-\pi(k,x)+\pi(\bar{m},x^{i,k})-\pi(l,x^{i,k})]\\
 & -[\pi(\bar{m},x)-\pi(k,x)+\pi(\bar{m},x^{\bar{m},k})-\pi(l,x^{\bar{m},k})]\\
 & =\frac{1}{n} \left( [-A_{\bar{m}i}+A_{\bar{m}k}+A_{li}-A_{lk}]-[-A_{\bar{m}\bar{m}}+A_{\bar{m}k}+A_{l\bar{m}}-A_{lk}]\right)\\
 & =\frac{1}{n} \left( A_{\bar{m}\bar{m}}-A_{l\bar{m}}- A_{\bar{m}i}+ A_{li}\right) > 0
\end{align*}
from the \textbf{MBP}. \\
(ii) We find  that
\begin{align*}
   & [I^{(n)}(\zeta_2) - I^{(n)}(\zeta_1)] + [I^{(n)}(\zeta_2) - I^{(n)}(\zeta_3)] \\
  = & [\pi(\bar m, x) - \pi(i, x) + \pi(\bar m, x^{\bar m, i}) - \pi(j, x^{\bar m, i}) + \pi(\bar m, x^{(\bar m,i)(\bar m, j)}) - \pi(i, x^{(\bar m, i)(\bar m, j)})] \\
   - &  [\pi(\bar m, x) - \pi(j, x) + \pi(\bar m, x^{\bar m, j}) - \pi(i, x^{\bar m, j}) + \pi(\bar m, x^{(\bar m, j)(\bar m, i)}) - \pi(i, x^{(\bar m, j)(\bar m, i)}))] \\
   + & [ \pi(\bar m, x) - \pi(i, x) + \pi(\bar m, x^{\bar m, i}) - \pi(j, x^{\bar m, i}) + \pi(\bar m, x^{(\bar m, i)(\bar m, j)}) - \pi(i, x^{(\bar m, i)(\bar m, j)})] \\
    -& [ \pi(\bar m, x) - \pi(i, x) + \pi(\bar m, x^{\bar m, i}) - \pi(i, x^{\bar m, i}) + \pi(\bar m, x^{(\bar m, i)(\bar m, i)}) - \pi(j, x^{(\bar m, i)(\bar m, i)})] \\
   = & [ A_{\bar m i} - A_{\bar m j} + A_{j \bar m } - A_{ji} - A_{i \bar m} + A_{ij}] +  [A_{i \bar m } - A_{i j} + A_{\bar m j} - A_{\bar m i} - A_{j \bar m} + A_{ji}] \\
   = & 0
\end{align*}
From this we obtain the desired results.
\end{proof}

\begin{proof}[\textbf{Proof of Proposition \ref{prop:red}}]
\smallskip

\textbf{Part (i).} In the proof, we suppress the superscript $(n)$. Let  $\gamma=(x_{1},x_{2},\cdots,x_{T})$ be a path in
$\mathcal{G}_{\bar{m}}  \setminus \mathcal{J}_{\bar{m}}$.  We recursively construct a  new path $\tilde{\gamma} \in \mathcal{J}_{\bar{m}}$ with a cost lower than or equal to the cost of $\gamma$.

For this, let  $t$ be the greatest number such that $x_{t+1}=(x_{t})^{i,l}$ with $i \neq \bar{m}, l$. We distinguish several cases.
If $t=T-1$, we consider a new path $\tilde{\gamma}$ obtained by modifying the last transition as follows:
\[
\tilde{\gamma}:=(x_{1},x_{2},\cdots,x_{T-1},(x_{T-1})^{\bar{m},l}).
\]
Then, we have $I(\tilde{\gamma})= I(\gamma)$, and show that the path still exits $D(e_{\bar{m}})$.  To prove this, we only need to show that if $z \notin D(e_{\bar{m}})$ then $z^{\bar{m}, i} \notin D(e_{\bar{m}})$, because this implies that if $(x_{T-1})^{i,l} \notin D(e_{\bar m})$, then $(x_{T-1})^{\bar m, l} \notin D(e_{\bar m})$. Now, suppose that $z \notin D(e_{\bar{m}})$ and that there exists $k$ such that $\pi(\bar{m},z) < \pi(k, z)$. Then, we have
\[
[\pi(k,z^{\bar{m}, i}) - \pi(\bar{m}, z^{\bar{m}, i})] - [\pi(k,z) - \pi(\bar{m}, z)]
=\frac{1}{n}\left( A_{ki} -A_{k\bar{m}} - A_{\bar{m}, i} + A_{\bar{m}, \bar{m}}\right) \ge 0
\]
by Condition \textbf{A}.  Thus, we have $[\pi(k,z^{\bar{m}, i} ) - \pi(\bar{m}, z^{\bar{m}, i})] \ge  [\pi(k,z) - \pi(\bar{m}, z)]  >0$
and so $z^{\bar{m}, i} \notin D(e_{\bar{m}})$.

Now, suppose that $t<T-1$. Then we have $x_{t+1}=(x_{t})^{i,l}$ and $x_{t+2}=(x_t)^{(i,l)(\bar{m},k)}$ for $k \neq \bar m$. Note that $k \neq \bar m$ and $l \neq i$. Now we need to distinguish four cases. \\
Case 1: If $k=i, l = \bar m$, then $x_{t+1} = (x_t)^{i, \bar m}, x_{t+2} = x_t$. Thus, we consider $\tilde{\gamma} =(x_{1},\cdots,x_t, x_{t+2}, \cdots, x_{T})$; clearly, $I(\tilde{\gamma}) \leq I(\gamma)$, since $c(x_{t}, x_{t+1})=0,  c(x_{t+1}, x_{t+2})\geq 0, \textrm{  and } c(x_t, x_{t+2}) =0$. \\
Case 2:  If $k=i, l \neq \bar m$ then  $x_{t+2}=(x_t)^{(i,l)(\bar{m},k)}=(x_t)^{\bar{m},l}$. Again, we consider the path $\tilde{\gamma} =(x_{1},\cdots,x_t, x_{t+2}, \cdots, x_{T})$ and find that $I(\tilde{\gamma}) \leq I(\gamma)$  because we have $c(x_t, x_{t+1}) = c(x_t, x_{t+2}) = \pi(m, x_t) - \pi(l, x_t)$ and $c(x_{t+1}, x_{t+2}) \geq 0 $.  \\
Case 3: If $k \neq i, l = \bar m$, then $x_{t+2} = x_t^{(i, \bar m)(\bar m, k)} = (x_t)^{i, k}$. Again, let $\tilde{\gamma} =(x_{1},\cdots,x_t, x_{t+2}, \cdots, x_{T})$. Then we have $c(x_t, x_{t+1}) = 0$ and
\begin{align*}
	c(x_{t+1}, x_{t+2}) - c(x_t, x_{t+2}) & =c(x_t^{i,l}, x_t^{(i,l)(\bar m ,k)}) - c(x_t, x_t^{(i,k)}) \\
	&= \pi(\bar m, x_t^{i,\bar m}) - \pi(k, x_t^{i,\bar m}) -[ \pi(\bar m,  x_t) - \pi(k, x_t)] \\
	& = \frac{1}{n} (A_{\bar m \bar m } -A_{k \bar m}-[A_{ \bar m i}-  A_{k i}] ) \geq 0
\end{align*}
from the \textbf{MBP}, implying that $ I(\tilde \gamma) \leq I(\gamma)$. \\
Case 4: If $k \neq i,\bar m$ and $l \neq i, \bar m$, then we can apply Lemma \ref{lem:Pos-feedbacks}.  We modify the path by considering  the alternative
transitions, ${\tilde x}_{t+1}= (x_{t})^{\bar{m},l}$ and ${\tilde x}_{t+2} =(x_{t})^{(\bar{m},l)(i,k)}$. If $(x_{t})^{\bar{m},l} \notin D(e_{\bar{m}})$, then we define
\[
\tilde{\gamma}:=(x_{1},x_{2},\cdots,x_{t},(x_{t})^{\bar{m},l})
\]
and because $c(x_t, (x_t)^{\bar m, l})= c(x_t, (x_t)^{i,l})$ and $c(x_{t+1}, x_{t+2}) \geq 0$, we obtain $I(\tilde{\gamma}) \leq I(\gamma)$.
If $(x_{t})^{\bar{m},l}\in D(e_{\bar{m}})$, then we define
\[
\tilde{\gamma}:=(x_{1},x_{2},\cdots,x_{t},(x_{t})^{\bar{m},l},(x_{t})^{(\bar{m},l)(i,k)},\cdots,x_{T}).
\]
to find that $I(\tilde{\gamma}) \leq I(\gamma)$ from Lemma \ref{lem:Pos-feedbacks}.
Proceeding inductively we construct a path $\tilde{\gamma} \in\mathcal{J}_{\bar{m}}$
with a cost lower than or equal to the cost of $\gamma$.

\smallskip
\noindent \textbf{Part (ii).}
We denote by $c(a,a^{i,j,\rho})$ be the cost of a path from $a$ to $a^{i,j,\rho}$ in which agents switch from $i$ to $j$, $\rho$-times consecutively and let $\pi(k, x-y) := \pi(k, x) - \pi(k, y)$ and $\gamma_{a \rightarrow b}$ be a path from $a$ to $b$. We first show the following lemma.
\begin{lem}\label{lem:1st}
    We have the following results.
    \begin{align*}
        & (i) \quad  c(a,a^{\bar m, k, \rho}) - c(b, b^{\bar m, k, \rho})= \rho[(\pi(\bar m, a) -\pi(k, a)) - (\pi(\bar m, b)-\pi(k,b))] \\
       & (ii) \quad \eta[c(a,a^{\bar m, k, \rho}) - c(b, b^{\bar m, k, \rho})] + \rho[c(b^{k, \bar m, \eta},b) - c(a^{k,\bar m, \eta}, a)] =0 \\
       & (iii) \quad \eta [ I(\gamma_{a^{\bar m, k, \rho} \rightarrow b^{\bar m , k, \rho}})-I(\gamma_{a \rightarrow b})] + \rho[I(\gamma_{a^{k,\bar m , \eta} \rightarrow b^{k, \bar m , \eta}})- I(\gamma_{a \rightarrow b} )]=0 \\
       & \text{ where } \gamma_{a^{k,\bar m , \eta} \rightarrow b^{k, \bar m , \eta}}, \gamma_{a^{k,\bar m , \eta} \rightarrow b^{k, \bar m , \eta}}, \text{and   } \gamma_{a \rightarrow b} \text{ consist of the same transitions}.
    \end{align*}
\end{lem}
\begin{proof}
  For (i), we have
  \begin{align*}
      c(a, a^{\bar m, k, \rho}) & = \pi(\bar m, x) - \pi(k, x) +\pi(\bar m, x^{\bar m ,k}) - \pi(k, x^{\bar m , k})+\cdots +\pi(\bar m, x^{\bar m ,k,\rho-1}) - \pi(k, x^{\bar m , k, \rho-1}) \\
     & =\rho( \pi(\bar m , x) - \pi(k,x)) + \frac{\rho(\rho-1)}{2} \frac{1}{n}(-A_{\bar m \bar m} + A_{\bar m k} + A_{k  \bar m } - A_{k k}).
  \end{align*}
  For (ii), first using (i) (by setting $b^{k, \bar m, \eta} =a$), we first find that
  \[
    c(b^{k, \bar m, \eta},b) - c(a^{k,\bar m, \eta}, a)= \eta[(\pi(\bar m, b^{k, \bar m ,\eta}) -\pi(k, b^{k, \bar m ,\eta})-(\pi(\bar{m}, a^{k,\bar m, \eta})-\pi(k,a^{k,\bar m, \eta}))].
 \]
 Then we have
 \begin{align*}
 	&\eta[c(a,a^{\bar m, k, \rho}) - c(b, b^{\bar m, k, \rho})] + \rho[c(b^{k, \bar m, \eta},b) - c(a^{k,\bar m, \eta}, a)]  \\
	= & \eta \rho [ (\pi(\bar m, a) -\pi(k, a)) - (\pi(\bar m, b)-\pi(k,b))]
	   + \eta \rho [(\pi(\bar m, b^{k, \bar m ,\eta}) -\pi(k, b^{k, \bar m ,\eta})-(\pi(\bar{m}, a^{k,\bar m, \eta})-\pi(k,a^{k,\bar m, \eta}))]  = 0	
 \end{align*}
 For (iii), suppose that $(a, b) = (a_1, a_2, \cdots, a_T)$ where $a_T= b$. Then $a_{t+1} = (a_t)^{i_t, l_t}$ for some $i_t, l_t$.  First we find
\begin{align*}
	& \eta[c( {a_t}^{\bar m, k, \rho}, {({a_t}^{\bar m, k, \rho}})^{i_t, l_t}) -c(a_t, {a_t}^{i_t, l_t})]  + \rho[c({a_t}^{k,\bar m , \eta}, ({a_t}^{k, \bar m , \eta})^{i_t, l_t})- c(a_t, {a_t}^{i_t, l_t})] \\
	=& \eta[ \pi(\bar m, {a_t}^{\bar m, k, \rho} -a_t) - \pi(l_t, {a_t}^{\bar m, k, \rho} - a_t)] +\rho[ \pi(\bar m, {a_t}^{k, \bar m, \eta} -a_t) - \pi(l_t, {a_t}^{ k, \bar m, \eta} -a_t)] \\
	=& \frac{1}{n} \eta[ \rho(- A_{\bar m \bar m} + A_{\bar m k}) - \rho( - A_{l_t \bar m} + A_{l_t k})] + \rho[ \eta (-A_{\bar m k} + A_{\bar m \bar m}) - \eta ( - A_{l_t k}+A_{l_t \bar m})] = 0
\end{align*}
We thus find that
	\begin{align*}
		& \eta [ c(a^{\bar m, k, \rho}, b^{\bar m , k, \rho})-c(a,b)] + \rho[c(a^{k,\bar m , \eta}, b^{k, \bar m , \eta})- c(a,b)] \\
		= & \sum_{t=1}^{T-1} \eta[c( {a_t}^{\bar m, k, \rho}, {({a_t}^{\bar m, k, \rho}})^{\bar m, l_t}) -c(a_t, {a_t}^{\bar m, l_t})]  + \rho[c({a_t}^{k,\bar m , \eta}, ({a_t}^{k, \bar m , \eta})^{\bar m, l_t})- c(a_t, {a_t}^{\bar m, l_t})] =  0
	\end{align*}
\end{proof}
\noindent Next, we show the following extended version of comparison principle 2, where we e denote by $(\bar m, k; \eta)$ $\eta$-times consecutive transitions from $\bar m$ to $k$. Also, let $x^{\bar m, k, \eta}$ be a new state induced by the agents' $\eta$-times consecutive switches from $\bar m$ to $k$ from an old state, $x$.
\begin{lem}
\label{prop:straight} Consider the following paths (see Panel C, Figure \ref{fig:comparison-2}):
\[
\begin{array}{ccccccccccc}
\gamma & : & x & \xrightarrow[(\bar{m},k;\eta)]{} & x^{\bar{m},k, \eta} & \xrightarrow[\,\,\,\,\,\,\,\,\,\,\,\,\,\,\,\,]{} & y & \cdots & z & \xrightarrow[(\bar{m},k; \rho)]{} & z^{\bar{m},k, \rho}\\
\gamma' & : & x & \xrightarrow[(\bar{m},k; \eta)]{} & x^{\bar{m},k,\eta} & \xrightarrow[(\bar{m},k; \rho)]{} & x^{(\bar{m},k, \eta)(\bar{m},k, \rho)} & \xrightarrow[\,\,\,\,\,\,\,\,\,\,\,\,\,\,\,\,]{} & y^{\bar{m},k,\rho} & \cdots & z^{\bar{m},k, \rho}\\
\gamma'' & : & x & \xrightarrow[\,\,\,\,\,\,\,\,\,\,\,\,\,\,\,\,\,\,\,]{} & y^{k,\bar{m},\eta} & \cdots & z^{k,\bar{m}, \eta} & \xrightarrow[(\bar{m},k ;\eta)]{} & z & \xrightarrow[(\bar{m},k ;\rho)]{} & z^{\bar{m},k, \rho}
\end{array}
\]
where $\cdots$ denotes the same transitions. Then the following holds:
\[
    \eta [ I^{(n)}(\gamma)- I^{(n)} (\gamma')  ] + \rho [I^{(n)} (\gamma)-I^{(n)}(\gamma'')] =0.
\]
Thus, either
\[
I^{(n)}(\gamma) \geq I^{(n)}(\gamma')\,\,\textrm{ or }\,\, I^{(n)}(\gamma) \geq I^{(n)}(\gamma'')
\]
holds.\end{lem}

\begin{proof}
    We find that
\begin{align*}
 	   & \eta [I(\gamma')-I(\gamma)]+\rho [I(\gamma'')-I(\gamma)] \\
	  = & \underbrace{\eta [c(x^{\bar m, k, \eta}, x^{({\bar m, k, \eta})({\bar m, k, \rho})}) - c(z, z^{\bar m, k , \rho})] + \rho[c(z^{k, \bar m, \eta}, z) - c( x, x^{\bar m, k , \eta})] }_{\text{(i)}}\\
	 + & \underbrace{\eta [c(x^{({\bar m, k, \eta})({\bar m, k, \rho})}, y^{\bar m, k, \rho} ) - c(x^{\bar m, k, \eta}, y)] + \rho[c(x, y^{k, \bar m, \eta}) - c( x^{\bar m, k , \eta},y)]}_{\text{(ii)}} \\
	 + & \underbrace{\eta[ I(\gamma_{y^{\bar m k, \rho} \rightarrow z^{\bar m , k, \rho}}) - I(\gamma_{y \rightarrow z})] + \rho[ I(\gamma_{y^{k, \bar m, \eta} \rightarrow z^{k, \bar m , \eta}}) - I(\gamma_{y\rightarrow z})] }_{\text{(iii)}}
\end{align*}
Then for (i), if we let $a= x^{\bar m, k, \eta}$ and $b=z$ in Lemma \ref{lem:1st} (ii), we have (i)$ =0$. For (ii), if we let $a=x^{(\bar m, k, \eta)}$ and $b=y$ in Lemma \ref{lem:1st} (ii), we have (ii)$=0$. For (iii), if we let $a=y$ and $b=z$ in Lemma \ref{lem:1st} (iii), we have (iii)$=0$.
\end{proof}

\noindent Then, \textbf{Part (ii)} follows from Lemma \ref{prop:straight}. Suppose that $\gamma \in \mathcal{K}_{\bar m}$. Then, by applying Lemma \ref{prop:straight} repeatedly, we collect the same transitions and find $\tilde \gamma \in  \mathcal{K}_{\bar{m}}$ such that $I(\tilde \gamma) \leq I(\gamma) $. Thus we obtain the desired result.
\end{proof}


\noindent \begin{proof}[\textbf{Proof of Proposition \ref{prop:main-approx}}]

Recall that
\begin{align*}
D^{(n)}(e_{\bar{m}}): & =\{x \in \Delta^{(n)}:\,\pi(\bar{m},x)\geq\pi(l,x) \textrm{ for all } l \}
\end{align*}
and let
\begin{equation}\label{eq:cont-basin}
  \bar{D}(e_{\bar{m}}):=\{p\in\Delta:\, \pi(\bar{m},p)\ge \pi(l,p)\,\,\textrm{for\, all \,\,}l\}
\end{equation}
and $\partial \bar D(e_{\bar m})$ be the boundary of $\bar D (e_{\bar m})$.
The following lemma serves to find the continuous version of the cost function, $c(x,x^{i,j})$. Suppose that  $p,q \in \Delta$  with $q=p + \alpha(e_i-e_j)$ for some $\alpha >0$.  If
$p,q \in \bar D(e_{\bar{m}})$,
we define
\begin{equation}
\bar{c}(p,q):=\frac{1}{2}(p_{j}-q_{j})(\pi (\bar m,p+q)- \pi(i,p+q)).
\label{eq:con-cost}
\end{equation}

\begin{lem}\label{lem:con_cost} Let $\gamma=\gamma_{x \to y}$ be a straight-line path between $x^{(n)}$ and $y^{(n)}$  in $D(e_{\bar{m}}) \subset \Delta^{(n)}$
with  $y^{(n)} = x^{(n)} + \frac{M^{(n)}}{n} ( e_i - e_{j})$. Suppose that $x^{(n)} \rightarrow p$ and $y^{(n)} \rightarrow q$ for $p, q \in \Delta$ as $n \rightarrow \infty$.  Then,
\begin{align}\label{eq:cost-n}
\lim_{n \rightarrow \infty} \frac{1}{n} I^{(n)}(\gamma_{x \to y}) =&  \frac{1}{2} (p_j-q_j) ( \pi\left(\bar{m}, p +q \right) - \pi\left(i, p + q\right) ) \nonumber \end{align}
\end{lem}
\begin{proof}  Since the path lies in $D(e_{\bar{m}})$ we have
\begin{equation}\label{eq000}
I^{(n)}(\gamma_{x \to y}) = \sum_{\iota=0}^{M^{(n)}-1}  \left[ \pi\left(\bar{m},x^{(n)} + \frac{\iota}{n} (e_i-e_j)\right)- \pi\left(i ,x^{(n)} + \frac{\iota}{n} (e_i-e_j)\right) \right] \,.
\end{equation}
Now using that $1 + 2+\cdots + K-1= (K-1)K/2$, we obtain
\begin{equation}\label{eq001}
\sum_{t=0}^{M^{(n)}-1} (x^{(n)} + \frac{\iota}{n} (e_i-e_j)) \,=\,  M^{(n)} x^{(n)} + \frac{M^{(n)}(M^{(n)}-1)}{2} \frac{1}{n} (e_i-e_j) \,=\,  M^{(n)} \frac{x^{(n)}+y^{(n)}}{2} - \frac{M^{(n)}}{2} \frac{1}{n}(e_i-e_j) \,.
\end{equation}
By combining equations \eqref{eq000} and \eqref{eq001} and noting that $\frac{M^{(n)}}{n} \rightarrow p_j - q_j$ as $n \rightarrow \infty$, we obtain the desired result.
\end{proof}
\com{
From \eqref{eq:cost-n}, we see that the first term of the cost of a straight-line path is essentially independent of $n$, and by construction if $x^{(n)}, y^{(n)} \in \Delta^{(n)}$ converge to $p$ and $q$ as $n \to \infty$ then $\frac{1}{n} I^{(n)}( \gamma_{x^{(n)}\to y^{(n)}})$ converges to $\bar{c}(p,q)$ in equation \eqref{eq:con-cost}.}
\noindent The expression of costs for continuous paths in Lemma 2 in \citet{SS2014}  is the same as the cost expression in Lemma \ref{lem:con_cost}, since  continuous paths in Lemma 2 in \citet{SS2014}  belong to the special class of paths obtained by comparison principles. Next, we prove the following lemma.
\begin{lem}
\label{lem:con-min}Suppose that $X^{(n)} \subset X $ and  $f:X\rightarrow\mathbb{R}$
is a continuous function that admits a minimum and $f^{(n)}:X\rightarrow\mathbb{R}$. Suppose also that for all $x \in X$, there exists $\{x^{(n)}\}$ such that $x^{(n)}\in X^{(n)}$, $x^{(n)} \rightarrow x$, and $f^{(n)}(x^{(n)}) \rightarrow f(x)$. Then, we have
\[
\min_{x\in X^{(n)}}f^{(n)}(x)\rightarrow\min_{x\in X}f(x)
\]

\begin{proof}
Let $\{x^{(n)}\}_{n}$ be the sequence of minimizers of $\min_{x\in X^{(n)}}f^{(n)}(x)$
and $x^{*}$ be the minimizer of $\min_{x\in X}f(x)$. Suppose that
$f^{(n)}(x^{(n)})$ does not converge to $f(x^{*})$. Then there exist
$\epsilon_{0}>0$ and $\{n_{k}\}$ such that
\begin{equation}
f^{(n_{k})}(x^{(n_{k})}) \geq f(x^{*})+ \epsilon_{0}.
\label{eq:lim1}
\end{equation}
Further, from the hypothesis, we choose $y^{(n)}$ such $y^{(n)} \rightarrow x^*$.  Since $\{ x^{(n)} \}$ is the sequence of minimizers, we have
\begin{equation}
	f^{(n^k)}(y^{(n^k)}) \geq f^{(n^k)}(x^{(n^k)})
\label{eq:lim2}
\end{equation}
Now, by taking $k \rightarrow \infty$ in equations \eqref{eq:lim1} and \eqref{eq:lim2}, we find that  $f(x^*) \geq f(x^*) + \epsilon_0$, which is a contradiction.
\end{proof}
\end{lem}

Now we let $X^{(n)}:=\mathcal{K}^{(n)}_{\bar m}$ and $X =\mathcal{K}_{\bar m}$ and $f^{(n)}=\frac{1}{n} I^{(n)}$ and $f=\bar I$. Then Lemmas \ref{lem:con_cost} and \ref{lem:con-min} show that
    \[
        \lim_{n \rightarrow \infty} \frac{1}{n} \min \{ I^{(n)}(\gamma): \gamma \in \mathcal{K}^{(n)}_{\bar m } \} =  \min \{ \bar I (\zeta): \zeta \in \mathcal{K}_{\bar m } \} = \min \{\omega(t): \zeta(t) \in \mathcal{K}_{\bar m} \}
    \]
\end{proof}

%
%

\begin{proof}[\textbf{Proof of Proposition \ref{prop:main-f-red}}]
The proof of Proposition \ref{prop:main-f-red} follows from Lemmas \ref{lem:mid} and \ref{lem:exit}.
\begin{lem} \label{lem:mid}
    Let $r \in \bar D(e_{\bar m })$. Suppose that
    \[
    w = r + \alpha ( e_k - e_{\bar m}),  \pi(\bar m, w) = \pi(k, w),\text{ and } w \not \in \bar D(e_{\bar m}).
    \]
    Then there exists $j \neq k, \bar m$ and $\beta < \alpha$ such that
    \[
    z:= r+ \beta (e_j - e_{\bar m}), \pi(\bar m, z) = \pi (j, z), \text{ and } \pi(j, r) > \pi(k,r)
    \]
\end{lem}
\begin{proof}
    Since $w \not \in \bar D(e_{\bar m})$, there exists $j \neq k, \bar m$ such that $\pi(j, w) > \pi(\bar m, w)$. Since $\pi(\bar m, r) \geq \pi(j,r)$, there exists $0<\alpha' < \alpha$ such that $\nu=r+\alpha' (e_k - e_{\bar m})$ and
    \[
        \pi(\bar m, \nu) = \pi(j,\nu).
    \]
    Let $o' = r + \alpha (e_{j} - e_{\bar m})$. Note that $o' = \nu-\alpha' (e_k-e_{\bar m}) +\alpha (e_j - e_{\bar m})$. Then
    \begin{align*}
    	\pi(j -\bar m , o')& =  \pi(j -\bar m, -\alpha' (e_k-e_{\bar m}) +\alpha (e_j - e_{\bar m})) \\
	& = - \alpha' \pi(\bar m - j, e_{\bar m} - e_k) + \alpha \pi( \bar m- j,  e_{\bar m}-e_j)  \\
	& > \alpha(\pi(j - \bar m, e_j - e_{\bar m})-\pi(\bar m- j,  e_{\bar m}-e_k) )\\
	& >0
    \end{align*}
Thus since $\pi(\bar m , r) \geq \pi(j, r)$, there exists $z=r + \beta (e_j - e_{\bar m})$ such that $\pi(\bar m, z) = \pi(j, z)$ and $\beta < \alpha$. Next, we show that $\pi(j, r) > \pi(k ,r )$. Suppose that $\pi(k,r) \geq \pi(j,r)$.  Then we find
\[
  \pi(\bar m - j, w) =  \pi(\bar m - j, w) - \pi(\bar m - k, w) = \pi(k, w)- \pi(j, w) = \pi(k-j, r) + \alpha \pi(k-j, e_k - e_{\bar m}) >0
\]
which is a contradiction to the fact that $\pi(\bar m-j, \nu )= \pi(\bar m-j, r + \alpha'(e_k - e_{\bar m})) =0 $ for $\alpha' < \alpha$. Thus, we have $\pi(j,r) > \pi(k,r)$.

\end{proof}

\begin{lem}\label{lem:exit}
  Let $r \in \bar D(e_{\bar m})$ and $q \in \partial \bar D(e_{\bar m})$ and $ q= r + t_L(e_l-e_{\bar m} )$. Suppose that
  \begin{equation} \label{eq:con-k1k2}
    \pi(\bar m, q) = \pi(k_1, q) \text{  and  } \pi(\bar m, q)= \pi(k_2, q).
  \end{equation}
  where $k_1 \neq k_2$. Then there exists $p \in \partial D(e_{\bar m})$ such that $j \neq l, \bar m$ and $p= r + \beta (e_j - e_{\bar m})$, where $0< \beta < t_L$,
  \[
    \pi(\bar m, p) = \pi( j, p )\,\,\text{ and } \,\, c(r,p) < c(r,q).
  \]
\end{lem}
\begin{proof}
  From the condition, $t_L$ is the length of transition from $\bar m$ to $l$, leading to $q$. Because of \eqref{eq:con-k1k2}, we can choose $k \neq l$ such that
  \[
    \pi(\bar m , q) = \pi (k, q).
  \]
  Let $    o:=r + t_L (e_{k} - e_{\bar m}) $. That is, $o$ is the point obtained from $r$ by $t_L$ transitions from $\bar m $ to $k$). Since
\begin{align*}
     & \pi(k -\bar m, r + t_L (e_k - e_{\bar m})) = \pi(k - \bar m , q + t_L (e_{\bar m} - e_{l}) + t_L (e_{k} - e_{\bar m})) \\
    = & t_L \pi (k - \bar m , e_{k}- e_{l}) > 0
  \end{align*}
  hold from the \textbf{MBP}, we have
 \[
    \pi(\bar m, r) \geq \pi( k, r) \text{ and } \pi(\bar m, o) < \pi( k, o)
  \]
  and since the payoff function is linear and the game is a coordination game, there exists $p$ such that $p = r + \alpha (e_{k} - e_{\bar m})$, where $\alpha >0$ and $\pi(\bar m, p) = \pi( k, p)$. Then $o=p+ (t_L - \alpha) (e_k - e_{\bar m})$. Thus
  \begin{align*}
    0 < \pi(k, o) - \pi(\bar m, o) & = \pi(k-\bar m, p+ (t_L - \alpha) (e_k - e_{\bar m}))  \\
    & \leq (t_L - \alpha)\pi(k - \bar m ,e_k - e_{\bar m})
 \end{align*}
  Thus from the \textbf{MBP}, we find $t_L >\alpha$ which implies that $p_k - r_k < q_l - r_l$.
  We divide cases. \\
  \textbf{Step 1}.  Suppose that $p \in \bar D(e_{\bar m})$.
  We also find
  \begin{align*}
        c(r,q) - c(r, p) = & \frac{1}{2} t_L \pi (\bar m - l, r +q )  -\frac{1}{2} (p_{k} -r_{k}) \pi (\bar m - k, r +p )  \\
     \geq & \frac{1}{2} t_L (\pi (\bar m - l, r +q ) -\pi (\bar m - k, r +p ) ) =\frac{1}{2} t_L ( \pi(k, r) - \pi(l, r))
     \\
     = & \frac{1}{2} t_L  \pi(k - l, q + t_L (e_{\bar m} - e_{l}))
     =  \frac{1}{2} t_L \pi(\bar m - l, q) +\frac{1}{2} t_L^2  \pi(k-l,  e_{\bar m} - e_{l} )  > 0
\end{align*}
where we used $\pi(\bar m - l, q) \geq 0$, $\pi(k,q) = \pi(\bar m, q) $, and the \textbf{MBP}. Thus we take $\beta  := \alpha$ and $j:=k$ and obtain the desired result. \\

\noindent \textbf{Step 2}. Suppose that $p \not \in \bar D(e_{\bar m})$. We use Lemma \ref{lem:mid}. By taking $w=p$ and using Lemma \ref{lem:mid}, we find $z$. If $z \in \bar D(e_{\bar m})$, then we set $p'=z$. Otherwise, we apply the same argument using Lemma \ref{lem:mid} and to find $z$ closer to $r$. In this way, we can find $j_1, j_2, \cdots$. Note that no two indices, $j_1, j_2$, are the same since if $j=j_1 =j_2$ then $\pi(\bar m - j, r + \beta_1 (e_{j} -e_{\bar m}))= \pi(\bar m - j_1, r + \beta_1 (e_{j_{1}} -e_{\bar m}))  =\pi(\bar m - j_2, r + \beta_2(e_{j_{2}} - e_{\bar m}) = \pi(\bar m - j, r + \beta_2(e_{j} - e_{\bar m})$. Thus we find $\beta_1 = \beta_2$ which is a contradiction. Since the number of strategies is finite, we can find $z \in  \bar D(e_{\bar m})$. Next, we show that $ j \neq l$. If $j =  l$, $\pi(\bar m, z)= \pi(l, z)$.  Thus, we find that
\begin{align*}
    0  \leq & \pi (\bar m - l, r + t_L (e_l - e_{\bar m})) - \pi(\bar m -l, r + \beta (e_l - e_{\bar m}))  \\
    = & \pi (\bar m - l , (t_L - \beta)(e_l - e_{\bar m})) = (t_L - \beta) (-A_{\bar m \bar m} + A_{\bar m l } + A_{l \bar m} - A_{l l})
\end{align*}
and thus we find $t_L \leq \beta$ which is a contradiction. So we have $j \neq l$.
Then observe that $p_j' - r_j < \beta < t_L$. Then, we compute as follows:
\begin{align*}
	c(r, q) - c(r, p')&=  \frac{1}{2} t_L  \pi(\bar m - l, r +q) - \frac{1}{2} (p_j' -r_j) \pi(\bar m - j, r + p') \\
	&  \geq \frac{1}{2} t_L (\pi(\bar m - l, r +q) -\pi(\bar m - j, r + p') ) =  \frac{1}{2} t_L ( \pi (j, r) - \pi (l, r))  \\
	& >   \frac{1}{2} t_L ( \pi (k, r) - \pi (l, r))  > 0
\end{align*}
Thus, we can take $p=p'$.

\end{proof}
  \noindent Now, let $t^*=((t_1,t_2, \cdots, t_L);(i_1, i_2, \cdots, i_L))$ be the solution to the minimization problem and $(\bar m \rightarrow i_1, \bar m \rightarrow i_2, \cdots, \bar m \rightarrow i_L)$ be the corresponding transitions. Suppose that \eqref{eq:binding-con} does not hold. Then there exists $k_1$ and $k_2$, $k_1 \neq k_2$, such that
  \[
    \pi(\bar m, q(t^*)) = \pi(k_1, q(t^*)) \text{ and }\pi(\bar m, q(t^*)) = \pi(k_2, q(t^*))
  \]
We apply Lemma \ref{lem:exit} and can obtain a lower cost exit path, $s^*$ such that $\omega(s^*) < \omega(t^*)$, which is a contradiction to optimality of $t^*$.
\end{proof}

\begin{proof}[\textbf{Proof of Proposition \ref{prop:main-binding-con}}]
   Suppose that $t_l^* >0$ for some $l \neq k$. To simplify notation, let $q=q(t^*)$ and $t^*=(t_1^*, \cdots, t_K^*)$ and define
  \[
     t^+_{\epsilon} = t^* + \epsilon_k(e_k - e_{\bar m}) - \epsilon_l(e_l - e_{\bar m}), \,\,\,
    t^-_{\epsilon} = t^* - \epsilon_k(e_k - e_{\bar m}) + \epsilon_l(e_l - e_{\bar m})
  \]
  Then, we have
  \begin{align*}
    \pi(\bar m, q(t_\epsilon^+)) - \pi( k, q(t_\epsilon^+)) = & \epsilon_k  \pi(\bar m, k- \bar m) - \epsilon_l \pi(\bar m, l - \bar m ) - \epsilon_k \pi(k, k -  \bar m) + \epsilon_l \pi(k, l- \bar m) \\
    = & - \epsilon_k (A_{\bar m \bar m} - A_{\bar m k} + A_{kk} - A_{k \bar m}) + \epsilon_l (A_{\bar m \bar m} - A_{k \bar m} + A_{\bar m l} - A_{k l}) \\
    \pi(\bar m, q(t_\epsilon^-)) - \pi( k, q(t_\epsilon^-)) = &  \epsilon_k (A_{\bar m \bar m} - A_{\bar m k} + A_{kk} - A_{k \bar m}) - \epsilon_l (A_{\bar m \bar m} - A_{k \bar m} + A_{\bar m l} - A_{k l})
  \end{align*}
  and similarly, for $j \neq k$, we find that
  \begin{align*}
    \pi(\bar m, q(t_\epsilon^+)) - \pi( j, q(t_\epsilon^+)) = & \pi(\bar m, q) - \pi(j , q) \\
     & + \epsilon_k  \pi(\bar m, k- \bar m) - \epsilon_l \pi(\bar m, l - \bar m ) - \epsilon_k \pi(j, k -  \bar m) + \epsilon_l \pi(j, l- \bar m) \\
    = &\pi(\bar m, q) - \pi(j , q)\\
             & - \epsilon_k (A_{\bar m \bar m} - A_{\bar m k} + A_{jk} - A_{j \bar m}) + \epsilon_l (A_{\bar m \bar m} - A_{j \bar m} + A_{\bar m l} - A_{j l}) \\
    \pi(\bar m, q(t_\epsilon^-)) - \pi( j, q(t_\epsilon^-)) = &  \pi(\bar m, q) - \pi(j ,q ) \\
    &   \epsilon_k (A_{\bar m \bar m} - A_{\bar m k} + A_{jk} - A_{j \bar m}) - \epsilon_l (A_{\bar m \bar m} - A_{j \bar m} + A_{\bar m l} - A_{j l})
  \end{align*}
  Thus, we can choose small $\epsilon_k, \epsilon_l>0$ such that
  \begin{align*}
    & \pi(\bar m, q(t_\epsilon^+)) = \pi(k, q(t_\epsilon^+)), \text{ and } \pi(\bar m, q(t_\epsilon^+)) > \pi(j, q(t_\epsilon^+)) \text{ for all } l \neq k\\
    & \pi(\bar m, q(t_\epsilon^-)) = \pi(k, q(t_\epsilon^-)), \text{ and } \pi(\bar m, q(t_\epsilon^-)) > \pi(j, q(t_\epsilon^-)) \text{ for all } l \neq k,
  \end{align*}
  which show that $t^+_{\epsilon}$ and $t^-_{\epsilon}$ both satisfy the constraints. Recall
  \[
    H_{i,j:k} := (A_{ii}- A_{ji}) - (A_{ik} - A_{jk}).
  \]
Then we find that
  \begin{align*}
     \text{  If $t_l$ is ahead of  $t_k$}, \,\,\, &(\omega(t^+_\epsilon) - \omega(t))-(\omega(t) - \omega(t^-_\epsilon))       = - \epsilon_l ^2 \pi(\bar m - l, \bar m -l) + 2 \epsilon_l \epsilon_k \pi(\bar m - k, \bar m-l) - \epsilon_k^2 \pi(\bar m - k , \bar m - k ) \\
     \text{  If $t_k$ is ahead of  $t_l$}, \,\,\,  &(\omega(t^+_\epsilon) - \omega(t))-(\omega(t) - \omega(t^-_\epsilon))       = - \epsilon_l ^2 \pi(\bar m - l, \bar m -l) + 2 \epsilon_l \epsilon_k \pi(\bar m - l, \bar m-k) - \epsilon_k^2 \pi(\bar m - k , \bar m - k )
  \end{align*}
  Thus, we find that
  \begin{align*}
      & (\omega(t^+_\epsilon) - \omega(t))-(\omega(t) - \omega(t^-_\epsilon))
     =       - H_{\bar m k: k} \epsilon_k^2 + 2  \max \{ H_{\bar m k: l}, H_{\bar m l: k}\}  \epsilon_k \epsilon_l- H_{\bar m l: l} \epsilon_l^2 \\
      & \leq - H_{\bar m k: k} \epsilon_k^2 + 2 \sqrt{H_{\bar m k: k} }\sqrt{H_{\bar m l: l} } \epsilon_k \epsilon_l  - H_{\bar m l: l}  \epsilon_l^2 \leq - (\sqrt{H_{\bar m k: k}}\epsilon_k -
     \sqrt{H_{\bar m l: l}} \epsilon_l)^2 < 0
  \end{align*}
  where we use
  \[
      \max \{ H_{\bar m k: l}, H_{\bar m l: k}\}  < H_{\bar m k: k},\,\,\, \max \{ H_{\bar m k: l}, H_{\bar m l: k}\}  < H_{\bar m l: l}.
  \]
  from \textbf{MBP}.
  This shows that either $\omega(t^+_\epsilon) < \omega(t)$ or $\omega(t) > \omega(t^-_\epsilon)$ holds, a contradiction to the optimality of $t$.
\end{proof}

\com{
\textbf{Need to change this: From Proposition \ref{prop:binding-con}, only one binding constraint exists in \eqref{eq:var-constraints}. Proposition \ref{prop:f-red} further shows that only one kind of transition to a specific strategy occurs in the minimal escape path.}
}

\begin{proof}[\textbf{Proof of Theorem \ref{thm:escape}}]

Let $t^*$ be the solution to the minimization problem:
  \[
    \min\{ \omega(t): \zeta(t) \in \bar{\mathcal{K}}_{\bar m} \}.
  \]
Propositions \ref{prop:main-binding-con} and \ref{prop:main-f-red}  show that there exists $k$ such that $t_k^* > 0$ and $t^*_l=0$ for all $l \neq k$ and Theorem  \ref{thm:escape}  follows immediately from this and Proposition \ref{prop:main-approx}.
\end{proof}

\renewcommand{\thesection}{B}

\section{Exit problem: two-population models} \label{appen:2p-exit}

The following lemma is analogous to Lemma \ref{lem:Pos-feedbacks}, which shows that it always costs less (or the same) to first switch from strategy $\bar m$, than from other strategies.
\begin{lem}
\label{lem:2p-pos1}Suppose that the \textbf{WBP} holds.
\begin{align*}
c^{(n)}(x^{\beta,\bar{m},k},x^{(\beta,\bar{m},k)(\alpha,j,h)})-c^{(n)}(x^{\beta,i,k},x^{(\beta,i,k)(\alpha,j,h)}) & = -A_{\bar{m}\bar{m}}^{\alpha}+A_{h\bar{m}}^{\alpha}+A_{\bar{m}i}^{\alpha}-A_{hi}^{\alpha} \leq 0\\
c^{(n)}(x^{\alpha,\bar{m},k},x^{(\alpha,\bar{m},k)(\beta,j,h)})-c^{(n)}(x^{\alpha,i,k},x^{(\alpha,i,k)(\beta,j,h)}) & =- A_{\bar{m}\bar{m}}^{\beta}+A_{\bar{m}h}^{\beta}+A_{i\bar{m}}^{\beta}-A_{ih}^{\beta}\leq 0.
\end{align*}
\end{lem}
\begin{proof}
These are immediate from the definition.
\end{proof}
Proposition \ref{prop:2p-pos2} shows that Lemma \ref{lem:2p-pos1} can be extended to arbitrary paths. We use Proposition \ref{prop:2p-pos2} to show how to remove the transitions from $i\neq\bar{m}$ in a given path to achieve a lower cost. In Proposition \ref{prop:2p-pos2}, $(\beta, i, k)$, for example, refers to a transition by a $\beta$-agent from strategy $i$ to $k$.
\begin{prop}
\label{prop:2p-pos2}Suppose that the \textbf{WBP} holds. We consider
two paths:
\begin{align*}
\gamma_{1}: & x\xrightarrow[(\beta,i,k)]{}x^{(1)}\xrightarrow[(\alpha,j_{1},k_{1})]{}x^{(2)}\xrightarrow[(\alpha,j_{2},k_{2})]{}x^{(3)}\cdots x^{(L-1)}\xrightarrow[(\alpha,j_{L},k_{L})]{}x^{(L)}\xrightarrow[(\beta,\bar{m},l)]{}y\\
\gamma_{2}: & x\xrightarrow[(\beta,\bar{m},k)]{}y^{(1)}\xrightarrow[(\alpha,j_{1},k_{1})]{}y^{(2)}\xrightarrow[(\alpha,j_{2},k_{2})]{}y^{(3)}\cdots y^{(L-1)}\xrightarrow[(\alpha,j_{L},k_{L})]{}y^{(L)}\xrightarrow[(\beta,i,l)]{}y
\end{align*}
Then, we have $I^{(n)}(\gamma_{1})\geq I^{(n)}(\gamma_{2})$ and a similar statement
holds for a path with transitions of $\alpha$ agents from $i$
to $k$ and $\bar{m}$ to $l$ and transitions of $\alpha$ agents from $\bar m $ to $k$ and from $i$ to $l$.\end{prop}
\begin{proof}
We find that
\begin{align*}
I^{(n)}(\gamma_{1}) & =c^{(n)}(x,x^{\beta,i,k})+c^{(n)}(x^{\beta,i,k},x^{(\beta,i,k)(\alpha,j_{1},k_{1})})+c^{(n)}(x^{\beta,i,k},x^{(\beta,i,k)(\alpha,j_{2},k_{2})}) \\
 & +\cdots c^{(n)}(x^{\beta,i,k},x^{(\beta,i,k)(\alpha,j_{L},k_{L})})
  +c^{(n)}(x^{(L)},(x^{(L)})^{(\beta,\bar{m},l)}). \\
 I^{(n)}(\gamma_{2}) & =c^{(n)}(x,x^{\beta,\bar{m},k})+c^{(n)}(x^{\beta,\bar{m},k},x^{(\beta,\bar{m},k)(\alpha,j_{1},k_{1})})+c^{(n)}(x^{\beta,\bar{m},k},x^{(\beta,\bar{m},k)(\alpha,j_{2},k_{2})})\\
 & +\cdots c^{(n)}(x^{\beta,\bar{m},k},x^{(\beta,\bar{m},k)(\alpha,j_{L},k_{L})}) +c^{(n)}(x^{(L)},(x^{(L)})^{(\beta,i,l)})
\end{align*}
from the fact that $c^{(n)}(x^{(l)}, (x^{(l)})^{\alpha, j_l, k_l})=c^{(n)}(x^{\beta,i,k},x^{(\beta,i,k)(\alpha,j_{l},k_{l})})$ for $l=2,\cdots, L-1$ and $c(y^{(l)}, (y^{(l)})^{\alpha, j_l, k_l})=c^{(n)}(y^{\beta,\bar m ,k},x^{(\beta,\bar m ,k)(\alpha,j_{l},k_{l})})$ for $l=2,\cdots, L-1$ (see Lemma \ref{lem:irrelev}).
Observe that $c^{(n)}(x,x^{\beta,\bar{m},k})=c^{(n)}(x,x^{\beta,i,k})$ and $c^{(n)}(x^{(L)},(x^{(L)})^{(\beta,i,l)})=c^{(n)}(x^{(L)},(x^{(L)})^{(\beta,\bar{m},l)})$.
Then by applying Lemma 2 successively, we obtain the desired result.
\end{proof}
We can also collect the same transitions as follows, analogously to Proposition \ref{prop:straight}. We also denote by $(\beta, \bar m, k; \eta)$ the consecutive transitions of $\beta$-agent from $\bar m$ to $k$ $\eta$-times.
\begin{prop}
\label{prop:2p-straight} Consider the
following paths:
\[
\begin{array}{ccccccccccc}
\gamma & : & x & \xrightarrow[(\beta,\bar{m},k; \eta)]{} & x^{\beta,\bar{m},k; \eta} & \xrightarrow[\,\,\,\,\,\,\,\,\,\,\,\,\,\,\,\,]{} & y & \cdots & z & \xrightarrow[(\beta,\bar{m},k; \rho)]{} & z^{\beta,\bar{m},k; \rho}\\
\gamma' & : & x & \xrightarrow[(\beta,\bar{m},k; \eta)]{} & x^{\beta,\bar{m},k;\eta} & \xrightarrow[(\beta,\bar{m},k; \rho)]{} & x^{(\beta,\bar{m},k;\eta)(\beta,\bar{m},k;\rho)} & \xrightarrow[\,\,\,\,\,\,\,\,\,\,\,\,\,\,\,\,]{} & y^{\beta,\bar{m},k; \rho} & \cdots & z^{\beta,\bar{m},k; \rho}\\
\gamma'' & : & x & \xrightarrow[\,\,\,\,\,\,\,\,\,\,\,\,\,\,\,\,\,\,\,]{} & y^{\beta,k\bar{,m};\eta} & \cdots & z^{\beta,k,\bar{m};\eta} & \xrightarrow[(\beta,\bar{m},k; \eta) ]{} & z & \xrightarrow[(\beta,\bar{m},k; \rho)]{} & z^{\beta,\bar{m},k; \rho}
\end{array}
\]
where $\cdots$ denotes the same transitions. Then either
\[
I^{(n)}(\gamma) \geq I^{(n)}(\gamma'), \,\,\text{ or } I^{(n)}(\gamma) \geq I^{(n)}(\gamma'')
\]
holds. A similar statement
holds for a path involving transitions of $\alpha$ agents' transitions.\end{prop}
\noindent \emph{Proof.}
We start with the following lemma.
\begin{lem} \label{lem:irrelev}
We have the following results:
\begin{align*}
    c^{(n)}(x, x^{\alpha, i, j}) = c^{(n)}(z, z^{\alpha, i, j}) \text{ for all } x_\beta = z_\beta \\
    c^{(n)}(x, x^{\beta, i, j}) = c^{(n)}(z, z^{\beta, i, j}) \text{ for all } x_\alpha = z_\alpha
\end{align*}
\begin{proof}
  This is immediate from the definition.
\end{proof}
\end{lem}
Next we show the following lemma.

\begin{lem} \label{lem:bet-comp}
We have the following results:\\
\[
    \eta[ c^{(n)} (a^{\beta, \bar m, k, \rho} , b^{\beta, \bar m, k,\rho}) -c^{(n)}(a,b)] + \rho [ c^{(n)}(a^{\beta, k,\bar m, \eta}, b^{\beta, k, \bar m, \eta}) -c^{(n)}(a, b)]=0
\]
\end{lem}
\begin{proof}
  Suppose that $(a,b)=(a_1, a_2, \cdots, a_T)$ where $a_T = b$.  Suppose that $a_{t+1} = (a_t)^{\beta, i_t, l_t}$.
  Then by applying Lemma \ref{lem:irrelev}, we obtain
  \[
    \eta[ c^{(n)} (a_t^{\beta, \bar m, k, \rho} , (a_t^{\beta, \bar m, k,\rho})^{\beta, i_t, l_t}) -c^{(n)}(a_t,a_t^{\beta, i_t, l_t})] + \rho [ c^{(n)}(a_t^{\beta, k,\bar m, \eta}, (a_t^{\beta, k, \bar m, \eta})^{\beta, i_t, l_t} ) -c^{(n)}(a_t, a_t^{\beta, i_t, l_t})]=0
  \]
  We next suppose that $a_{t+1} = (a_t)^{\alpha, i_t, l_t}$.
  \begin{align*}
     & \eta[ c^{(n)} (a_t^{\beta, \bar m, k, \rho} , (a_t^{\beta, \bar m, k,\rho})^{\alpha, i_t, l_t}) -c^{(n)}(a_t,a_t^{\alpha, i_t, l_t})] + \rho [ c^{(n)}(a_t^{\beta, k,\bar m, \eta}, (a_t^{\beta, k, \bar m, \eta})^{\alpha, i_t, l_t} ) -c^{(n)}(a_t, a_t^{\alpha, i_t, l_t}) \\
     = & \eta[ \pi_\alpha(\bar m, a_t^{\beta, \bar m, k, \rho}) - \pi_\alpha(l_t, a_t^{\beta, \bar m, k, \rho}) -\pi_\alpha(\bar m, a_t) +\pi_\alpha(l_t, a_t)] \\
      & +\rho[ \pi_\alpha(\bar m, a_t^{\beta, k, \bar m,  \eta}) - \pi_\alpha(l_t, a_t^{\beta, k, \bar m,  \eta}) -\pi_\alpha(\bar m, a_t) +\pi_\alpha(l_t, a_t)] \\
      = &0
    \end{align*}
    Thus we find
    \begin{align*}
		& \eta [ I^{(n)}(\gamma_{a^{\beta, \bar m, k, \rho} \rightarrow b^{\beta, \bar m , k, \rho}})-I^{(n)}(\gamma_{a \rightarrow b})] + \rho[I^{(n)}(\gamma_{a^{\beta, k,\bar m , \eta} \rightarrow b^{\beta, k, \bar m , \eta}})- I(\gamma_{a \rightarrow b})] \\
		= & \sum_{t=1}^{T-1} \eta[c^{(n)}( {a_t}^{\beta, \bar m, k, \rho}, {({a_t}^{\beta, \bar m, k, \rho}})^{i_t, l_t}) -c^{(n)}(a_t, {a_t}^{\beta, i_t, l_t})]  + \rho[c^{(n)}({a_t}^{\beta, k,\bar m , \eta}, ({a_t}^{\beta, k, \bar m , \eta})^{i_t, l_t})- c^{(n)}(a_t, {a_t}^{i_t, l_t})] \\
		= & 0
	\end{align*}
\end{proof}

\begin{lem} \label{lem:bet-comp}
We have the following results:\\
(i) $\eta [ c^{(n)}(x^{\beta, \bar m, k; \eta}, x^{(\beta, \bar m, k; \eta),(\beta, \bar m, k; \rho) }) -c^{(n)}(z, z^{(\beta, \bar m, k; \rho)})] + \rho[ c^{(n)}(z^{(\beta, k, \bar m;  \eta)}, z)-c^{(n)}(x, x^{\beta, \bar m , k; \eta})]=0$ \\
(ii) $\eta[ c^{(n)}(x^{(\beta, \bar m, k; \eta),(\beta, \bar m, k; \rho) }, y^{\beta, \bar m , k; \rho}) - c^{(n)} (x^{\beta, \bar m ,k ; \eta}, y)] + \rho [c^{(n)}(x, y^{\beta, k, \bar m; \eta}) - c^{(n)}(x^{\beta, \bar m, k , \eta}, y)] = 0$  \\
(iii) $\eta[ I^{(n)}(\gamma_{y^{\beta, \bar m, k; \rho} \rightarrow z^{\beta, \bar m, k \rho}}) - I^{(n)}(\gamma_{y \rightarrow z})] + \rho [ I^{(n)}(\gamma_{y^{\beta, k, \bar m, \eta} \rightarrow z^{\beta, k, \bar m, \eta}})-I^{(n)}(\gamma_{y \rightarrow z})] =0$
\end{lem}

\begin{proof}
    (i) By applying Lemma \ref{lem:irrelev}, we find that
    \begin{align*}
       & \eta [ c^{(n)}(x^{\beta, \bar m, k; \eta}, x^{(\beta, \bar m, k; \eta),(\beta, \bar m, k; \rho) }) -c^{(n)}(z, z^{(\beta, \bar m, k; \rho)})] + \rho[ c^{(n)}(z^{(\beta, k, \bar m;  \eta)}, z)-c^{(n)}(x, x^{\beta \bar m , k; \eta})] \\
      = & \eta [ c^{(n)}(x, x^{(\beta, \bar m, k; \rho) } -c^{(n)}(z, z^{(\beta, \bar m, k; \rho)})] + \rho[ c^{(n)}(z, z^{(\beta, \bar m, k;  \eta)})-c^{(n)}(x, x^{\beta, \bar m , k; \eta})] \\
      = & \eta \rho [ \pi_\beta (\bar m, x) - \pi_\beta (k, x) - \pi_\beta (\bar m, z) + \pi_\beta (k, z)] + \rho \eta[ \pi_\beta(\bar m, z) - \pi_\beta(k, z) - \pi_\beta(\bar m ,x) + \pi_\beta (k, x)] \\
      = & 0
    \end{align*}
       (ii) follows from by letting $a:=x^{\beta, \bar m ,k ;\eta}$ and $b:=y$ in Lemma \ref{lem:bet-comp} and (iii) follows from by letting $a:=y$ and $b:=z$ in Lemma \ref{lem:bet-comp}.
\end{proof}

\begin{proof}[Proof of Proposition \ref{prop:2p-straight}]
We find that
\begin{align*}
 	   & \eta [I^{(n)}(\gamma')-I^{(n)}(\gamma)]+\rho [I^{(n)}(\gamma'')-I^{(n)}(\gamma)] \\
	  = & \underbrace{\eta [ c^{(n)}(x^{\beta, \bar m, k; \eta}, x^{(\beta, \bar m, k; \eta),(\beta, \bar m, k; \rho) }) -c^{(n)}(z, z^{(\beta, \bar m, k; \rho)})] + \rho[ c^{(n)}(z^{(\beta, k, \bar m;  \eta)}, z)-c^{(n)}(x, x^{\beta, \bar m , k; \eta})]}_{\text{(i)}}\\
	 + & \underbrace{\eta[ c^{(n)}(x^{(\beta, \bar m, k; \eta),(\beta, \bar m, k; \rho) }, y^{\beta, \bar m , k; \rho}) - c^{(n)} (x^{\beta, \bar m ,k ; \eta}, y)] + \rho [c^{(n)}(x, y^{\beta, k, \bar m; \eta}) - c^{(n)}(x^{\beta, \bar m, k , \eta}, y)]}_{\text{(ii)}} \\
	 + & \underbrace{\eta[ I^{(n)}(\gamma_{y^{\beta, \bar m, k; \rho} \rightarrow z^{\beta, \bar m, k \rho}}) - I^{(n)}(\gamma_{y \rightarrow z})] + \rho [ I^{(n)}(\gamma_{y^{\beta, k, \bar m, \eta} \rightarrow z^{\beta, k, \bar m, \eta}})-I^{(n)}(\gamma_{y\rightarrow z})] }_{\text{(iii)}}
\end{align*}
and Lemma \ref{lem:bet-comp} (i), (ii), and (iii) show the desired result.
\end{proof}

We also define $\mathcal{J}_{\bar m}^{(n)}$ and $\mathcal{K}_{\bar{m}}^{(n)}$ analogously to equations \eqref{eq:path-J} and \eqref{eq:path-K}. That is, $\mathcal{J}_{\bar m}^{(n)}$ is the set of all paths in which all the transitions are from strategy $\bar m$ and $\mathcal{K}_{\bar m}^{(n)}$ is the set of all paths consisting of consecutive transitions from $\bar m$ to some other strategy. From Propositions \ref{prop:2p-pos2} and \ref{prop:2p-straight}, we next show that the minimum transition cost path $\gamma$ involves only transitions from $\bar{m}$.
\begin{prop}
\label{prop:2p-fromm}Suppose that the \textbf{WBP} holds.

(i) We have
\[
\min\{I^{(n)}(\gamma):\gamma\in\mathcal{G}_{\bar{m}}^{(n)}\}=\min\{I^{(n)}(\gamma):\gamma\in\mathcal{J}_{\bar{m}}^{(n)}\}.
\]

(ii) We have
\[
\min\{I^{(n)}(\gamma):\gamma\in\mathcal{G}_{\bar{m}}^{(n)}\}=\min\{I^{(n)}(\gamma):\gamma\in\mathcal{K}_{\bar{m}}^{(n)}\}.
\]
\end{prop}
\begin{proof}
For the proof, we suppress the superscript $(n)$. Part (i). Let $\gamma\in\mathcal{G}_{\bar{m}}\backslash\mathcal{J}_{\bar{m}}$.
Let the last transition of $\gamma$ be from $z$ to $z^{\beta,i,l}$
for some $i\neq\bar{m}.$ Since $c^{(n)}(z,z^{\beta,i,l})=c^{(n)}(z,z^{\beta,\bar{m},l})$,
by modifying the last transition from $z^{\beta, i, l}$ to $z^{\beta,\bar{m},l}$
the cost will not be changed. Now, suppose that $x$ is the last state from which a transition occurs from $i\not=\bar{m}$ in the modified path (see $\gamma_1$ in Proposition \ref{prop:2p-pos2}). Then, by applying Proposition \ref{prop:2p-pos2}, we obtain the new path whose last transition is from $i\not=\bar{m}$ (see $\gamma_2$ in Proposition \ref{prop:2p-pos2}). By changing this last transition again, we can obtain a new modified path. In this way, we can remove all $\beta$-agents'
transitions from $i\not=\bar{m}$. Similarly, we can also remove all $\alpha$-agents'
transitions from $i\neq\bar{m}$ using the corresponding part for
 $\alpha$ agents in Proposition \ref{prop:2p-pos2}. Thus, we can
obtain the desired results. Part (ii) immediately follows from Proposition
\ref{prop:2p-straight}.
\end{proof}

Next, we consider the continuous limit. For this, we define
a cost function $\bar{c}(\mathbf{p},\mathbf{q})$, for $\mathbf{p}=(p_\alpha,p_\beta), \mathbf{q}=(q_\alpha, q_\beta) \in \Delta_\alpha \times \Delta_\beta$. Let $\mathbf{q}=\mathbf{p}+(\rho(e_{i}^{\alpha}-e_{j}^{\alpha}),0)$
or $\mathbf{q}=\mathbf{p}+(0, \rho(e_{i}^{\beta}-e_{j}^{\beta}))$ for some $\rho>0$. If
$\mathbf{p},\mathbf{q}\in\bar{D}(e_{\bar{m}})$,
\[
\bar c(\mathbf{p},\mathbf{q})=(p_{\alpha, j}-q_{\alpha, j})(\pi_{\alpha}(\bar{m},p)-\pi_{\alpha}(j,p)) \text{  or  } \bar c(\mathbf{p},\mathbf{q})=(p_{\beta, j}-q_{\beta, j})(\pi_{\beta}(\bar{m},p)-\pi_{\beta}(j,p)).
\]

We similarly define $\bar{\mathcal{K}}_{\bar{m}}$ as in the one population model  and from $\zeta=\zeta(t)\in\mathcal{K}_{\bar{m}}$, where $t=((t^{\alpha},t^{\beta});(i^\alpha, j^\beta))=(((t_{1}^{\alpha},\cdots,t_{K}^{\alpha}),(t_{1}^{\beta},\cdots,t_{K}^{\beta}));(i^\alpha_1,\cdots, i^\alpha_K);(j^\beta_1,\cdots, j^\beta_K))$
and define $\omega(t)=\sum_{s=0}^{K-1}\bar{c}(\mathbf{p}^{(s)},\mathbf{p}^{(s+1)}).$
Then, we have the following lemma.
\begin{lem}
\label{lem:2p-affine}Let $\bar{t}^{\beta}$, $i^\alpha$, $j^\alpha$ be fixed. Then $\omega(\cdot,\bar{t}^{\beta})$
is affine. A similar statement holds for the case where $\bar{t}^{\alpha}$
is fixed.\end{lem}
\begin{proof}
Suppose that $t_{i}^{\alpha}$
is associated with $\alpha$ agents' transitions from $\bar{m}$ to
$i$. Similarly, $t_{j}^{\beta}$ is associated with $\beta$ agents'
transitions from $\bar{m}$ to $j.$ Let $\mathbf{p}$ be the state from which the
transitions represented by $t_{i}^{\alpha}$ start.
{} Then we find that
\begin{align*}
\frac{\partial\omega}{\partial t_{i}^{\alpha}}=(\pi_{\alpha}(\bar{m},p_{\beta})-\pi_{\alpha}(i,p_{\beta})) & +\bar{t}_{i}^{\beta}(-A_{ii}^{\beta}+A_{i\bar{m}}^{\beta}-A_{\bar{m}\bar{m}}^{\beta}+A_{\bar{m}i}^{\beta})\\
 & +\sum_{j\neq i}\bar{t}_{j}^{\beta}(-A_{ij}^{\beta}+A_{i\bar{m}}^{\beta}-A_{\bar{m}\bar{m}}^{\beta}+A_{\bar{m}j}^{\beta})
\end{align*}
and observe that $\pi_{\alpha}(\bar{m},p_{\beta})-\pi_{\alpha}(i,p_{\beta})$
depends only on $\bar{t}^{\beta}$; this shows that $\omega(\cdot,\bar{t}^{\beta})$
is affine.
\end{proof}
Thus, we similarly consider
\[
\min\{\omega(t):\zeta(t)\in\mathcal{K}_{\bar{m}}\}.
\]

Using the characterization that $\omega$ is affine, we show that
if $t_{i}^{\alpha^{*}}>0$ in an optimal path, then $\pi_{\beta}(\bar{m},\mathbf{q}^{*}(t^{*}))=\pi_{\beta}(i,\mathbf{q}^{*}(t^{*}))$
at the exit point $\mathbf{q}^{*}(t^{*})$, where $t_i^{\alpha^*}$ denotes the transition by an $\alpha$-agent from strategy $\bar m$ to $i$.
\begin{prop}
\label{prop:2p-binding}Suppose that Condition \textbf{B} holds. Then, there
exists $\zeta(t^{*})\in\mathcal{K}_{\bar{m}}$ such that $\omega(t^{*})=\min\{\omega(t):\zeta(t)\in\mathcal{K}_{\bar{m}}\}$
and if $t_{i}^{\alpha^{*}}>0$, then $\pi_{\beta}(\bar{m},\mathbf{q}(t^{*}))=\pi_{\beta}(i,\mathbf{q}(t^{*}))$
and if $t_{j}^{\beta^{*}}>0$, then $\pi_{\alpha}(\bar{m},\mathbf{q}(t^{*}))=\pi_{\alpha}(j,\mathbf{q}(t^{*}))$,
where $\mathbf{q}(t^{*})$ is the end state of $\zeta(t^{*})$. \end{prop}
\begin{proof}
Let $t^{*}$ be given such that $\omega(t^{*})=\min\{\omega(t):\zeta(t)\in\mathcal{K}_{\bar{m}}\}$.
Suppose that $t_{i}^{\alpha^{*}}>0$. The other case follows similarly.
Let $\bar t^\alpha_i$ such that
\[
    \pi_\beta( \bar m, (1- \bar t_i^\alpha) e^\alpha_{\bar m} + \bar t_i^\alpha e^\alpha_i) = \pi_\beta (i, (1-\bar t_i^\alpha) e^\alpha_{\bar m}+\bar t_i^\alpha e^\alpha_{\bar m})
\]
Then, we have
\begin{align}\label{eq:prof-m1}
   & \pi_\beta(i, q_\alpha(t^{\alpha^{*}})) - \pi_\beta(\bar m, q_\alpha(t^{\alpha^{*}}))  \nonumber \\
  = & \pi_\beta (i-\bar m, (1-t_i^{\alpha^{*}})e_{\bar m}^\alpha + t^{\alpha^{*}}_i e_i^\alpha) + \sum_{l \neq i} t^{\alpha^{*}}_l \pi_\beta ( i - \bar m, e_l^\alpha - e_{\bar m}^\alpha ) \nonumber \\
  = & \pi_\beta (i-\bar m, (1-t_i^{\alpha^{*}})e_{\bar m}^\alpha + t^{\alpha^{*}}_i e_i^\alpha) + \sum_{l \neq i} t^{\alpha^{*}}_l ( A^\beta_{li} - A^\beta_{l \bar m} - A^\beta_{\bar m i} + A^\beta_{\bar m \bar m }).
\end{align}
Now, we have two cases: \\
\smallskip

\noindent \textbf{Case 1}: $t^{\alpha^*}_i = \bar t^\alpha_i$. \\
 Since $t^* \in \mathcal{K}_{\bar{m}}$, $\pi_\beta(i, q_\alpha(t^{\alpha^{*}})) - \pi_\beta(\bar m, q_\alpha(t^{\alpha^{*}})) \leq 0$, the second term in \eqref{eq:prof-m1} ($\sum_{l \neq i} t^{\alpha^{*}}_l ( A^\beta_{li} - A^\beta_{l \bar m} - A^\beta_{\bar m i} + A^\beta_{\bar m \bar m })$) is non-positive. Also, the \textbf{WBP} implies that the same term is non-negative, and hence zero. Thus, we have $\pi_\beta(i, q_\alpha(t^{\alpha^{*}})) = \pi_\beta(\bar m, q_\alpha(t^{\alpha^{*}}))$, which is the desired result. \\
\smallskip

\noindent  \textbf{Case 2}:  $0< t^{\alpha^*}_i < \bar t^\alpha_i$. \\
Suppose that
\begin{equation}\label{eq:contra}
  \pi_{\beta}(\bar{m},q_{\alpha}(t^{\alpha^{*}}))>\pi_{\beta}(i,q_{\alpha}(t^{\alpha^{*}})).
\end{equation}
and
\begin{equation}
\pi_{\beta}(\bar{m},q_{\alpha}(t^{\alpha^{*}}))=\pi_{\beta}(j_{1},q_{\alpha}(t^{\alpha^{*}})),\pi_{\beta}(\bar{m},q_{\alpha}(t^{\alpha^{*}}))=\pi_{\beta}(j_{2},q_{\alpha}(t^{\alpha^{*}})),\cdots,\pi_{\beta}(\bar{m},q_{\alpha}(t^{\alpha^{*}}))=\pi_{\beta}(j_{L},q_{\alpha}(t^{\alpha^{*}})),
\label{eq:2p-binding-const}
\end{equation}
where the other constraints for $\pi_{\beta}$ are non-binding. To reach $q_\alpha(t^{\alpha^*})$, there are transitions, $\bar m \rightarrow j_1, \bar m \rightarrow j_2, \cdots, \bar m \rightarrow j_L$ and thus
\[
    q_\alpha(t^{\alpha^*}) = e_{\bar m}^\alpha + \sum_{l=1}^{L}t_{j_l} (e^\alpha_{j_l} - e^\alpha_{\bar m }) + t_i^{\alpha^*}(e^\alpha_i - e^\alpha_{\bar m }) + \sum_k s_k(e^\alpha_k- e_{\bar m})
\]
And we find that
\begin{align*}
   \pi_\beta(j_1 - \bar m, q_\alpha(t^{\alpha^*})= & \sum_{l=1}^L \pi_\beta (j_1 -\bar m, e_{j_l}^\alpha - e_{\bar m}^\alpha) t_{j_l} + \pi(j_1 - \bar m, e^\alpha_i - e^\alpha_{\bar m }) t_i^{\alpha^*} + \pi(j_1 - \bar m, e^\alpha_{\bar m} +  \sum_k s_k(e^\alpha_k- e_{\bar m})) \\
   = &  \sum_{l=1}^L (A^\beta_{\bar m \bar m} - A^\beta_{j_1 \bar m})-(A_{\bar m j_l}-A_{j_1 j_l})t_{j_l} + \pi(j_1 - \bar m, e^\alpha_i - e^\alpha_{\bar m }) t_i^{\alpha^*} + \pi(j_1 - \bar m, e^\alpha_{\bar m} +  \sum_k s_k(e^\alpha_k- e_{\bar m}))\\
   \cdots \\
   \pi_\beta(j_L - \bar m, q_\alpha(t^{\alpha^*})= & \sum_{l=1}^L \pi_\beta (j_L -\bar m, e_{j_l}^\alpha - e_{\bar m}^\alpha) t_{j_l} + \pi(j_L - \bar m, e^\alpha_i - e^\alpha_{\bar m }) t_i^{\alpha^*} + \pi(j_L - \bar m, e^\alpha_{\bar m} +  \sum_k s_k(e^\alpha_k- e_{\bar m})) \\
   = &  \sum_{l=1}^L (A^\beta_{\bar m \bar m} - A^\beta_{j_L \bar m})-(A_{\bar m j_l}-A_{j_L j_l})t_{j_l} + \pi(j_L - \bar m, e^\alpha_i - e^\alpha_{\bar m }) t_i^{\alpha^*} + \pi(j_L- \bar m, e^\alpha_{\bar m} +  \sum_k s_k(e^\alpha_k- e_{\bar m}))
\end{align*}
Thus we can regard equations in \eqref{eq:2p-binding-const} as a set of linear equations in variables,
$t_{j_{1}},t_{j_{2}},\cdots,t_{j_{L}}.$ Then, from the implicit function theorem and Lemma \ref{lem:det} (Condition B)
we can find functions $t_{j_{1}}^{*}(t_{i})$, $t_{j_{2}}^{*}(t_{i}),\cdots$$,\,t_{j_{L}}^{*}(t_{i})$
satisfying \eqref{eq:contra} and \eqref{eq:2p-binding-const} for all $t_{i}\in[t_{i}^{\alpha^*}-\epsilon,t_{i}^{\alpha^*}+\epsilon]$
for some $\epsilon>0$. Observe that $t_{j_{1}}^{*}(t_{i})$, $t_{j_{2}}^{*}(t_{i}),\cdots$$,\,t_{j_{L}}^{*}(t_{i})$
are affine in $t_{i}$. Then, we define $\phi(t_{i})=\omega((t_{i},t_{j_{1}}^{*}(t_{i}),t_{j_{2}}^{*}(t_{i}),\cdots,t_{j_{L}}^{*}(t_{i}),\bar{t}_{i_{1}},\bar{t}_{i_{2}},\cdots,\bar{t}_{i_{L'}}),\bar{t}^{\beta})$.
From Lemma \ref{lem:2p-affine}, we see that $\phi(t_{i})$ is
affine with respect to $t_{i}$. We then find $\phi'$ and again have
two cases. \\
Case 2-1. Suppose that $\phi'=0$. Then, by increasing $t_{i}$
up to $\pi_{\beta}(\bar{m},\mathbf{q}(t_\alpha))=\pi_{\beta}(i,\mathbf{q}(t_\alpha))$, we
can find $t^{**}$ which satisfies $\omega(t^{**})=\omega(t^{*})$
and obtain the desired properties in the proposition. \\
Case 2-2. Suppose that
$\phi'\ne0$.
Then, we have either $\phi(t_{i}^{\alpha^{*}}-\epsilon)>\phi(t_{i}^{\alpha^{*}})>\phi(t_{i}^{\alpha^{*}}+\epsilon)$
or $\phi(t_{i}^{\alpha^{*}}-\epsilon)<\phi(t_{i}^{\alpha^{*}})<\phi(t_{i}^{\alpha^{*}}+\epsilon)$,
in contradiction to the optimality of $t^{*}$.
\end{proof}

\begin{lem}
\label{lem:det} The following statement holds:
\[
\pi_\kappa(\bar{m},r) = \pi_\kappa(i_{1},r),\cdots, \pi_\kappa(\bar{m},r)=\pi_\kappa(i_{K},r),\,r_{\bar m} + \sum_{i=1}^{K}r_{i_l}=1, \,\, \Sigma_r= \{\bar m,  i_1,  \cdots, i_K\}
\]
have a unique solution.\\
$\iff$
$det(D) \neq 0$ where
\[
D=\begin{pmatrix}A^\kappa_{\bar{m}\bar{m}}-A^\kappa_{i_{1}\bar{m}}-(A^\kappa_{\bar{m}i_{1}}-A^\kappa{i_{1}i_{1}})  & \cdots & A^\kappa_{\bar{m}\bar{m}}-A^\kappa_{i_{1}\bar{m}}-(A^\kappa_{\bar{m}i_{K}}-A^\kappa_{i_{1}i_{K}})\\
A^\kappa_{\bar{m}\bar{m}}-A^\kappa_{i_{2}\bar{m}}-(A^\kappa_{\bar{m}i_{1}}-A^\kappa_{i_{2}i_{1}}) & \cdots & A^\kappa_{\bar{m}\bar{m}}-A^\kappa_{i_{2}\bar{m}}-(A^\kappa_{\bar{m}i_{K}}-A_{i_{2}i_{K}})\\
\vdots & \ddots & \vdots\\
A^\kappa_{\bar{m}\bar{m}}-A^\kappa_{i_{K}\bar{m}}-(A^\kappa_{\bar{m}i_{1}}-A^\kappa_{i_{K}i_{1}}) & \cdots & A^\kappa_{\bar{m}\bar{m}}-A^\kappa_{i_{K}\bar{m}}-(A^\kappa_{\bar{m}i_{K}}-A^\kappa_{i_{K}i_{K}})
\end{pmatrix}
\]
\end{lem}
\begin{proof}
We have the following equivalence:
\[
\pi_\kappa(\bar{m},r) = \pi_\kappa(i_{1},r),\cdots, \pi_\kappa(\bar{m},r)=\pi_\kappa(i_{K},r),\,r_{\bar m} + \sum_{i=1}^{K}r_{i_l}=1, \,\, \Sigma_r= \{\bar m,  i_1,  \cdots, i_K\}
\]
have a unique solution if and only if
\begin{align*}
\pi_\kappa(\bar{m},(1-\sum_{l=1}^{K}r_{i_l}) e_{\bar m} + \sum_{l=1}^{K}r_{i_l} e_{i_l})-\pi_\kappa(i_{1},(1-\sum_{l=1}^{K}r_{i_l}) e_{\bar m} + \sum_{l=1}^{K}r_{i_l} e_{i_l}) & =0,\cdots,\\
\pi_\kappa(\bar{m},(1-\sum_{l=1}^{K}r_{i_l}) e_{\bar m} + \sum_{l=1}^{K}r_{i_l} e_{i_l})-\pi_\kappa(i_{K},(1-\sum_{l=1}^{K}r_{i_l}) e_{\bar m} + \sum_{l=1}^{K}r_{i_l} e_{i_l}) & =0
\end{align*}
have a unique solution. Let fix $k$. Then we have
\begin{align*}\label{eq:pi-payoff}
      & \pi_\kappa(i_{k},(1-\sum_{l=1}^{K}r_{i_l}) e_{\bar m} + \sum_{l=1}^{K}r_{i_l} e_{i_l})- \pi_\kappa(\bar{m},(1-\sum_{l=1}^{K}r_{i_l}) e_{\bar m} + \sum_{l=1}^{K}r_{i_l} e_{i_l})  \\
     =&   A^\kappa_{i_k \bar m}- A^\kappa_{\bar m \bar m} + \sum_{l=1}^K ((A^\kappa_{\bar m \bar m} - A^\kappa_{i_k \bar m})-(A^\kappa_{\bar m i_l}-A^\kappa_{i_k i_l})) r_{i_l}
\end{align*}
and from this, we obtain the desired result.

\end{proof}

Let $\mathcal{K}_{\bar{m}}^{*}$ be the set of all paths in $\mathcal{K}_{\bar{m}}$
that satisfy the conditions in Proposition \ref{prop:2p-binding}.
Then, we obviously have
\[
\min\{\omega(t):t\in\mathcal{K}_{\bar{m}}^{*}\}=\min\{\omega(t):t\in\mathcal{K}_{\bar{m}}\}
\]
Next, suppose that $\mathbf{q}^{*}$ is the exit point of the minimum escaping
path. If $\pi_{\beta}(\bar{m},\mathbf{q}^{*})=\pi_{\beta}(i,\mathbf{q}^{*})$ for some
$i$, then $\pi_{\alpha}(\bar{m},\mathbf{q}^{*})>\pi_{\alpha}(l,\mathbf{q}^{*})$ for
all $l$ and vice versa. This is because if $\pi_\beta (\bar m, \mathbf{q}^*) = \pi_\beta (i, \mathbf{q}^*)$ and $\pi_\alpha (\bar m, \mathbf{q}^*) = \pi_\alpha (l, \mathbf{q}^*)$, then we can always construct the escaping path with a smaller cost by removing $\alpha$-agents' (or $\beta$-agents') transitions. Thus, Proposition \ref{prop:2p-binding} implies
that if $\pi_{\beta}(\bar{m},\mathbf{q}^{*})=\pi_{\beta}(i,\mathbf{q}^{*})$ for some
$i$, $t_{j}^{\alpha^{*}}=0$ for all $j$.

\begin{prop}[One-population mistakes]\label{prop:2p-1p}
Suppose that Condition \textbf{B} holds.
Then there exists $t^{*}$ such that $\omega(t^{*})=\min\{\omega(t):t\in\mathcal{K}_{\bar{m}}^{*}\}$ and $t^*$ involves only mistakes of one population.\end{prop}
\begin{proof}
Let $t^*$ that satisfies Proposition \ref{prop:2p-binding} be given. Suppose that $t_i^{\alpha^*} >0$. The other case follows similarly.  Then, by Proposition \ref{prop:2p-binding}, $\pi_{\beta}(\bar{m},\mathbf{q}^{*})=\pi_{\beta}(i,\mathbf{q}^{*})$ for some
$i$. From the remarks before the proposition, we have $\pi_\alpha(\bar m, \mathbf{q}^*) > \pi_\alpha ( l, \mathbf{q}^*)$ for all $l$. Again, Proposition \ref{prop:2p-binding} implies that $t_l^{\beta^{*}}=0$ for all $l$.
\end{proof}
Finally, we have the following result.
\begin{prop}
\label{prop:reduction}Suppose that Condition \textbf{B} holds. Then there exists
$t^{*}$ such that $\min\{\omega(t):\zeta(t)\in\mathcal{K}_{\bar{m}}^{*}\}$ and
\[
t_k^{\alpha^*} >0 \,\,\text{for some } k \,\,\text{and }\,\, t_k^{\alpha^*}=0 \,\, \text{ for all }  k \neq l
\]
or
\[
t_k^{\beta^*} >0 \,\, \text{for some } k \,\,\text{and }\,\, t_k^{\beta^*}=0 \,\,\text{ for all }  k \neq l
\]
\end{prop}
\begin{proof}
Suppose that the minimum cost escaping path involves only one population, say $\alpha$-population, by Proposition \ref{prop:2p-1p}. Then, $x_\beta = e^{\bar m}_\beta$ for all $x$ in the minimum cost escaping path. Thus we have $\pi_\alpha (i, x)=\pi_\alpha(j,x)$ for all $i,j \neq \bar m$ and for all $x$ in the minimum cost escaping path. The costs of intermediate states in the minimum cost escaping path are the same; the \textbf{WBP} implies that the minimum cost escaping path lies in at the boundary of the simplex, yielding the desired result.
\end{proof}
Now the proof for Theorem \ref{thm:main-2p}  follows from Proposition \ref{prop:reduction}.

\renewcommand{\thesection}{C}

\section{Stochastic stability: the maximin criterion} \label{appen:suff}


In this section, we examine the problem of finding a stochastically stable state \citep{FPY1}. When $\beta=\infty$, the strategy updating dynamic is called an unperturbed process, where each convention becomes an absorbing state for the dynamic. For all $\beta < \infty$, since the dynamic is irreducible, there exists a unique invariant measure. As the noise level becomes negligible ($\beta \rightarrow \infty$), the invariant measure converges to a point mass on one of the absorbing states, called a stochastically stable state. One popular way to identify a stochastically stable state is the so-called ``maxmin criterion''\footnote{\label{ft}See \citet{Young93JET, Young98, Kandori98, BSY1, HLNN2016}}; when some sufficient conditions are satisfied, this method, along with our results on the exit problem (Theorems \ref{thm:escape} and \ref{thm:main-2p}), provides the characterization of stochastic stability.


To study stochastic stability, we have to find a minimum cost path from one convention to another. More precisely, we fix conventions $i$ and $j$. For one-population models, we let the set of all paths from convention $i$ to $j$ be
\begin{align*}
\mathcal{L}^{(n)}_{i, j}: & =\{\gamma:\gamma=(x_{0},\cdots,x_{T})\:\text{and\,}x_{0}=e_{i},\, x_{t+1}=(x_{t})^{k,l},\text{ for some \ensuremath{k, l}, \,\ for all }t<T-1,\,\\
 & x_{T}\in D(e_{j})\,\,\text{for some}\,\, T>0\}.
\end{align*}
We define a similar set for two-population models. We then consider the following problem:
\begin{equation}
C^{(n)}_{i j} := \min \{I^{({n})} (\gamma): \gamma \in \mathcal{L}^{(n)}_{i, j} \}.
\label{eq:sto-prob}
\end{equation}
Again, when $n$ is finite, $C^{(n)}_{ij} $ is complicated, involving many negligible terms; we thus study the stochastic stability problem at $n=\infty$, which again provides the asymptotics of the invariant measure and stochastic stability when $n$ is large. We let
\begin{equation}
	C_{ij} = \lim_{n \rightarrow \infty} \frac{1}{n} C^{(n)}_{ij}
	\label{eq:c-cost}
\end{equation}
and $C$ be a $|S| \times |S|$ matrix whose elements are given by $C_{ij}$ for $i \neq j$ (we set an arbitrary number if $i=j$).
Having solved the problems in equation \eqref{eq:sto-prob} (and \eqref{eq:c-cost}), the standard method to find a stochastically stable state is to construct an $i-$ rooted tree with vertices consisting of the absorbing states and whose cost is defined as the sum of all costs between the absorbing states connected by edges. Then, the stochastic stable state is precisely the root of the minimal cost tree from among all possible rooted trees (see \citet{Young98} for more details). In principle, to find a minimal cost tree (hence a stochastically stable state), we need to explicitly solve the problem in equation \eqref{eq:sto-prob}. However, in many interesting applications such as bargaining problems, the minimum cost estimates of the escaping path in Theorem \ref{thm:escape} are sufficient to determine stochastic stability without knowing the true costs of transition between conventions; this method is called the ``maxmin'' criterion (see the papers cited in footnote \ref{ft}; see also Proposition \ref{prop:BSY} below).
More precisely, we define the incidence matrix of matrix $C$, $\mathbf{Inc}(C)$, as follows:
\begin{align*}
	(\mathbf{Inc}(C))_{ij} := \begin{cases}
		1 \textrm{  if  } j = \arg \min_{l \neq i} C_{i l} \\
		0 \textrm{ otherwise }
		\end{cases}
\end{align*}
In words, the incidence matrix of $C$ has $1$ at the $i$-th and $j$-th position if the minimum of elements in the $i$th row achieves at the $i$-th and $j$-th position, and  $0$ otherwise. We also say that the incidence matrix of $C$ contains a cycle, $(i, i_1, i_2, \cdots, i_{t-1}, i)$, if
\[
	\mathbf{Inc}(C)_{i i_1}	\mathbf{Inc}(C)_{i_1 i_2} \cdots \mathbf{Inc}(C)_{i_{t-1} i} >0
\]
for $t \ge 2$.
Observe that we can obtain a graph by connecting the vertices of conventions $i, j$ whose $(\mathbf{Inc}(C))_{i j}$ is 1. Also, $\mathbf{Inc}(C)$ always contains a cycle and hence the graph contains the corresponding cycle. If this cycle is unique, by removing an edge from the cycle, we can obtain a tree; this is a candidate tree to the problem of finding a minimal cost tree. Now, we are ready to state some known sufficient conditions to identify stochastic stable states.

\begin{prop}[\citet{BSY1}] Let $i^* \in \arg \max_i \min_{j \neq i} C_{i j}$. Suppose that either \\
(i) $ \max_{j \neq i} C_{j i^*} < \min_{j \neq i} C_{i^* j} $ \\
 or \\
(ii) $\mathbf{Inc}(C)$ has a unique cycle containing $i^*$. \\
Then $i^*$ is stochastically stable.
\label{prop:BSY}
\end{prop}
\proof
See \citet{BSY1}
\qed

The sufficient conditions (i) and (ii) for stochastic stability in Proposition \ref{prop:BSY} are called the ``local resistance test'' and ``naive minimization test,'' respectively \citep{BSY1}. If strategy $i$ pairwisely risk-dominates strategy $j$ (i.e., $A_{ii}- A_{ji} > A_{jj}-A_{ji}$), then under the uniform mistake model, $C_{i j} > 1/2$ and $C_{j i} < 1/2$ hold.
Thus, if strategy $i^*$ pairwisely risk-dominates all strategies (called a globally pairwise risk-dominant strategy), then $C_{i^*j} > 1/2$ for all $j \neq i$ and $C_{ji^*} < 1/2$ for all $j \neq i$. Thus condition (i) in Proposition \ref{prop:BSY} holds and $i^*$ is stochastically stable (see Theorem 1 in \citet{Kandori98} and Corollary 1 in \citet{Ellison2000}).

The number $\min_{j \neq i} C_{ij}$ in Proposition \ref{prop:BSY} is, as mentioned, often called the ``radius'' of convention $i$; this measures how difficult it is  to escape from convention $i$ \citep{Ellison2000}.
Proposition \ref{prop:BSY} shows that if either (i) or (ii) holds, the state with the greatest radius (and hence the state most difficult to escape) is stochastically stable. To check whether either condition (i) or (ii) holds, clearly it is enough to know that $\min_{j \neq i} C_{i j}, \max_{j \neq i} C_{j i}$ etc.

An important consequence of our main theorem on the exit problem (Theorem \ref{thm:escape}) is that it provides the lower  and upper bounds of the radius of convention $i$, $\min_{j \neq i} C_{ij}$, as follows. On  the one hand,  a path escaping from convention $i$ to $j$ (in $\mathcal{L}^{(n)}_{i,j}$) by definition exits the basin of attraction of convention $i$ and thus $\mathcal{L}^{(n)}_{i,j} \subset \mathcal{G}^{(n)}_{i}$ in equation \eqref{eq:path-G}. Thus,
\begin{equation}
  		C^{(n)}_{ij} =\min \{ I^{(n)} (\gamma): \gamma \in \mathcal{L}^{(n)}_{i,j} \} \geq	\min \{ I^{(n)} (\gamma): \gamma \in \mathcal{G}^{(n)}_{i} \},
	\label{eq:lower1}
\end{equation}
and Theorem \ref{thm:escape} shows that
\begin{equation}
 	\lim_{n \rightarrow \infty} \frac{1}{n} \min \{ I^{(n)} (\gamma): \gamma \in \mathcal{G}^{(n)}_{i} \} = \min_{j \neq i}   R_{ij}.
	\label{eq:lower2}
\end{equation}
Then equations \eqref{eq:lower1} and \eqref{eq:lower2} together give a lower bound for $\min_{j \neq i} C_{ij}$. On the other hand,  if $\gamma_{i \rightarrow j}$ is the straight line path from convention $i$ to $j$ ending at the mixed strategy  Nash equilibrium involving $i$ and $j$, we have
\begin{equation}
	I^{(n)}(\gamma_{i \rightarrow j}) \geq \min \{ I^{(n)} (\gamma): \gamma \in \mathcal{L}^{(n)}_{i,j} \} = C^{(n)}_{ij}
	\label{eq:upper1}
\end{equation}
and
\begin{equation}
	\lim_{n \rightarrow \infty} \frac{1}{n} I^{(n)}(\gamma_{i \rightarrow j}) = R_{ij}.
	\label{eq:upper2}
\end{equation}
Thus, equations \eqref{eq:upper1} and \eqref{eq:upper2} give an upper bound for $\min_{j \neq i} C_{ij}$. These are the main contents of the following proposition.


\begin{prop} Suppose Condition \textbf{A} or Condition \textbf{B} holds. Then \\ \smallskip
(i) $C_{i j} \leq R_{i j}$  for all $i, j$. \\ \smallskip
(ii) $  \min_{j \neq i} C_{ij} =  \min_{j \neq i} R_{i j} $. \\ \smallskip
(iii) $ \arg \min_{j \neq i} R_{i j} \subset \arg \min_{j \neq i} C_{i j}  $ for all $i$.
\label{prop:SSE}
\end{prop}
\proof
We obtain (i) by dividing equation \eqref{eq:upper1} by $n$ , taking the limit, and using \eqref{eq:upper2}. For (ii), from equations \eqref{eq:lower1} and \eqref{eq:lower2},  $\lim_{n \rightarrow \infty} \frac{1}{n} C_{ij}^{(n)}  \geq \min_{j \neq i} R_{ij}$, implying that $ \min_{j \neq i} C_{ij}  \geq \min_{j \neq i} R_{ij}$. Also from (i), we have $\min_{j \neq i} C_{ij}  \leq \min_{j \neq i} R_{ij}$. Thus, (ii) follows. We next prove (iii). Suppose that $j^{**} \in \arg \min_{j \neq i} R_{ij}$ and $j^* \in \arg \min_{j \neq i} C_{ij}$. Then from (i) and (ii), $R_{i j^{**}} =  C_{i j^*} \leq  C_{i j^{**}} \leq R_{i j^{**}}$. Thus $j^{**} \in \arg \min_{j \neq i} C_{ij}$ and we have $ \arg \min_{j \neq i} R_{i j} \subset \arg \min_{j \neq i} C_{ij}$.
\qed

The immediate consequence of Proposition \ref{prop:SSE} is that $\arg \max_i \min_{j \neq i} C_{i j} =  \arg \max_i \min_{j \neq i} R_{i j}$ and $\max_{j \neq i} C_{j i}  \leq \max_{j \neq i} R_{j i}$.  Further, if $\arg \min_{j \neq i} C_{ij}$ is unique for all $i$, from Proposition \ref{prop:SSE}, the incidence matrices of $C$ and $R$ are the same.  In general,  $\arg \min_{j \neq i} C_{ij}$ may not be unique for some $i$. In this case, Proposition \ref{prop:SSE} (iii) implies that if $R_{ij} =1$, then $C_{ij}=1$, which, in turn, implies that whenever $R$ yields a graph containing a unique cycle, $C$  yields the same graph containing the unique cycle. These facts enable us to replace $C$ in Proposition \ref{prop:BSY} by  $R$---a $|S| \times |S|$ matrix consisting of $R_{ij}$s (again, we assign arbitrary numbers at the diagonal positions). This is our main result on stochastic stability.


\begin{thm}[\textbf{Stochastic Stability}]
Suppose that Condition \textbf{A} or Condition \textbf{B} holds.
Let $i^* \in \arg \max_i \min_{j \neq i} R_{i j}$. Suppose also that either \\ \smallskip
(i) $ \max_{j \neq i} R_{j i^*} < \min_{j \neq i} R_{i^* j} $  \\ or \\
(ii) $\mathbf{Inc}(R)$ has a unique cycle containing $i^*$. \\
Then, $i^*$ is stochastically stable.
\label{thm:SSE}
\end{thm}
\proof Let $i^* \in \arg \max_i \min_{j \neq i} R_{i j}$. From Proposition \ref{prop:SSE} (iii),  $i^* \in  \arg \max_i \min_{j \neq i} C_{i j}$. We first suppose that (i) holds. Now, Propositions \ref{prop:SSE} (i) and \ref{prop:SSE} (ii) imply that
\[
\max_{j \neq i^*} C_{j i^*} \leq \max_{j \neq i^*} R_{j i^*} < \min_{j \neq i^*}  R_{i^* j} = \min_{j \neq i^*} C_{i^* j}.
\]
Thus, Proposition \ref{prop:BSY} implies that $i^*$ is stochastically stable. Now, suppose that (ii) holds. From Proposition \ref{prop:SSE} (iii) and the remarks before Theorem \ref{thm:SSE}, $\mathbf{Inc}(C)$ contains a unique cycle containing $i^*$, too. Thus, Proposition \ref{prop:BSY} again implies that $i^*$ is stochastically stable.
\qed

Note that two-strategy games trivially satisfy both conditions (i) and (ii) in Theorem \ref{thm:SSE}. Here, we can easily check that the stochastic stable state is the risk-dominant equilibrium. In particular, \citet{Kandori98} show that when a coordination game exhibits positive feedback (the marginal bandwagon property), a ``globally pairwise risk-dominant equilibrium'' is stochastically stable under the uniform mistake model  (see also \citet{BSY1}). However, when the number of strategies exceeds two, Theorem \ref{thm:SSE} shows that stochastically stable states
under the logit choice rule do not necessary satisfy
the criterion of pairwise risk dominance. To summarize, Theorem \ref{thm:SSE} asserts that when either condition (i) or condition (ii) is satisfied, the state with the largest radius (and hence the most difficult state to escape) is stochastically stable, in line with the existing results for uniform interaction models. However, the  radius now depends on the opportunity cost of individuals' mistakes as well as the threshold number of agents inducing others to play a new best-response.

\renewcommand{\thesection}{D}
\section{Stochastic  stable states for Nash demand games \label{appen:sss-NDG} }

We first show that Nash demand game,
\begin{equation}\label{eq:ndg}
  (A^\alpha_{ij}, A^\beta_{ij}) :=
  \begin{cases}
    (\delta i, f(\delta j)), & \mbox{if } i \leq j \\
    (0, 0) , & \mbox{if } i > j,
  \end{cases}
\end{equation}
satisfies \textbf{Condition B}.

\noindent{\textbf{Condition B} (i)}. \\
We divide cases as follows: \\
\noindent (1) $\bar m > i > j$.
\[
    A^\alpha_{\bar m \bar m } -A^{\alpha}_{i \bar m } - (A_{\bar m j}^\alpha - A_{ij}^\alpha) = \delta m - \delta i >0, \,\,\, A^\beta_{\bar m \bar m } -A^{\beta}_{\bar m i } - (A^\beta_{j \bar m } - A_{ij}^\beta)= f(\delta \bar m) - (f(\delta \bar m) - f(\delta i)) >0
\]
\noindent (2) $\bar m > j > i$.
\[
    A^\alpha_{\bar m \bar m } -A^{\alpha}_{i \bar m } - (A_{\bar m j}^\alpha - A_{ij}^\alpha) = \delta m - \delta i + \delta_j>0, \,\,\, A^\beta_{\bar m \bar m } -A^{\beta}_{\bar m i } - (A^\beta_{j \bar m } - A_{ij}^\beta)= f(\delta \bar m) - f(\delta \bar m)  >0
\]
\noindent (3) $i > \bar m > j $.
\[
    A^\alpha_{\bar m \bar m } -A^{\alpha}_{i \bar m } - (A_{\bar m j}^\alpha - A_{ij}^\alpha) = \delta m >0, \,\,\, A^\beta_{\bar m \bar m } -A^{\beta}_{\bar m i } - (A^\beta_{j \bar m } - A_{ij}^\beta)= f(\delta \bar m) -f(\delta i) - (f(\delta \bar m) - f(\delta i)) =0
\]
\noindent (4) $j > \bar m > i$.
\[
    A^\alpha_{\bar m \bar m } -A^{\alpha}_{i \bar m } - (A_{\bar m j}^\alpha - A_{ij}^\alpha) = \delta \bar m - \delta i -(\delta \bar m - \delta i )=0, \,\,\, A^\beta_{\bar m \bar m } -A^{\beta}_{\bar m i } - (A^\beta_{j \bar m } - A_{ij}^\beta)= f(\delta \bar m) >0
\]
\noindent (5) $i> j > \bar m $.
\[
    A^\alpha_{\bar m \bar m } -A^{\alpha}_{i \bar m } - (A_{\bar m j}^\alpha - A_{ij}^\alpha) = \delta \bar m - \delta \bar m =0, \,\,\, A^\beta_{\bar m \bar m } -A^{\beta}_{\bar m i } - (A^\beta_{j \bar m } - A_{ij}^\beta)= f(\delta \bar m) -f(\delta i)-(-f(\delta i)) >0
\]
\noindent (6) $j> i > \bar m $.
\[
    A^\alpha_{\bar m \bar m } -A^{\alpha}_{i \bar m } - (A_{\bar m j}^\alpha - A_{ij}^\alpha) = \delta \bar m - (\delta \bar m -\delta i)> 0, \,\,\, A^\beta_{\bar m \bar m } -A^{\beta}_{\bar m i } - (A^\beta_{j \bar m } - A_{ij}^\beta)= f(\delta \bar m) -f(\delta i)>0
\]

\noindent{\textbf{Condition B} (ii)}. \\
We first show the following lemma.
\begin{lem} \label{lem:mat-sol}
Suppose that $A$ is a $n \times n$ matrix such that
\[
     A_{ij} = a_i \text{ if } i \leq j ,\,\, = 0 \text{ if } i > j, \,\, a_i < a_{i+1} \text{ for all } i=1, \cdots, n-1
\]
Then there exists a unique $x \gg 0$ such that $Ax = \mathbf{1}$ where $\mathbf{1}$ is the column vector consisting all 1's.
\end{lem}
\begin{proof}
Let $x$ be
\[
    x^T =\left( \frac{1}{a_1} - \frac{1}{a_2}, \cdots, \frac{1}{a_{n-1}} - \frac{1}{a_n}, \frac{1}{a_n} \right)
\]
Note that by the assumption, we have $x \gg 0$.
Then we have
\[
    (Ax)_k = \sum_{i=1}^n A_{ki} x_i = \sum_{i=1}^k a_k x_i =  a_{k} \sum_{i=k}^n x_i =  a_{k} \frac{1}{a_k} =1
\]
Suppose that there exists $y$ such that $Ay =\mathbf{1}$. Then, since $det(A) \neq 0$, $y=A^{-1}\mathbf{1}=x$. Thus $x \gg 0$ is unique.
\end{proof}
\noindent Now let $i_1, \cdots, i_K$. We rearrange $i_k$'s such that $i_1 < \cdots < i_K$. Let $A$ be a matrix whose rows and columns consist of $i_1, \cdots, i_K$. Then from \eqref{eq:ndg}, the hypothesis of Lemma \ref{lem:mat-sol} is satisfied. Thus, by normalizing $x$, we can find a unique $q \in \Delta_\beta$ which satisfies the desired property.

Recall that
\com{
\[
    R^U_{\bar{m} j}:=\min\{(A_{\bar{m}\bar{m}}^{\beta}
    -A_{\bar{m}j}^{\beta})\frac{(A_{\bar{m}\bar{m}}^{\alpha}
-A_{j\bar{m}}^{\alpha})}{(A_{\bar{m}\bar{m}}^{\alpha}
-A_{j\bar{m}}^{\alpha})+(A_{jj}^{\alpha}-A_{\bar{m}j}^{\alpha})},
(A_{\bar{m}\bar{m}}^{\alpha}-A_{\bar{j}m}^{\alpha})
\frac{(A_{\bar{m}\bar{m}}^{\beta}-A_{\bar{m}j}^{\beta})}
{(A_{\bar{m}\bar{m}}^{\beta}-A_{\bar{m}j}^{\beta})
+(A_{jj}^{\beta}-A_{j\bar{m}}^{\beta})}\}.
\]
}
\[
    R^U_{m j}:=\min\{(A_{mm}^{\beta}
    -A_{mj}^{\beta})\frac{(A_{mm}^{\alpha}
-A_{jm}^{\alpha})}{(A_{mm}^{\alpha}
-A_{jm}^{\alpha})+(A_{jj}^{\alpha}-A_{mj}^{\alpha})},
(A_{mm}^{\alpha}-A_{\bar{j}m}^{\alpha})
\frac{(A_{mm}^{\beta}-A_{mj}^{\beta})}
{(A_{mm}^{\beta}-A_{mj}^{\beta})
+(A_{jj}^{\beta}-A_{jm}^{\beta})}\}
\]
and
\[
    (A_{ij}^\alpha, A_{ij}^\beta) := \begin{cases}
                                       (\delta i, f(\delta j)), & \mbox{if } i \leq j\\
                                       (0, 0) , & \mbox{if } i > j
                                     \end{cases}
\]
Then we divide cases:

\noindent \textbf{(i) $m< j$}.
We find that
\[
    A_{mm}^\beta = f(\delta m), A_{mj}^\beta = f(\delta j), A^\alpha_{mm} = \delta m, A^\alpha_{jm} =0, A^\alpha_{jj}=\delta j,
    A^\alpha_{mj} = \delta m
\]
and
\[
    A_{mm}^\alpha = \delta m, A_{jm}^\alpha  = 0, A^\beta_{mm} = f(\delta m), A^\beta_{mj} =f(\delta j), A^\beta_{jm}=0,
    A^\beta_{jj} = f(\delta j)
\]
Using these, we find that
\[
    (A_{mm}^\beta - A_{mj}^\beta) \frac{A_{mm}^\alpha-A_{jm}^\alpha}{A_{mm}^\alpha-A_{jm}^\alpha+(A_{jj}^\alpha-A_{mj}^\alpha)} =
    (f(\delta m) - f(\delta j)) \frac{\delta m}{\delta j}
\]
and
\[
    (A_{mm}^\alpha - A_{jm}^\alpha) \frac{A_{mm}^\beta-A_{mj}^\beta}{A_{mm}^\beta-A_{mj}^\beta+(A_{jj}^\beta-A_{jm}^\beta)} =
    \delta m \frac{f(\delta m)-f(\delta j)}{f(\delta m)}.
\]

\noindent \textbf{(ii) $m > j$}.
We find that
\[
    A_{mm}^\beta = f(\delta m), A_{mj}^\beta = 0, A^\alpha_{mm} = \delta m, A^\alpha_{jm} =\delta j, A^\alpha_{jj}=\delta j,
    A^\alpha_{mj} = 0
\]
and
\[
    A_{mm}^\alpha = \delta m, A_{jm}^\alpha  = \delta j, A^\beta_{mm} = f(\delta m), A^\beta_{mj} =0, A^\beta_{jm}=f(\delta m),
    A^\beta_{jj} = f(\delta j)
\]
Using these, we find that
\[
    (A_{mm}^\beta - A_{mj}^\beta) \frac{A_{mm}^\alpha-A_{jm}^\alpha}{A_{mm}^\alpha-A_{jm}^\alpha+(A_{jj}^\alpha-A_{mj}^\alpha)} =
    f(\delta m) \frac{\delta m - \delta j}{\delta m}
\]
and
\[
    (A_{mm}^\alpha - A_{jm}^\alpha) \frac{A_{mm}^\beta-A_{mj}^\beta}{A_{mm}^\beta-A_{mj}^\beta+(A_{jj}^\beta-A_{mj}^\beta)} =
    (\delta m -\delta j) \frac{f(\delta m)}{f(\delta j)}.
\]

Thus we have
\[
    R^U_{mj} = \begin{cases}
               (f(\delta m) - f(\delta j)) \frac{\delta m}{\delta j} \wedge \delta m \frac{f(\delta m)-f(\delta j)}{f(\delta m)}
               & \mbox{if } m<j \\
               f(\delta m) \frac{\delta m - \delta j}{\delta m} \wedge  (\delta m -\delta j) \frac{f(\delta m)}{f(\delta j)} &
               \mbox{if } m > j
             \end{cases}
\]
Or
\[
    R^U_{mj} =  \min_{m<j} \{ (f(\delta m) - f(\delta j)) \frac{\delta m}{\delta j} \wedge \delta m \frac{f(\delta m)-f(\delta j)}{f(\delta m)} \} \wedge \min_{m > j} \{f(\delta m) \frac{\delta m - \delta j}{\delta m} \wedge  (\delta m -\delta j) \frac{f(\delta m)}{f(\delta j)} \}.
\]
Note that we have
\[
    R^I_{mj} =  \min_{m<j} \{ \delta m \frac{f(\delta m)-f(\delta j)}{f(\delta m)} \} \wedge \min_{m > j} \{f(\delta m) \frac{\delta m - \delta j}{\delta m} \}.
\]
Then we would like to find $\min_j R^U_{mj}$.
To do this, we first have the following lemma.
\begin{lem} \label{lem:der}
    Suppose that $f(x) \geq 0$, $f'(x)<0$ and $f''(x)<0$ for all $x$. Let $y$ be given.\\
    (i) $\frac{f'(x)}{f(x)}$ is decreasing in $x$. \\
    (ii) $x f'(x) - f(x)$ is decreasing in $x$. \\
    (iii) $ f'(x) + \frac{f(x)}{x}$ is decreasing in $x$. \\
    (iv) $f'(x) + (\frac{f(x)}{x})^2$ is decreasing in $x$. \\
    (v) $(f(y)-f(x)) \frac{y}{x}$ is increasing in $x$  \\
    (vi) $(y-x) \frac{f(y)}{f(x)}$ is decreasing in $x$.
\end{lem}
\begin{proof}
(i)-(iv) are easily verified by taking derivatives. We show (v). (vi) follows similarly.
Let $\varphi(x):=(f(y)-f(x)) \frac{y}{x}$. We find that
\[
    \varphi'(x)=y \frac{- f'(x)x +f(x) - f(y)}{x^2}
\]
Then since $- f'(x)x +f(x)$ is increasing in $x$, we have
\[
    -f'(x) x + f(x) - f(y) \geq  f(0) - f(y) \geq 0
\]
since $f$ is decreasing. Thus $\varphi'(x)>0$.
\end{proof}

Thus using Lemma \eqref{lem:der}, we find that
\[
    \min_{j} R^U_{mj} =
    \min \{
   (f(\delta m) - f(\delta (m+1))) \frac{\delta m}{\delta (m+1)}, \delta m \frac{f(\delta m)-f(\delta
               (m+1))}{f(\delta (m))},
               f(\delta m) \frac{\delta }{\delta m},   \delta  \frac{f(\delta m)}{f(\delta (m-1))}
    \}
\]
We let
\[
    r_1(m):= (f(\delta m) - f(\delta (m+1))) \frac{\delta m}{\delta (m+1)},\,\, r_2(m):=\delta m \frac{f(\delta m)-f(\delta
               (m+1))}{f(\delta (m))}
\]
and
\[
    l_1(m):= f(\delta m) \frac{\delta }{\delta m},\,\, l_2(m):=\delta  \frac{f(\delta m)}{f(\delta (m-1))}.
\]

\begin{lem}
We have the following results:\\
(i) $r_1$ and $r_2$ are increasing in $m$.\\
(ii)  $l_1$ and $l_2$ are decreasing in $m$.
\end{lem}
\begin{proof}
  (i). Since $f''<0$, $f(\delta m) - f(\delta(m+1))$ is increasing. Since $\frac{\delta m}{\delta (m+1)}$ is increasing, two terms in $r_1$ are both positive and increasing, hence $r_1$ is increasing. Also since $f''<0$, $\frac{f(\delta(m+1))}{f(\delta m)}$ is decreasing in $m$. Thus $r_2$ is increasing.
\end{proof}

Then $r_1$ and $r_2$ are increasing in $m$ and $l_1$ and $l_2$ are decreasing in $m$.
\begin{lem} \label{lem:sss}
Suppose that
\[
    m^* \in \arg \max_m \min_{j} R^U_{mj}
\]
Then for all $m < m^*$, $\min_j R^U_{m j} = R^U_{m, m+1}$ and for all $m> m^*$,  $\min_j R^U_{m j} = R^U_{m, m-1}$
\end{lem}
\begin{proof}
  Let $\hat R(m):=\min_j R^U_{mj}$. We show that
  \begin{align*}
     & \text{ If } m < m^*, \text{ then } \hat R(m) = r_1(m) \,\,\, \text{or } r_2(m) \\
     & \text{ If } m > m^*, \text{ then } \hat R(m) = l_1(m) \,\,\, \text{or } l_2(m)
  \end{align*}
and then the desired results follow. We show the first claim. (the second claim follows similarly). Let $m < m^*$ and $\hat R(m) = l_1(m)$. Then since $l_1(m)$ is decreasing in $m$, $l_1(m)>l_1(m^*)$ and by definition, we have $\hat R(m^*) \leq l_1(m^*)$. Thus we have
\[
    \min_j R^U_{mj} =\hat R(m)  = l_1(m)  > l_1(m^*) \geq \hat R(m^*)
\]
which is contradiction to $m^* \in \arg \max_m \min_{j} R^U_{mj}$ If $\hat R(m) = l_2(m)$, the exactly same argument leads to a contradiction. Thus if $m < m^*$, then $\hat R(m) = r_1(m)$ or $r_2(m)$.
\end{proof}

Let $s^*$ and $s^I$ such that
\[
    -f'(s^*) = \frac{f(s^*)}{s^*} \text{ and } - f'(s^I) = (\frac{f(s^I)}{s^I})^2
\]
and for $\mu \in [0, \frac{\bar s_\alpha}{\delta}] \cap \mathbb{R}$, let $\mu^I=\mu^I(\delta)$, $\mu^*=\mu^*(\delta)$, and  $\mu^{**}=\mu^{**}(\delta)$ such that
\begin{equation}\label{eq:def-mus}
  r_1 (\mu^*) = l_1(\mu^*), r_2(\mu^{**}) = l_2(\mu^{**}) \text{ and } r_2(\mu^I) = l_1(\mu^I).
\end{equation}
\begin{lem} \label{lem:conv}
We have the following results. As $\delta \rightarrow 0$,
\com{
\[
    \delta \mu^*(\delta) \rightarrow x^*, \,\,  \frac{r_2(\mu^*)}{l_1(\mu^*)} \rightarrow \frac{f'(x^*)}{\frac{f(x^*)}{x^*}}, \qquad \delta \mu^I(\delta) \rightarrow x^I,\frac{r_2(\mu^I)}{l_1(\mu^I)} \rightarrow \frac{f'(x^I)}{(\frac{f(x^I)}{x^I})^2}.
\]
}
\[
    \delta \mu^*(\delta) \rightarrow s^*, \qquad \delta \mu^{**}(\delta) \rightarrow s^{*}, \qquad \delta \mu^I(\delta) \rightarrow s^I.
\]
\end{lem}
\begin{proof}
  For $\delta \mu^*(\delta) \rightarrow s^*$, let
   \[
    \varphi_\delta(x):= \frac{(f(x) - f(x + \delta) )}{\delta}  \frac{x^2}{(x + \delta)f(x)},\,\,\, \varphi (x):= - f'(x) \frac{x}{f(x)}.
  \]
  Then $\varphi_\delta$ converge uniformly to $\varphi$ and $\varphi_\delta(\delta \mu^*(\delta))= \frac{r_1(\mu^*)}{l_1(\mu^*)}=1$ and $\varphi(x^*)=1$. Then the uniform convergence of $\varphi_\delta$ to $\varphi$ implies that  $\delta \mu^*(\delta) \rightarrow s^*$. The second and third parts follow similarly.
\end{proof}
Next we show that

\begin{lem} \label{lem:comp}
    We have the following result.\\
    (i) If $s^* > s^E$, then $s^* > s^I > s^E$ and $-f'(s^I) \frac{s^I}{f(s^I)} <1$ and $-f'(s^*)   <1$ \\
    (ii) If $s^* < s^E$, then $s^* < s^I < s^E$ and $-f'(s^I) \frac{s^I}{f(s^I)} >1$ $-f'(s^*)   >1$

\end{lem}
\begin{proof}
  We show (i) and (ii) follows similarly. Suppose that $s^* > s^E$. Let $s^I \geq s^*$. Since from Lemma \ref{lem:der}  $-f'(x) -\frac{f(x)}{x}$ is increasing, we have
  \[
  -f'(s^I) - \frac{f(s^I)}{s^I} \geq - f'(s^*) - \frac{f(s^*)}{s^*}=0= -f'(s^I) - (\frac{f(s^I)}{s^I})^2
\]
which implies that
\[
    \frac{f(s^I)}{s^I} \geq 1 = \frac{f(s^E)}{s^E}
\]
Since $\frac{f(s)}{s}$ is decreasing in $s$, we have
\[
    s^E \geq s^I \geq s^* > s^E
\]
which is a contradiction. Now suppose that $s^I \leq s^E$. Then since $s^E < s^*$,
\[
-f'(s^I) - \frac{f(s^I)}{s^I} < - f'(s^*) - \frac{f(s^*)}{s^*}=0= -f'(s^I) - (\frac{f(s^I)}{s^I})^2
\]
which implies that
\[
   \frac{f(s^I)}{s^I} < 1.
\]
which is a contradiction to $\frac{f(s^I)}{s^I} \geq \frac{f(s^E)}{s^E} =1$ from  $s^I \leq s^E$.  Now from $s^*> s^I$ and $s^*> s^E$, respectively we have
\[
    -f'(s^I) \frac{s^I}{f(s^I)} <1 \text{ and }  -f'(s^*) <1.
\]
\end{proof}

\begin{lem}\label{lem:com} We have the following results. \\
(i) If $s^* > s^E$, then there exists $\underline \delta$ such that for all $\delta < \underline \delta$, $\mu^* > \mu^I$ and
\[
     r_1(\mu^I) < r_2(\mu^I) = l_1(\mu^I) \text{ and } r_1(\mu^*) < l_2(\mu^*)
\]
where $\mu^I=\mu^I(\delta)$ and $\mu^*=\mu^*(\delta)$ are defined in \eqref{eq:def-mus}. \\
(ii) If $s^* < s^E$, then there exists $\underline \delta$ such that for all $\delta < \underline \delta$, $\mu^{**} < \mu^I$ and
\[
     l_2(\mu^I) < r_2(\mu^I) = l_1(\mu^I)  \text{ and }  l_2(\mu^{**}) < r_1(\mu^{**})
\]
where $\mu^I=\mu^I(\delta)$ and $\mu^*=\mu^*(\delta)$ are defined in \eqref{eq:def-mus}.
\end{lem}
\begin{proof}
We first prove (i). Suppose that $s^* > s^E$.
From Lemma \ref{lem:comp}, we have
\begin{equation} \label{eq:con2}
    -f'(s^I) \frac{s^I}{f(s^I)} <1 \text{ and  } -f'(s^*)  <1
\end{equation}
Since  $\delta \mu^I \rightarrow s^I$ (Lemma \ref{lem:conv}) and $s^I < s^*$ and from \eqref{eq:con2}
\[
    \frac{r_1(\mu^I)}{l_1(\mu^I)} \rightarrow - f'(s^I) \frac{s^I}{f(s^I)} < 1,
\]
there exists $\underline \delta $ such that for all $\delta < \underline \delta$, $ r_1(\mu^I)< l_1(\mu^I) $ and $\mu^I < \mu^*$. For the second inequality $r_1(\mu^*) < l_2(\mu^*)$ similarly follows from
\[
    \frac{r_1(\mu^*)}{l_2(\mu^*)} < 1 \iff \frac{f(\delta \mu^*) - f(\delta(\mu^*+1))}{\delta } \frac{\delta \mu^*}{\delta (\mu^*+1)} \frac{f(\delta(\mu^* -1)) }{f(\delta \mu^*)} < 1
\]
and
\[
    \frac{r_1(\mu^*)}{l_2(\mu^*)} \rightarrow - f'(s^*) <1
\]
from \eqref{eq:con2}.

Next we show (ii).  Similarly to (i), from Lemma \ref{lem:comp}, we have we have
\[
    -f'(s^I) \frac{s^I}{f(s^I)} > 1 \text{  and   } f'(s^*) >1
\]
Then we have
\[
    \frac{l_2(\mu^I)}{r_2(\mu^I) } < 1 \iff \frac{\delta}{f(\delta \mu^I) - f(\delta (\mu^I + 1))} \frac{f(\delta \mu^I)}{f(\delta (\mu^I-1))} \frac{f(\delta \mu)}{\delta \mu } < 1
\]
and
\[
    \frac{l_2(\mu^{**})}{r_1(\mu^{**})} <1 \iff \frac{\delta}{f(\delta \mu^{**}) - f(\delta(\mu^{**}+1)) } \frac{\delta (\mu^{**}+1)}{\delta \mu^{**} } \frac{f(\delta \mu^{**}) }{f(\delta(\mu^{**} -1))} < 1
\]
and from these, (ii) follows.
\end{proof}

\begin{figure}
  \centering
  \includegraphics[scale=0.7]{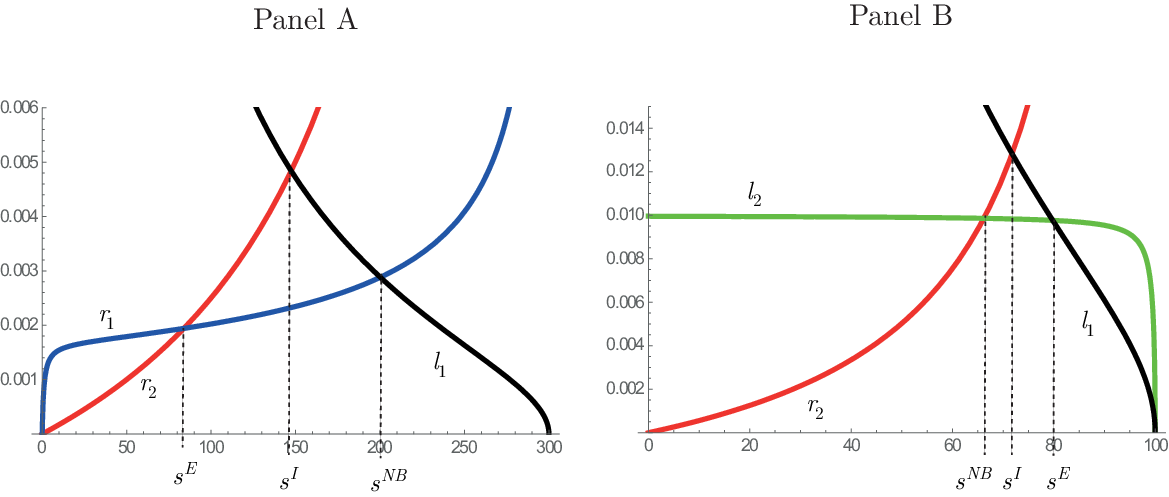}
  \caption{\textbf{Determinations of stochastically stable states.} For Panel A, $f(x)=\sqrt{1- \frac{x}{3}}$ for $x \in [0,3]$, $\delta=0.01$. For Panel B, $f(x)=\sqrt{3(1- x)}$, for $x \in [0, 1]$,  $\delta=0.01$.}\label{fig:appen-sse-det}
\end{figure}

\begin{lem} \label{lem:s-fin}
    Suppose that $\mu^*$ is given by \eqref{eq:def-mus}.\\
    (i) If $s^*>s^E$, then
    \[
        \mu^* \in \arg \max_{\mu \in [0, \frac{\bar s}{\delta}]} \min  \{ r_1(\mu), r_2 (\mu), l_1(\mu), l_2(\mu) \}
    \]
    (ii) If $s^{*}<s^E$, then
    \[
        \mu^{**} \in \arg \max_{\mu \in [0, \frac{\bar s}{\delta}]} \min  \{ r_1(\mu), r_2 (\mu), l_1(\mu), l_2(\mu) \}
    \]
\end{lem}
\begin{proof} Let $s^* > s^E$. Choose $\underline \delta$ satisfying Lemma \ref{lem:com}. Then for all $\delta<  \underline \delta$, we have
  \[
    r_1(\mu^*) = l_1(\mu^*) < l_1 (\mu_0) = r_2 (\mu_0) < r_2 (\mu^*)
  \]
  and thus $r_1(\mu^*) \leq \min \{r_2(\mu^*), l_1(\mu^*), l_2(\mu^*) \}$. Now, if $\mu< \mu^*$ then  $r_1(\mu^*) > r_1(\mu)$ since $r_1(\cdot)$ is increasing. If $\mu> \mu^*$, then $r_1(\mu^*) = l_1(\mu^*) > l_1(\mu)$ since $l_1(\cdot)$ is decreasing. Thus we have
  \[
    r_1(\mu^*) \geq \min \{ r_1(\mu), r_2 (\mu), l_1(\mu), l_2(\mu) \}
  \]
  for all $\mu \in [0, \frac{\bar s}{\delta}]$. This shows that
  \[
        \mu^* \in \arg \max_{\mu \in [0, \frac{\bar s}{\delta}]} \min  \{ r_1(\mu), r_2 (\mu), l_1(\mu), l_2(\mu) \}
  \]
  Now let $s^* < s^E$.  Again choose $\underline \delta$ satisfying Lemma \ref{lem:com}. Then for all $\delta < \underline \delta$, we have
  \[
    r_2(\mu^{**}) = l_2 (\mu^{**}) < r_1( \mu^{**}) < r_1(\mu^I) = l_1 (\mu^I) < l_1(\mu^{**})
  \]
  and similarly since $r_2$ is increasing and $l_2$ is decreasing, we obtain the desired result.
\end{proof}

\begin{table}
  \centering
  \scalefont{0.8}
  {\renewcommand{\arraystretch}{1.5}
  \begin{tabular}{c|c|c|c|c|c}
    \toprule

      &  \multicolumn{2}{c|}{$\alpha$ favored transition } & \multicolumn{2}{c|}{$\beta$ favored transition } & \multicolumn{1}{c}{$\substack{\text{Stochastic} \\ \text{stability}}$} \\
     \midrule

     & $\beta$ mistake ($A$) & $\alpha$ mistake ($B$) & $\beta$ mistake ($C$) & $\alpha$ mistake ($D$)\\
     \midrule
    Uniform & $\frac{\delta m}{\bar s - \alpha}$  & $\frac{\Delta f(\delta m)}{f(\delta m)}$ & $\frac{\delta}{\delta m }$& $\frac{f(\delta m)}{f (\delta)}$ \\
     $\substack{\text{Unintentional} \\ \text{Intentional}}$& & $ \bigcirc $ & $\bigcirc$&   & $ \frac{\Delta f(\delta m)}{f(\delta m)} \approx \frac{\delta}{\delta m }$ \\
         \midrule
    $\substack{\text{Logit} \\ \text{Unintentional}}$  & $\Delta f(\delta m)\frac{\delta m}{\delta (m+1)}$ & $\delta m \frac{\Delta f(\delta m)}{f(\delta (m))}$ & $f(\delta m) \frac{\delta }{\delta m}$ & $\delta  \frac{f(\delta m)}{f(\delta (m-1))}$ \\
   $s^{NB} > s^E$   & $\bigcirc $ & $\bigtriangleup$ & $\bigcirc$ & $ $ & $ \Delta f(\delta m)\frac{\delta m}{\delta (m+1)} \approx f(\delta m) \frac{\delta }{\delta m}$ \\
     $s^{NB} < s^E$ & $ $ & $\bigcirc $ & $\bigtriangleup$ & $ \bigcirc$ & $\delta m \frac{\Delta f(\delta m)}{f(\delta (m))} \approx \delta  \frac{f(\delta m)}{f(\delta (m-1))} $  \\
     $\substack{\text{Logit} \\ \text{Intentional}}$ &  & $\bigcirc $ & $\bigcirc $ & $  $ &  $\delta m \frac{\Delta f(\delta m)}{f(\delta (m))} \approx f(\delta m) \frac{\delta }{\delta m}$ \\

    \bottomrule
  \end{tabular}
  }
  \caption{Comparison of solutions under various mistake models. $\Delta f(\delta m) :=f(\delta m)-f(\delta (m+1))$. Resistances are determined by the minimum of $A,B,C$, and $D$. In the rows tilted with ``unintentional'', ``intentional'', $s^{NB} > s^E$, $s^{NB} < s^E$, and  ``logit intentional'' show the smaller ones. Thus under the logit unintentional dynamic, when $s^{NB} > s^E$, the transition always occurs by $\beta$ population, while $s^{NB} < s^E$, the transition always occurs by $\alpha$ population.  Entries marked by $\bigtriangleup$ and $\bigcirc$ occurs in the minimal tree, but entries marked by $\bigcirc$ are only binding and hence determining the stochastic stable convention. }\label{tab:appen-comp}
\end{table}

Thus we have the following result.
\begin{thm}
  Consider the logit choice rule. There exists $\underline \delta$ such that for all $\delta < \underline \delta$,  the stochastic stable state $m^{st}(\delta)$ converges to $s^{NB}$: i.e.,
  \[
    \delta m^{st}(\delta) \rightarrow s^{NB}
  \]
  where
  \[
    - f'(s^{NB}) = \frac{f(s^{NB})}{s^{NB}}.
  \]
\end{thm}
\begin{proof}
  Choose $\underline \delta$ satisfying Lemma \ref{lem:com}. Let $\delta < \underline \delta$. If $s^*> s^E$, then pick $m^{st}(\delta)$ to be the integer closest to $\mu^{*}(\delta)$ in \eqref{eq:def-mus}. If $s^* < s^E$, the pick $m^{st}(\delta)$ to be the integer closest to $\mu^{**}(\delta)$. Then Lemma \ref{lem:sss}, Lemma \ref{lem:s-fin} and Theorem \ref{thm:SSE} show that $m^{st}(\delta)$ is a stochastically stable state. Since $\mu^{*}(\delta), \mu^{*}(\delta)  \rightarrow s^*$, we have $\delta m^{st}(\delta) \rightarrow s^*=s^{NB}$ and obtain the desired result.
\end{proof}

\begin{thm}
  Consider the intentional logit choice rule. There exists $\underline \delta$ such that for all $\delta < \underline \delta$,  the stochastic stable state $m^{st}(\delta)$ converges to $s^{I}$: i.e.,
  \[
    \delta m^{st}(\delta) \rightarrow s^{I}
  \]
  where
  \[
    - f'(s^{I}) = (\frac{f(s^{I})}{s^{I}})^2.
  \]
\end{thm}
\begin{proof}
Under the intentional logit choice rule, we have

\[
\min_j R^I_{m j} = \min \{  \delta m \frac{ f(\delta m ) - f(\delta (m+1))}{f(\delta m)}, f(m \delta) \frac{\delta }{m \delta}  \}
\]
Then the exactly same argument as for the unintentional logit choice rule shows the desired result.
\end{proof}

\end{document}